\documentclass[11pt]{article}

\usepackage[english]{babel} 
\usepackage[T1]{fontenc} 
\usepackage{lmodern,microtype} 
\usepackage{titlesec,titling} 
\usepackage[nohead]{geometry} 
\usepackage{setspace} 
\usepackage[inline]{enumitem} 

\usepackage[comma]{natbib} 

\usepackage{authblk}

\usepackage{amsmath, amssymb, amsthm,bbm, mathtools}
\usepackage{thmtools} 
\usepackage{thm-restate}

\usepackage{algorithmic}
\usepackage[ruled]{algorithm}
\usepackage{graphicx}
\usepackage{color}
\usepackage[dvipsnames]{xcolor}
\usepackage{tabularx}
\usepackage{booktabs}

\usepackage{color-edits}
\addauthor{xq}{brown}

\usepackage{tikz}
\usetikzlibrary{patterns}
\usepackage{subcaption}
\usetikzlibrary{calc}
\usetikzlibrary{decorations.pathreplacing}
\usetikzlibrary{arrows}

\usepackage[hidelinks]{hyperref}
\usepackage[capitalize]{cleveref}

\usepackage{nicefrac}

\usepackage{multicol}

\usepackage{gensymb}

\theoremstyle{definition}

\newtheorem{definition}{Definition}
\newtheorem{assumption}{Assumption}

\newtheorem{examplex}{Example}[section]

\newtheorem{conditionalt}{Condition}

\newenvironment{conditionp}[1]{
  \renewcommand\theconditionalt{#1}
  \conditionalt
}{\endconditionalt}

\theoremstyle{remark}
\newtheorem*{remark}{Remark}

\newtheorem{claim}{Claim}

\theoremstyle{plain}
\newtheorem{theorem}{Theorem}
\newtheorem{lemma}{Lemma}
\newtheorem{proposition}{Proposition}
\newtheorem{corollary}{Corollary}

  {
   \pushQED{\qed}\begin{examplex}}
  {\popQED\end{examplex}}
\newcommand{\E}{\mathbf{E}}
\newcommand{\pr}{\mathbf{P}}

\newcommand{\bbR}{\mathbb{R}}

\newcommand{\calA}{\mathcal{A}}
\newcommand{\calB}{\mathcal{B}}
\newcommand{\calC}{\mathcal{C}}
\newcommand{\calD}{\mathcal{D}}

\newcommand{\calF}{\mathcal{F}}

\newcommand{\calJ}{\mathcal{J}}

\newcommand{\calM}{\mathcal{M}}

\newcommand{\calQ}{\mathcal{Q}}

\newcommand{\calT}{\mathcal{T}}

\newcommand{\tarc}{\mbox{\large$\frown$}}
\newcommand{\arc}[1]{\stackrel{\tarc}{#1}}

\newcommand{\indicate}[1]{{\bf 1}\left[#1\right]}

\newcommand{\nullset}{\varnothing}
\newcommand{\compl}{\textsc{c}}

\DeclareMathOperator*{\argmin}{arg\,min}

\usepackage{ upgreek }
\newcommand{\feature}{x}
\newcommand{\features}{\boldsymbol{\feature}}
\newcommand{\Featurespace}{X}
\newcommand{\orifeature}{\feature^0}
\newcommand{\orifeatures}{\features^0}

\newcommand{\firstfeatures}{\features^1}

\newcommand{\secondfeatures}{\features^2}

\newcommand{\dist}{\calF}
\newcommand{\dists}{\boldsymbol{\dist}}
\newcommand{\density}{\boldsymbol{f}}

\newcommand{\genericfeature}{z}
\newcommand{\genericfeatures}{\boldsymbol{\genericfeature}}

\newcommand{\qualregion}{H}
\newcommand{\ballcenter}{\boldsymbol{o}}

\newcommand{\weight}{w}
\newcommand{\weights}{\boldsymbol{\weight}}
\newcommand{\classifier}{h}
\newcommand{\Line}{L}

\newcommand{\strategy}{\sigma}
\newcommand{\strategies}{\boldsymbol{\strategy}}

\newcommand{\cost}{c}
\newcommand{\onecost}{C}
\newcommand{\mc}{\eta}
\newcommand{\metric}{d}
\newcommand{\util}{u}

\newcommand{\test}{t}

\newcommand{\probprincipal}{q}

\newcommand{\disclose}{D}

\newcommand{\simultaneous}{S_0}
\newcommand{\sequential}{S_1}

\newcommand{\lineartest}{\calT_{\ell}}
\newcommand{\twolineartest}{\calT_{s}}

\newcommand{\setperp}{\calJ} 

\newcommand{\setO}{\calB} 
\newcommand{\setonetwo}{\calC} 
\newcommand{\settwoone}{\calD} 

\newcommand{\manipulation}{\calM}

\newcommand{\qualified}{\calQ}

\newcommand{\symmetric}{\boldsymbol{\text{Sym}}}

\newcommand{\conv}{\textbf{conv}}

\newcommand{\pointonetwo}{E}
\newcommand{\pointtwoone}{G}

\hypersetup{
    colorlinks=true,
    linkcolor=blue,
    filecolor=magenta,      
    urlcolor=cyan,
    pdftitle={Overleaf Example},
    pdfpagemode=FullScreen,
    }
 \usepackage{authblk}   
\newcommand*{\thisdraft}{This draft: } 

\title{Multi-dimensional Test Design\thanks{Xiaoyun Qiu thanks Asher Wolinsky, Bruno Strulovici, and Wojciech Olszewski for their advice, time, and constructive suggestions. She also acknowledges Abhishek Sarkar for inspiring conversations throughout the development of the project, and Junko Oguri, Matthew Thomas, Siqi Zheng, Yoshimasa Katayama, Yuichiro Tsuji for providing domain expertise on the applications.  
We also thank  Nemanja Antic, Ian Ball, Federico Bugni, Kwok Yan Chiu, Hershdeep Chopra, Jose Higueras Corona, Eddie Dekel, David Dranove, Piotr Dworczak, Jeff Ely, Guillaume Gex, Yingni Guo, Igal Hendel, Panagiotis Kyriazis, Joshua Mollner, Kyohei Okumura, Alessandro Pavan,  Harry Pei,  Marciano Siniscalchi, Alex Smolin, Alison Zhao for helpful suggestions, comments, and discussions. All errors are our own.}
}
\author{Xiaoyun Qiu\thanks{Northwestern University. Email: \texttt{xiaoyun.qiu@u.northwestern.edu}},
Liren Shan\thanks{Toyota Technological Institute at Chicago. Email: \texttt{lirenshan@ttic.edu}}}
\date{\thisdraft \today}

\begin{document}

\maketitle
\begin{center}
\Large\href{https://drive.google.com/file/d/1lUeOxCeY-h_rutxn6S_oWLtbuwvSkTwy/view?usp=sharing}{[Click here for the latest version]}
\end{center}

\begin{abstract}
   How should one jointly design tests and the arrangement of agencies to administer these tests (testing procedure)?
  To answer this question, we analyze a model where a principal must use multiple tests to screen an agent with a multi-dimensional type, knowing that the agent can change his type at a cost.
  We identify a new tradeoff between setting difficult tests and using a difficult testing procedure. 
  We compare two settings: (1) the agent only misrepresents his type (manipulation) and (2) the agent improves his actual type (investment).
  Examples include interviews, regulations, and data classification.
  We show that in the manipulation setting,   \emph{stringent} tests combined with an easy procedure, i.e., offering tests sequentially in a fixed order, is optimal.
  In contrast, in the investment setting,  \emph{non-stringent} tests with a difficult procedure, i.e., offering tests simultaneously, is optimal;
  however, under mild conditions offering them sequentially in a random order may be as good. Our results suggest that whether the agent manipulates or invests in his type determines which arrangement of agencies is optimal.
\end{abstract}


\newpage
 \section{Introduction}
Assessing job candidates through interviews, evaluating banks' financial health, judging the legitimacy of firms' activities, and classifying data with algorithms are different forms of testing.
Testing is an imperfect means to evaluate another party's information, especially when such information is unverifiable. 
Such situations arise when the interested party can fabricate their type at a cost \citep{perez2022test,perez2024score,li2023screening}.
This type of information is termed semi-hard, lying between hard information, which cannot be fabricated \citep{green1986partially}, and soft information, which can be fabricated at no cost.\footnote{This is usually the case in classic mechanism design settings and in cheap-talk communication settings where each agent can freely report their own type.}

Within the context of testing, many real-world scenarios involve multiple principals, each with different criteria and the power to veto the participant being tested.
Examples include a hiring committee with multiple members and multiple agencies collaborating on regulations. 
  Monetary transfer is typically unavailable. 
  More realistically, the principal(s) can use two tools that are more restricted: (1) tests, which establish the standards for what is deemed acceptable, and (2) a testing procedure, which outlines how tests are conducted, either simultaneously (all tests conducted together) or sequentially (tests conducted one after another).  For example, governments regulate firms to ensure they act in the best interests of both consumers and shareholders. This regulatory task can be carried out either by a single agency evaluating both criteria simultaneously or by two separate agencies, each focusing on one criterion and evaluating them sequentially.\footnote{For example, the Japan Financial Services Agency (FSA) works in coordination with the the Japan Fair Trade Commission (JFTC) to ensure that M\&A activities are compliant with both securities and antitrust laws.}\footnote{This also includes the special case of one principal with multiple criteria. For example, data classification with multiple classifiers and bank regulation, where the regulator uses leverage and liquidity ratios to assess banks' financial health. In Europe, banks report these ratios separately, whereas in the U.S., banks report them together.}

Given this, we ask: what are the optimal tests and testing procedure to select a strategic participant meeting multiple criteria? Specifically, the choice of testing procedures can also be viewed as the arrangement of agencies responsible for administering the tests. While \citet{perez2022test} have studied the optimal test when both the screening criterion (or the participant's type) and the test are one-dimensional, an interesting trade-off concerning the choice of testing procedure arises when both the type and the tests are multi-dimensional.\footnote{The problem is trivial when only one side (either the type or the tests) is multi-dimensional. See \cref{sec:discussion}.}

While manipulation, either by costly lying or by costly fabricating the truth, is common in many of these scenarios,\footnote{Examples include job applicants tailoring their profiles to stand out in interviews, firms misrepresenting their impacts to bypass regulatory scrutiny, or banks involving in `window dressing' by temporarily adjusting their financial statements right before the reporting dates to appear safer than they actually are. See \citet{perez2022test,perez2024score, hardt2016strategic} for more examples.} it is also common for participants to invest in improving their actual types. For instance, job candidates seeking promotion acquire new skills, while firms and banks may restructure their business models to meet regulatory requirements.
Therefore, we compare the optimal policy under two different technologies: (1) the participant misrepresents themselves (\emph{manipulation}), or (2) the participant improves their actual type (\emph{investment}).\footnote{In \cref{sec:discussion}, we consider allowing the participant to choose between manipulation and investment. We argue that it does not add to the analysis. }




In practice, we observe that  the arrangement of agencies is different in different settings. We argue that the critical difference is whether participants manipulate or invest in their type. To show this, we build on a stylized principal-agent framework introduced by \citet{zigzag} and conduct a comprehensive mechanism design analysis. 
The principal (hereafter she) only wants to select an agent (he) whose true (two-dimensional) type meets her criteria.
The principal uses two tests to judge the agent's true type, together with a testing procedure, which could be simultaneous, sequential with a fixed order or sequential with a random order combined with a disclosure decision.\footnote{\citet{zigzag} did not study all possible sequential mechanisms. In contrast, we extend the analysis to the entire class of sequential mechanisms (\cref{sec: distance cost}).}
The principal commits to select the agent if and only if he passes both tests.
The agent seeks to be selected and could change his type at a type-dependent cost before each test either by manipulation or by investment. When the type is changed multiple times, the cost adds up.

Within this framework, we first look for the optimal sequential mechanism.
Consider the manipulation setting.
Intuitively, to exclude unqualified agent, the principal could use either tests stricter than her true criteria (\emph{stringent} tests) or a testing procedure that makes manipulation more costly (a \emph{stringent} procedure), given that every agent can pay a type-dependent cost to inflate his type.
For example, a procedure with fixed order of tests is less stringent than one that randomizes over the sequence of tests.\footnote{In the evidence literature \citep[e.g.,][]{glazer2004optimal, sher2014persuasion}, randomization takes a different form and usually refers to which attribute to check because the principal can check only one attribute.}
The reason is as follows.
In any sequential mechanism, the agent can choose a strategy to pass one test (potentially failing another) at a time. 
While the way the agent manipulates his type to pass one test is quite different from that to pass another test, such a strategy results in a lower chance of selection when  the order of tests is sometimes secretly reversed.
Hence, the agent may switch to passing both tests together under randomization, which is more costly but guarantees selection.
Similarly, not disclosing the first (realized) test after it takes place makes a random-order procedure even more stringent because the agent has to accomplish the second test even after he fails the first one.
While \citet{perez2022test,zigzag} show that stringent tests can enhance outcomes by shutting down the channel of testing procedures, it is unclear how tests and a testing procedure interact and which combination is optimal.
 
We show that a fixed-order sequential mechanism with stringent tests is optimal in the manipulation setting (\cref{thm: optimal max qualified}).
This is because any testing procedure has to be combined with stringent tests to exclude any unqualified agent.
However, stringent tests also exclude some qualified agents, which needs to be compensated by a less stringent procedure.
This is a simplified intuition because it separates the choice of tests and testing procedures.
In general, the set of tests available under random-order procedures is larger than that under fixed-order procedures.  
 In the formal analysis, we offer a new constructive argument to show that, given any feasible random-order mechanism with arbitrary tests, there exists a feasible fixed-order mechanism that dominates it.
 To show dominance, we introduce a feasible mechanism that clearly dominates the random-order mechanism using a set inclusion argument and is dominated by a fixed-order mechanism using a probability argument.
 The key challenge comes from constructing such a feasible mechanism, which will be discussed in details in \cref{sec: distance cost}.
  We further show that this result generalizes to  \emph{perfect tests} (tests that output the agent's attributes; \cref{thm: sequential perfect tests}), and that when cheap-talk communication is allowed upfront, a fixed-order sequential mechanism can achieve the first best (\cref{thm: optimal max qualified cheap talk}).

 In contrast,  we show that under mild conditions, a random-order mechanism using  tests that coincide with the true criteria (\emph{non-stringent} tests) can achieve the first best in the investment setting (\cref{thm: true effort sequential}).
As discussed earlier, the agent may switch to passing both tests together facing randomization.\footnote{Whether to disclose the first test to the agent also affects his preference over different strategies. Detailed discussions are in \cref{sec: distance cost}.}
When this happens and when \emph{non-stringent} tests are used, the agent is selected if and only if he invests to some type satisfying the true criteria.
We show that this result generalizes to other cost functions (\cref{thm:optimal sequential metric investement}).

The second question we ask is whether the simultaneous procedure can achieve a better outcome than any sequential one.
Consider bank regulations where the central bank sets minimum requirements on liquidity and leverage ratios.
Suppose the liquidity ratio is tested before the leverage ratio.
A bank can temporarily acquire more liquid assets before the first test to increase the liquidity ratio and after passing the first test, it can get rid of extra liquid assets to boost its leverage ratio for the second test.
However, if the two tests are conducted simultaneously, such a manipulation strategy would not be available.
This suggests that a simultaneous procedure is more stringent than a sequential one.
\citet{zigzag} argue that in the manipulation setting, there is a fixed-order sequential mechanism better than any simultaneous mechanism. We give a more precise proof of this claim (\cref{lem:fix-simultaneous-distance cost}).
On the other hand, we prove that in the investment setting, a simultaneous procedure combined with non-stringent tests can achieve the first best (\cref{thm:optimal investment}).
Since the agent is forced to pass both tests together in any simultaneous mechanism, for the same reason discussed earlier, everyone selected becomes qualified under non-stringent tests.
We also show that these results generalize to general cost functions (\cref{thm:opt_manipulation}), \emph{perfect tests} and convex qualified region (\cref{thm: perfect tests arbitrary qualified region}), and an alternative design objective (\cref{prop:Pii manipulation} and \ref{prop:Pii investment}).

To summarize, we study the optimal tests and testing procedure involving multiple criteria. 
We conduct a non-traditional mechanism design analysis in the context of testing. 
Unlike the classic framework, we study how the stringency of the tests interact with the testing procedure. 
Our theoretical contribution is a new constructive argument to tackle this problem.

 \section{Model}
A principal (hereafter she) is going to test whether an agent (hereafter he)
has desired attributes. Initially, the agent's attributes (or type) are $\orifeatures=(\orifeature_{A},\orifeature_{B})\in \Featurespace = 
\mathbb{R}
^{2}$. The attributes are privately known by the agent. The principal only
knows their distribution: $\orifeatures\sim \dists$. The agent can change his attributes to $\features=(\feature_{A},\feature_{B})$
at a cost $\onecost(\orifeatures,\features)$. In large parts of the paper, we will assume this cost
is the Euclidean distance between $\orifeatures$ and $\features$. We model such a change in two ways: In the \textit{investment}
setting,  $\features$
become the agent's new \textit{true} attributes.  In the \textit{manipulation}
setting,  the agent's true attributes stay unaltered at $\orifeatures$ and $\features$ are fake attributes, acquired only for the purpose of testing. 

The principal is interested in learning  the agent's true attributes belong to
a set $\qualregion\subset 
\mathbb{R}
^{2}$. If they do, we will say that the agent is \textit{qualified}. For
most of the analysis,  we assume that the principal has two criteria $\classifier_A$ and $\classifier_B$.
Each $\classifier_i$ is a linear constraint on the agent's attributes, which can be represented by a half plane.
Hence, the qualified region is $\qualregion=\classifier_{A}\cap \classifier_{B}$. 
The half planes are marked in red and blue in \cref{fig: hiring}. 
Let $\theta $ stand for the angle between the lines that
bound the two half planes. 

We will assume that the principal can choose two tests $\widetilde{h}_{A}$ and $\widetilde{h}_{B}$. 
Each test $\widetilde{h}_{i}$ is a linear constraint (half plane) that specifies which attributes pass the test. 
Define $\widetilde{\qualregion}=\widetilde{h}_{A}\cap \widetilde{h}_{B}$%
, and we will study two kinds of testing mechanisms:

\paragraph{Simultaneous mechanisms} determine whether $\features\in \widetilde{H}$.
 Recall that this $\features$ are true attributes in the investment  setting; in the manipulation setting, $\features$ are fake attributes.

\paragraph{Sequential mechanisms} are somewhat richer and more involved. Let $(M,s)$ denote a sequential mechanism, where $M$ is a set (of messages) and $s: M\rightarrow\sequential$ is the set of all possible sequential testing procedures. 
Suppose first $M=\nullset$, meaning no upfront communication is allowed.
In
each such a mechanism, the principal selects whether test $\widetilde{h}_{A}$
or test $\widetilde{h}_{B}$ will be performed first. This choice can be
random. Let the triple $(\widetilde{h}_{A},\widetilde{h}_{B},q)$ denote the
mechanism that applies test $\widetilde{h}_{A}$ as first with probability $q$%
, and test $\widetilde{h}_{B}$ as first with the complementary probability.\footnote{Note that $(\widetilde{h}_{A},\widetilde{h}_{B},q)=(\widetilde{h}_{B},%
\widetilde{h}_{A},1-q)$.}
 Test $\widetilde{h}_{i}$, $i=A,B$, determines
whether the agent's attributes at the moment of taking the test belong to the
half plane $\widetilde{h}_{i}$. 
If $q\neq 0,1$, the agent does not know which test will be applied first.
After the first test takes place, the principal can decide whether to disclose this information.
Let $\widetilde{h}_{1}$ denote the first realized test.
Let $s= (\widetilde{h}_{A},\widetilde{h}_{B},q,\disclose)$ denote the sequential mechanism with disclosure decision after the first test $\disclose\in \{\widetilde{h}_{1},\nullset\}$, where
 $\disclose=\widetilde{h}_{1}$ ($\disclose=\nullset$) means the mechanism does (not) disclose the first test after it takes place. 

We allow the agent to change her attributes twice: (1) from $\orifeatures$ to $\firstfeatures$
before taking the first test, and (2) from $\firstfeatures$ to $\secondfeatures$ before taking
the second test. Again, $\firstfeatures$ and $\secondfeatures$ are true attributes in the
investment setting, but are only fake attributes in the manipulation setting.
If $q\neq 0,1$, the agent  chooses $\firstfeatures$ when he knows that the first test will be $\widetilde{h}_{A}$
with probability $q$. 
In the case of disclosure $\disclose=\widetilde{h}_{1}$, he knows which test will be applied
second when he chooses $\secondfeatures$. 
In the case of no disclosure $\disclose=\nullset$, he only knows  test  $\widetilde{h}_{B}$ will be applied
second with probability $q$ when he chooses $\secondfeatures$.
The cost of changing her attributes from $\orifeatures
$ to $\firstfeatures$ and then from $\firstfeatures$ to $\secondfeatures$ is denoted by $%
\cost(\orifeatures,\firstfeatures,\secondfeatures)$. In large parts of the paper, this cost will be
assumed to be additive $\cost(\orifeatures,\firstfeatures,\secondfeatures)= \onecost(\orifeatures,\firstfeatures)+\onecost(\firstfeatures,\secondfeatures)$.

The agent just wants to pass the tests.
His utility is 
\[
\indicate{\features\in \widetilde{\qualregion}}-\onecost(\orifeatures,\features)
\]%
for simultaneous mechanisms, and 
        \[
\E\left\{\indicate{\firstfeatures\in \widetilde h_{1}\cap \secondfeatures\in \widetilde  h_{2}}\right\}
-\cost(\orifeatures,\firstfeatures,\secondfeatures)
\]%
for sequential mechanisms. Here, $\indicate{\cdot}$ denotes the
indicator function of the set $[ \cdot ]$, $ \widetilde h_{1}$ and $ \widetilde h_{2}$ denote the
tests that are applied first and second, respectively, and the expectation
refers to the randomness in choosing $ \widetilde h_{1}$ and $ \widetilde h_{2}$.

Suppose $M\neq\nullset$, meaning there is a cheap-talk communication before the principal chooses the sequential procedure. A direct mechanism is one where $M=\Featurespace$ and a menu of sequential procedures $s=<\widetilde{h}_{A}(m),\widetilde{h}_{B}(m),q(m),\disclose(m); m\in M>$. 

 For every menu of sequential procedures $s=<\widetilde{h}_{A}(m),\widetilde{h}_{B}(m),q(m),\disclose(m); m\in M>$ and every message $\features$, define $s(\features)$ to be the sequential procedure (mechanism) chosen by the principal after seeing message $\features$.
 We assume that given a mechanism $(M,s)$, an agent of type $\features$ will choose a message such that  his utility under mechanism $s(\features)$ is maximized for all $\features\in M$.
 Moreover, once the mechanism $s(\features)$ is announced, the agent chooses a strategy to change his attributes to maximize his utility under the mechanism $s(\features)$.
 
We will explore various objectives of the principal. However, the most
interesting results will be obtained when the principal's utility is the
probability that a qualified agent (the one whose true attributes are in $\qualregion$)
passes the tests subject to the constraint that unqualified agent never
passes the tests, i.e.,
\begin{equation}\tag{$\mathcal{P}_{I}$}\label{max qualified}
    \begin{aligned}
        \max & \quad \pr[\text{selecting qualified agent}] \\
        s.t. & \quad \text{ no unqualified agent is selected};\\
    \end{aligned}
\end{equation}
This case will be studied in most  sections. The objective represents lexicographic preferences:  avoiding any unqualified agent is a primary objective, since the error of commission is much more costly than the error of omission in many economic settings, including regulation and hiring.

We will now recap the timeline of the game induced by any mechanism:

1. The agent sends a message $m\in M$ to the principal when $M\neq\nullset$.

2. Principal announces and commits to a mechanism (simultaneous or
sequential).

3. In a simultaneous mechanism, the agent chooses $\features$; A passes the test if $\features\in \widetilde{\qualregion}=\widetilde{h}_{A}\cap \widetilde{h}_{B}$, and the game ends.

4. In a sequential mechanism, 

(a) the agent chooses $\firstfeatures$;

(b) a random device chooses the order of the tests, and test $\widetilde h_{1}$ is
applied; 

(c) if $\disclose=\widetilde{h}_{1}$, the first test is announced to the agent;

(d) the agent chooses $\secondfeatures$, and test $\widetilde h_{2}$ is applied.

Finally, A passes the tests if $\firstfeatures\in \widetilde h_{1}$ and $\secondfeatures\in \widetilde h_{2}$, and the game ends.\footnote{%
Of course, if $\firstfeatures\notin \widetilde h_{1}$ and the mechanism discloses the first test, the agent fails the first test, and the
game could end at this point.}

\begin{remark}
    In any simultaneous mechanism, the agent can change his attributes once.
We call this a  \emph{one-step strategy}.
To ease notations, we assume that there is no cost of maintaining any attributes.

\begin{assumption}\label{assump: cost function one step}
	For any $\orifeatures,\firstfeatures$, $\onecost(\orifeatures,\firstfeatures)=\cost(\orifeatures,\firstfeatures,\firstfeatures)=\cost(\orifeatures,\orifeatures,\firstfeatures).$
 
\end{assumption}

\end{remark}


 \section{An example and preliminary analysis}\label{subsec:example}
In this section, we provide a stylized example to motivate our model. We discuss other applications in \cref{sec:applications}. Suppose a biotech company X is seeking candidates for a senior scientist
position. Senior scientists must balance technical expertise and domain
knowledge with managerial skills and leadership. Suppose X is looking for
candidates that have some management experience and especially strong
scientific expertise. Suppose, not unrealistically, that excessive
experience as a manager may be an indication of career goal as a manager
rather than a scientist, and even be a signal of declining technical
expertise.

We use \cref{fig: hiring} to illustrate X's requirements. The horizontal axis 
represents managerial experience and the vertical axis  represents
scientific experience. The blue vertical line represents the minimum
requirement on managerial experience, say managing one direct report. The
red diagonal line represents the joint requirement on both aspects. The purple dashed region represents the set of
attributes that X's considers as qualified for senior scientist position.

\begin{figure}[h] 
\centering 
\begin{tikzpicture}[xscale=7.8,yscale=7.8]
\draw [thick,->] (0.6,0.8) -- (1.5,0.8);
\draw [thick, ->] (0.9,0.66) -- (0.9,1.33);
\node [below] at (1.5,0.8) {$\feature_A=$managerial experience};
\node [above] at (0.9,1.33) {$\feature_B=$scientific experience};

\draw [domain=0.66:1.4, thick, red] plot (\x, {3/4*\x+1/4});
\draw [thick, blue] (1,0.66) -- (1,5.2/4);
\node [left] at (1.2,1.15) {$+$};
\node [right] at (1, 1.15) {$+$};
\node [left,font=\tiny] at (0.9, 0.78 ) {\footnotesize$O$};
\path[pattern color=purple,pattern=north west lines] (1,5.2/4) -- (1,1)-- (1.4,5.2/4); 

\end{tikzpicture}
\caption{Hiring requirements for a senior scientist position}
\label{fig: hiring}
  \end{figure}
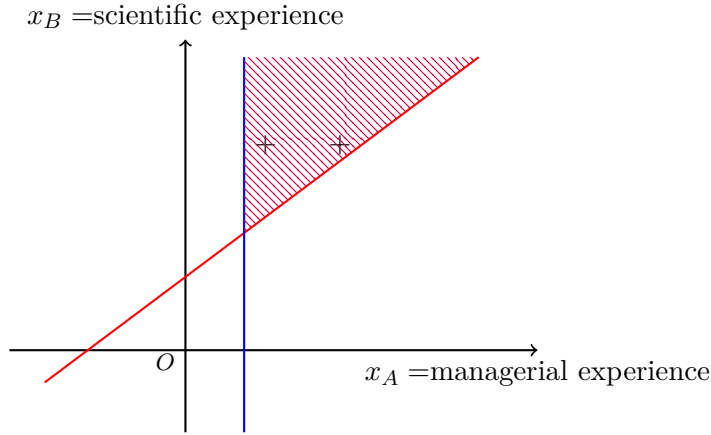
  
 Past experience in other biotech companies can reflect a candidate's
 suitability for the position. The resume provides some information, but does not fully reveal which of the two functions (manager or scientist) was the candidate's main duty. So, the company finds a verbal interview helpful to assessing the candidates' qualifications. 
Consider two types of hiring
procedures: one consists of resume screening, followed by an interview of
each candidate who was not rejected in reviewing resumes (sequential
procedure), and a joint procedure in which a committee reviews resumes and
interviews candidates at the same time (simultaneous procedure).

Candidates can obviously tailor their resumes and answers for common
interview questions to influence the interviewers' impression of their
profiles. So, we analyze this example within our manipulation setting.
Suppose that the cost of tailoring is $\onecost(\orifeatures,\features)$, where $%
\orifeatures=(x_{A}^{0},x_{B}^{0})$ are the true attributes and $\features=(\feature_{A},\feature_{B})$ are
the attributes tailored for the purpose of hiring, is equal to the Euclidean
distance between $\orifeatures$ and $\features$. Under the sequential procedure, if a
candidate chooses to tailor his resume, for example, to emphasize managerial
experience, then he is constrained during the interview. He cannot present
himself in a way that contradicts his resume. Hence, the total cost of
manipulating is $\onecost(\orifeatures,\firstfeatures)+\onecost(\firstfeatures,\secondfeatures)$, where $\firstfeatures$ stands for
attributes presented in the resume and $\secondfeatures$ stands the attributes presented
in the interview. The candidate's payoff is the probability of being hired
minus the cost of tailoring, be it one-step or two-step.

Suppose that as in our model, the company's preferences are lexicographic:
avoiding unqualified applicants is a primary objective, perhaps because
hiring an unqualified candidate is very costly. For the sake of execution,
the company must set some specific cutoffs in testing in practice. The
question the company faces is whether to use a simultaneous or a sequential
procedure (i.e., mechanism).\bigskip 

\paragraph{Preliminary analysis}
Next we provide a reflection argument, which is the key to characterize the set of attributes that are selected under any fixed-order mechanism. Suppose now the principal uses a sequential mechanism that (1) offers test $h_{A}$ at the first stage, and (2) among those
selected in (1), those whose attributes belong to $h_{B}$ will be eventually
selected.

Notice that any agent with true attributes $\orifeatures\in h_{A}\cap h_{B}$ will
not manipulate and will be selected.  Consider
any agent with unqualified true attributes $\orifeatures\notin h_{A}\cap h_{B}$.
He has now two options of getting selected: (1) to adopt attributes that are in
the qualified region, and (2) to adopt attributes $\firstfeatures\in h_{A}$ for the
first test, and after passing the first test, to adopt yet other attributes $%
\secondfeatures\in h_{B}$. To illustrate the two types of strategies, consider an agent with attributes $\orifeatures$ in \cref{fig: zig zag}. The least costly one-step
strategy that guarantees selection is to move to the  point $O$, at which
the boundaries of $h_{A}$ and $h_{B}$ intersect. The utility of adopting
this one-step strategy is $1-\onecost(\orifeatures,O)$.

\begin{figure}[t]
\centering
\begin{tikzpicture}[xscale=10,yscale=10]

\draw [domain=0.66:1.36, thick] plot (\x, {3/4*\x+1/4});
\node [left] at (1.32, 1.25 ) {$\classifier_B$};
\draw [thick] (1,0.66) -- (1,1.25);
\node [right] at (1, 1.25 ) {$\classifier_A$};
\node [left] at (1.27,1.2) {$+$};
\node [right] at (1, 1.2) {$+$};

\draw [domain=1+0.25*0.6:1.45, loosely dashed] plot (\x, {3/4*\x+1/4-5/16}); 
\draw [red,ultra thick] (1,1) -- (1+0.25*0.6,1-0.25*0.8);
\node [right] at (1+0.25*0.6,1-0.25*0.8) {\footnotesize$B$};
\node [left] at (1-0.25*0.6,1-0.25*0.8) {\footnotesize$B'$};

\draw [loosely dashed] (1-0.25,1) -- (1-0.25,1.25);

\draw [loosely dashed] (1-0.25,1) -- (1,1);
\node [left] at (1-0.25, 1 ) {\footnotesize$A$};

\draw[blue,ultra thick] (1-0.25,1) arc (180:233.1:0.25);
\draw[blue,ultra thick] (1,1) -- ++(180:0.25);
\draw[blue,ultra thick] (1,1) -- ++(233.1:0.25) ;

\draw [red,ultra thick] (1,1) -- (1-0.25*0.6,1-0.25*0.8);
\draw [ red,ultra thick] (1,11/16) -- (1-0.25*0.6,1-0.25*0.8);
\node [left] at (0.98,11/16) {\footnotesize$C$};
\draw [ red,ultra thick] (1,11/16) -- (1+0.25*0.6,1-0.25*0.8);

\draw [dotted, ultra thick] (1-0.1,1-0.2) -- (1+0.1,1-0.2);
\node [below] at (1-0.1, 1-0.2 ) {\footnotesize$\orifeatures$};
\draw [dotted,ultra thick]  (1+0.1,1-0.2) -- (1+0.1-5.45/24*0.6,1-0.2+5.45/24*0.8);
\node [below] at (1+0.1,1-0.2) {\footnotesize$\tilde\features$};
\node [left] at  (1+0.1-5.5/24*0.6,1-0.2+5.5/24*0.8) {\footnotesize$\secondfeatures$};
\draw [dotted,ultra thick] (1-0.1,1-0.2) -- (1,1-0.072);
\node [right] at (1, 1-0.072 ) {\footnotesize$\firstfeatures$};

\node [above] at (1.016, 1 ) {\footnotesize$O$};
\end{tikzpicture}
\caption{Two-step strategy: $\orifeatures\rightarrow \firstfeatures\rightarrow \secondfeatures$}
\label{fig: zig zag}
\end{figure}

What is the least costly two-step strategy such that first attributes $%
\firstfeatures\in h_{A}$ and next attributes $\secondfeatures\in h_{B}$ are attained? This question has been solved by \citet{zigzag} under the same additive Euclidean cost function.  Here we provide a geometric intuition. Observe that
the two chosen attributes $\firstfeatures$ and $\secondfeatures$ must belong to the boundary of $%
h_{A}$ and $h_{B}$, respectively. Reflect point $\orifeatures$ over line $h_{A}$.
This gives point $\widetilde{\features}$ in \cref{fig: zig zag}. Under the additive costs $%
\cost(\orifeatures,\firstfeatures,\secondfeatures)=\onecost(\orifeatures,\firstfeatures)+\onecost(\firstfeatures,\secondfeatures)$, where $\onecost(\features^{k},x^{l})$
is equal to the Euclidean distance between $\features^{k}$ and $\features^{l}$, 
\[
\cost(\orifeatures,\firstfeatures,\secondfeatures)=\cost(\widetilde{\features},\firstfeatures,\secondfeatures)\text{,}
\]%
and by the triangle inequality, $\cost(\widetilde{\features},\firstfeatures,\secondfeatures)$ is the
shortest when $\widetilde{\features}$, $\firstfeatures$ and $\secondfeatures$ are co-linear.\footnote{This is also equivalent to the following shortest path problem in optics.
Suppose there is a light emitting from point $\orifeatures$.
The light has to first reach mirror $A$ (line $A$) before it reaches mirror $B$ (line $A$). We find the shortest path in a similar way.}

Observe that 
there exists some types such that the least costly two-step strategy of the agent is always cheaper than the
least costly one-step strategy (for example, any type in $OBCB'$ in \cref{fig: zig zag}). As we have already noticed, for these types, the cost of the
former strategy is $\cost(\widetilde{\features},\firstfeatures,\secondfeatures)$ when $\widetilde{\features}$, $%
\firstfeatures$ and $\secondfeatures$ are co-linear. The cost of the latter strategy is $%
\onecost(\orifeatures,O)=\onecost(\widetilde{\features},O)$. And since points $\widetilde{\features}$, $\secondfeatures$
and $O$ form a triangle, $\onecost(\widetilde{\features},O)>\cost(\widetilde{\features},\firstfeatures,\secondfeatures)$
by triangle inequality, unless $\secondfeatures=O$ (see points $B$ and $B^{\prime }$
in \cref{fig: zig zag}), in which case $\onecost(\widetilde{\features},O)=\cost(\widetilde{\features},\firstfeatures,\secondfeatures)$%
.

Despite its equivalence to the algebraic argument provided by \citet{zigzag}, we  highlight that this geometric intuition is very convenient and powerful when characterizing the agent's best response in any sequential mechanism, especially those with randomization.

 \section{Manipulation}\label{sec: distance cost}

In this and next section, we consider the cost function $\cost$ induced by the Euclidean distance. For any two attributes $\firstfeatures, \secondfeatures$, we have $\onecost(\firstfeatures,\secondfeatures) = \eta\cdot \|\firstfeatures - \secondfeatures\|_2$, where $\eta>0$ and $\| \cdot \|_2$ is the $\ell^2$ norm. It is without loss of generality to assume that $\eta=1$, which implies that $\onecost(\firstfeatures,\secondfeatures)$ is just the Euclidean distance between $\firstfeatures$ and $\secondfeatures$. We also assume the cost is additive, which means for a path with three attributes $\orifeatures, \firstfeatures, \secondfeatures$, we have 
$\cost(\orifeatures,\firstfeatures,\secondfeatures) = \onecost(\orifeatures,\firstfeatures)+\onecost(\firstfeatures,\secondfeatures) $.

In this section, we focus on the manipulation setting.

\subsection{Optimal sequential mechanism}\label{subsec:seq manipulation}
First consider the scenario where the principal has to  conduct one test at a time.
That is, she is constrained to use sequential mechanisms.
Suppose also that tests cannot be tailored for each agent, i.e., cheap-talk communication is not possible.
Our first theorem answers the question: What is the optimal sequential mechanism? 

\begin{theorem}\label{thm: optimal max qualified}
Suppose $M=\nullset$. 
   For any distribution $\dist$, the optimal sequential mechanism uses
    \begin{enumerate}
        \item  \emph{stringent} tests $\tilde \classifier_A$ and $\tilde \classifier_B$, such that $\tilde\classifier_A \cap \tilde \classifier_B \subset \classifier_A \cap \classifier_B$;
        \item a fixed-order procedure.
    \end{enumerate}
\end{theorem}

We call a test that is stricter than the true criterion a stringent test.
    When the agent's type is one-dimensional, there is no room for the design of testing procedure, and the notion of stringent test is a raised threshold \citep{perez2022test}.  When the agent's type is two-dimensional, the first property in \cref{thm: optimal max qualified} points out that the correct notion of stringent tests is $\tilde\classifier_A \cap \tilde \classifier_B \subset \classifier_A \cap \classifier_B$.\footnote{One might conjecture that stringent tests are two shifted parallel tests, i.e.,  $\tilde\classifier_i  \subset \classifier_i$ for $i\in \{A,B\}$. However, \cref{fig:cheap-talk} provides an example on why shifted parallel tests might not be always optimal. Depending on the distribution, the gain from using non-parallel tests could sometimes offset its loss in terms of selecting qualified agents.}
    The second property in \cref{thm: optimal max qualified} is our key innovation. It highlights that when using multiple tests to screen an agent with multi-dimensional  types, the choice of testing procedure is non-trivial: offering the tests sequentially in a fixed-order is (sometimes strictly) better than other procedures.

    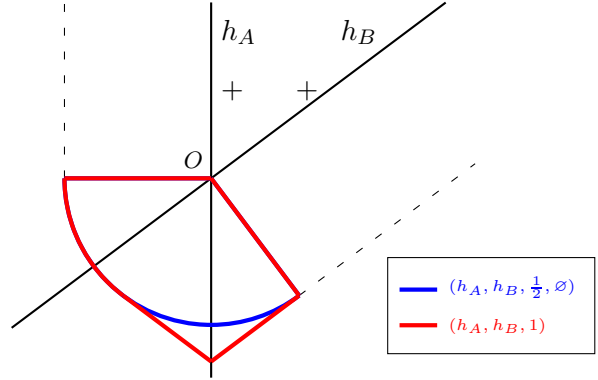
\begin{figure}
    \centering
\begin{tikzpicture}[xscale=7.8,yscale=7.8,
    pics/legend entry/.style={code={%
        \draw[pic actions] 
        (-0.25,0.25) -- (0.25,0.25);}}]

\draw [domain=0.66:1.4, thick] plot (\x, {3/4*\x+1/4});
\node [left] at (1.3, 1.25 ) {$\classifier_B$};
\draw [thick] (1,0.66) -- (1,1.3);
\node [right] at (1, 1.25) {$\classifier_A$};
\node [above,font=\tiny] at (0.97, 1 ) {\footnotesize$O$};
\node [left] at (1.2,1.15) {$+$};
\node [right] at (1, 1.15) {$+$};

\draw [domain=1+0.25*0.6:1.45, loosely dashed] plot (\x, {3/4*\x+1/4-5/16}); 

\draw [loosely dashed] (1-0.25,1) -- (1-0.25,1.3);

\draw[blue, ultra thick] (1-0.25,1) arc (180:360-53.1:0.25);
\draw[blue, ultra thick] (1,1) -- ++(180:0.25);
 \draw[blue, ultra thick] (1,1) -- ++(360-53.1:0.25) ;

\draw[red, ultra thick] (1-0.25,1) arc (180:233.1:0.25);
\draw[red, ultra thick] (1,1) -- ++(180:0.25);

\draw [ red, ultra thick] (1,11/16) -- (1-0.25*0.6,1-0.25*0.8);
\draw [red, ultra thick] (1,11/16) -- (1+0.25*0.6,1-0.25*0.8);
\draw [red, ultra thick] (1,1) -- (1+0.25*0.6,1-0.25*0.8);

\matrix [draw, above right] at (1.3,0.7) {
\pic[blue, ultra thick]{legend entry}; &  \node[blue,font=\tiny] {$(\classifier_A,\classifier_B,\frac12,\nullset)$}; \\
  \pic[red, ultra thick]{legend entry}; &  \node[red,font=\tiny] {$(\classifier_A,\classifier_B,1)$}; \\
};
\end{tikzpicture}
\caption{Stringency of procedures: random-order without disclosure vs fixed-order} 
 \rule{0in}{1.2em}$^\dag$\scriptsize In this graph, the two procedures share the same tests $\classifier_A$ and $\classifier_B$. Under the same tests, the set in blue and the set in red highlights the critical difference between the set of types that have profitable strategies to pass both tests under a random-order procedure without disclosure $(\classifier_A,\classifier_B,\frac12,\nullset)$ and that under a fixed-order procedure  $(\classifier_A,\classifier_B,1)$.
\label{fig: random uninformed vs fixed}
\end{figure}

    We call a testing procedure X is more stringent than a testing procedure Y if under the same choice of tests,  the set of types that have profitable strategies to pass both tests is larger (or has higher probability) than that under X than that under Y.  As an example, \cref{fig: random uninformed vs fixed} illustrates that a random-order procedure without disclosure is more stringent than a fixed-order procedure because of the set inclusion relation. This is true in general, regardless of the randomization probability.
    As another example, the left panel of \cref{fig: same graph ran_informed fixed 2} compares a fixed-order procedure and a random-order procedure with disclosure. Although there is no direct set inclusion relation here, we prove that the probability of the set of types that can profitably pass both tests under a random-order procedure with disclosure, i.e., the cost of passing both tests is less than one, is smaller than that under a fixed-order procedure, regardless of the randomization probability.
    Hence, a random-order procedure with disclosure is also more stringent than a fixed-order procedure.
    As a third example, we argue that a random-order procedure without disclosure is more stringent than a random-order procedure with disclosure, for the same randomization probability. Consider  point C  in the left panel of \cref{fig: same graph ran_informed fixed 2}. Since this type already passes test $\classifier_A$, he could just wait in the first round and if he is told that the first test was $\classifier_A$, he could pass the second test with a cost $1$. However, this strategy is not profitable if he does not know the first test.
    So without disclosure, this type does not have a profitable strategy to get selected since pass both tests together has a cost higher than $1$.
    
    Given that the principal can use either stringent tests or stringent testing procedure to screen the agent, the critical question is which combination is optimal. \cref{thm: optimal max qualified} states that the least stringent testing procedure combined with stringent tests are optimal.

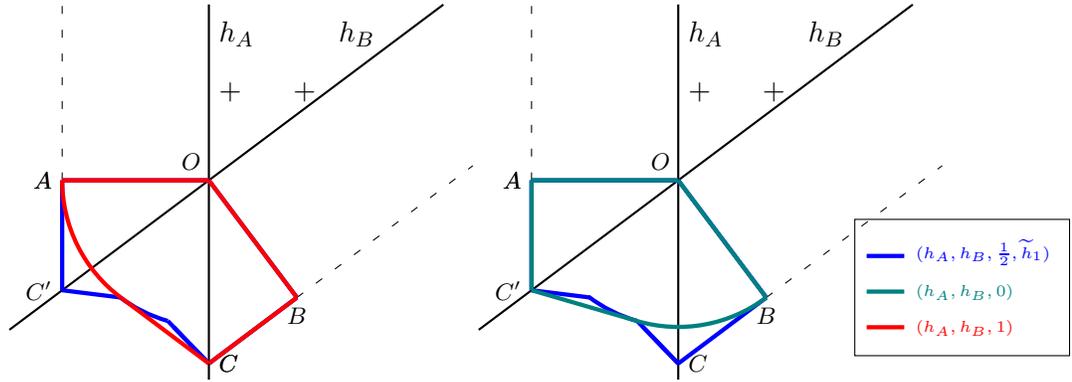
\begin{figure}
    \centering
    \begin{subfigure}[b]{0.45\linewidth}
\begin{tikzpicture}[xscale=7.8,yscale=7.8,
    pics/legend entry/.style={code={%
        \draw[pic actions] 
        (-0.25,0.25) -- (0.25,0.25);}}]

\draw [domain=0.66:1.4, thick] plot (\x, {3/4*\x+1/4});
\node [left] at (1.3, 1.25 ) {$\classifier_B$};
\draw [thick] (1,0.66) -- (1,1.3);
\node [right] at (1, 1.25) {$\classifier_A$};
\node [above,font=\tiny] at (0.97, 1 ) {\footnotesize$O$};
\node [left] at (1.2,1.15) {$+$};
\node [right] at (1, 1.15) {$+$};

\draw [domain=1+0.25*0.6:1.45, loosely dashed] plot (\x, {3/4*\x+1/4-5/16}); 
\draw [loosely dashed] (1,1) -- (1+0.25*0.6,1-0.25*0.8);

\draw [loosely dashed] (1-0.25,1) -- (1-0.25,1.3);
\draw [loosely dashed] (1-0.25,1) -- (1,1);

\draw [blue, ultra thick] (1-0.25,1) -- (1-0.25,{3/4*0.75+1/4});
\node [left, font=\tiny] at (1-0.25, 1 ) {\footnotesize$A$};

\node [left, font=\tiny] at (1-0.25, {3/4*0.75+1/4} ) {\footnotesize$C'$};

\draw[blue, ultra thick]  (0.93 ,{-1/sqrt(1/0.72^2-1)*(0.93-1)+11/16}) arc (254.1:232.8:0.25);

\draw[blue, ultra thick] (1,1) -- ++(180:0.25);
\draw[ blue, ultra thick] (1,1) -- ++(360-53.1:0.25) ;

\draw [blue, ultra thick] (1,11/16) -- (0.94,{-1/sqrt(1/0.72^2-1)*(0.94-1)+11/16});
\node [right] at (1,11/16) {\footnotesize$C$};
\draw [blue, ultra thick] (1,11/16) -- (1+0.25*0.6,1-0.25*0.8);

\draw [domain=0.93:1, blue, ultra thick] plot (\x, {-1/sqrt(1/0.72^2-1)*(\x-1)+11/16}); 

 \draw [domain=1-0.25:0.85, blue, ultra thick] plot (\x, {tan(-7.075621914)*(\x-0.75)+13/16});

\draw[red, ultra thick] (1-0.25,1) arc (180:233.1:0.25);
\draw[red, ultra thick] (1,1) -- ++(180:0.25);
\draw [ red, ultra thick] (1,11/16) -- (1-0.25*0.6,1-0.25*0.8);
\draw [red, ultra thick] (1,11/16) -- (1+0.25*0.6,1-0.25*0.8);
\draw [red, ultra thick] (1,1) -- (1+0.25*0.6,1-0.25*0.8);

\node [left, font=\tiny] at (1-0.25, 1 ) {\footnotesize$A$};
\node [below, font=\tiny] at (1+0.25*0.6,1-0.25*0.8) {\footnotesize$B$};
\node [right, font=\tiny] at (1,11/16) {\footnotesize$C$};

\end{tikzpicture}
 \end{subfigure}
\begin{subfigure}[b]{0.45\linewidth}
    \centering
\begin{tikzpicture}[xscale=7.8,yscale=7.8,
    pics/legend entry/.style={code={%
        \draw[pic actions] 
        (-0.25,0.25) -- (0.25,0.25);}}]

\draw [domain=0.66:1.4, thick] plot (\x, {3/4*\x+1/4});
\node [left] at (1.3, 1.25 ) {$\classifier_B$};
\draw [thick] (1,0.66) -- (1,1.3);
\node [right] at (1, 1.25) {$\classifier_A$};
\node [above,font=\tiny] at (0.97, 1 ) {\footnotesize$O$};
\node [left] at (1.2,1.15) {$+$};
\node [right] at (1, 1.15) {$+$};

\draw [domain=1+0.25*0.6:1.45, loosely dashed] plot (\x, {3/4*\x+1/4-5/16}); 
\draw [loosely dashed] (1,1) -- (1+0.25*0.6,1-0.25*0.8);

\draw [loosely dashed] (1-0.25,1) -- (1-0.25,1.3);
\draw [loosely dashed] (1-0.25,1) -- (1,1);

\node [left, font=\tiny] at (1-0.25, 1 ) {\footnotesize$A$};

\node [left, font=\tiny] at (1-0.25, {3/4*0.75+1/4} ) {\footnotesize$C'$};

\draw [blue, ultra thick] (1-0.25,1) -- (1-0.25,{3/4*0.75+1/4});
\draw[blue, ultra thick]  (0.93 ,{-1/sqrt(1/0.72^2-1)*(0.93-1)+11/16}) arc (254.1:232.8:0.25);

\draw[blue, ultra thick] (1,1) -- ++(180:0.25);
\draw[ blue, ultra thick] (1,1) -- ++(360-53.1:0.25) ;

\draw [blue, ultra thick] (1,11/16) -- (0.94,{-1/sqrt(1/0.72^2-1)*(0.94-1)+11/16});
\node [right] at (1,11/16) {\footnotesize$C$};
\draw [blue, ultra thick] (1,11/16) -- (1+0.25*0.6,1-0.25*0.8);

\draw [domain=0.93:1, blue, ultra thick] plot (\x, {-1/sqrt(1/0.72^2-1)*(\x-1)+11/16}); 

 \draw [domain=1-0.25:0.85, blue, ultra thick] plot (\x, {tan(-7.075621914)*(\x-0.75)+13/16});

\draw[teal, ultra thick] (1+0.25*0.6,1-0.25*0.8) arc (360-53.1:360-53.1*2:0.25);

\draw [teal, ultra thick]  (1,1 ) -- (1-0.25, 1 );
\draw [teal, ultra thick]  (1-0.25, {3/4*0.75+1/4} ) -- (1-0.25, 1 );
\draw [teal, ultra thick]  (1-0.25, {3/4*0.75+1/4} ) -- ({1-0.25*sin(106.2-90)},{1-0.25*cos(106.2-90)});
\draw[teal, ultra thick] (1,1) -- (1+0.25*0.6,1-0.25*0.8);

\node [left, font=\tiny] at (1-0.25, 1 ) {\footnotesize$A$};
\node [below, font=\tiny] at (1+0.25*0.6,1-0.25*0.8) {\footnotesize$B$};
\matrix [draw, above right] at (1.3,0.7) {
\pic[blue, ultra thick]{legend entry}; &  \node[blue,font=\tiny] {$(\classifier_A,\classifier_B,\frac12,\widetilde h_1)$}; \\
 \pic[teal, ultra thick]{legend entry}; &  \node[teal,font=\tiny] {$(\classifier_A,\classifier_B,0)$}; \\
  \pic[red, ultra thick]{legend entry}; &  \node[red,font=\tiny] {$(\classifier_A,\classifier_B,1)$}; \\
};
\end{tikzpicture}
 \end{subfigure}
 \caption{Lack of nested structure: random-order procedure with disclosure vs fixed-order procedures}
 \label{fig: same graph ran_informed fixed 2}
\end{figure}

We first provide a tempting but oversimplified intuition.
To exclude unqualified agents (feasibility), stringent tests must be used regardless of the testing procedure. 
Notice that this also excludes potential qualified agents.
Hence if the same stringent tests are used under all procedures, then a fixed-order mechanism is better than a random-order mechanism, because this former helps more qualified agents to get selected under the stringent tests (see \cref{fig: oversimplified intuition}).
However, this intuition is oversimplified, because the choice of tests usually interacts with the choice of testing procedures.
While the fixed-order procedure helps qualified agents to get selected more easily, it also helps unqualified agents to get selected more easily.
To exclude unqualified agents, the set of tests that can be used by a fixed-order procedure is also smaller than the set of tests that can be used by a random-order procedure.
Ex-ante, it is not clear which class of procedures is better when combined with different stringent tests.

\begin{figure} 
\centering 
      \begin{tikzpicture}[xscale=4,yscale=4]

\draw [thick] (1,0.65) -- (1,1.8);
\node [right] at (1, 1.9 ) {$\classifier_A$};
\node [right] at (1, 1.75 ) {$+$};

\draw [domain=0.6:2, thick] plot (\x, {3/4*\x+1/4});
\node [right] at (2, 1.75 ) {$\classifier_B$};
\node [left] at (2,1.75) {$+$};

\node [above] at (0.95, 1 ) {\footnotesize$O$};

\draw [domain=1+0.25*0.6:2, loosely dashed] plot (\x, {3/4*\x+1/4-5/16}); 
\draw [loosely dashed] (1,1) -- (1+0.25*0.6,1-0.25*0.8);
\draw[loosely dashed] (0.75,1) -- (0.75,1.8);

\draw[loosely dashed] (1-0.25,1) arc (180:233.1:0.25);
\draw [loosely dashed] (1-0.25,1) -- (1,1);
\draw [loosely dashed] (1,11/16) -- (1-0.25*0.6,1-0.25*0.8);
\draw [loosely dashed] (1,11/16) -- (1+0.25*0.6,1-0.25*0.8);
\draw [loosely dashed] (1,1) -- (1+0.25*0.6,1-0.25*0.8);

\draw [ultra thick,  red] (1.25,0.65) -- (1.25,1.88);
\node [right, red] at (1.25, 1.8 ) {$\classifier_A^{+}$};
\draw [<->, red, densely dashed] (1,1.5) -- (1+0.25,1.5);
\node [above, red] at (1+0.12, 1.5 ) {$1/\eta$};
\draw [domain=0.35:1.8, ultra thick, red] plot (\x, {3/4*\x+1/4+5/16});
\node [left, red] at (1.66, 3/4*1.68+1/4+5/16 ) {$\classifier_B^{+}$};


\draw [<->, red, densely dashed] (1.25,3/4*1.25+1/4+5/16) -- (1.25+0.25*0.6,3/4*1.25+0.25*0.6*3/4+1/4 );
\node [right, red] at (1.25+0.125*0.6,3/4*1.25+3/4*0.125*0.6+7/16) {$1/\eta$};

\draw[red, densely dashed] (1,3/4*1.25+1/4+5/16) arc (180:233.1:0.25);

\draw [densely dashed, red, ultra thick] (1.25,1.25*0.75+0.25) -- (1.25-0.25*0.6,3/4*1.25+1/4+5/16-0.25*0.8);
\draw [densely dashed,red, ultra thick] (1.25,1.25*0.75+0.25) -- (1.25+0.25*0.6,3/4*1.25+0.25*0.6*3/4+1/4);

\draw[blue,densely dashed, ultra thick] (1,3/4*1.25+1/4+5/16) arc (180:360-53.1:0.25);

\fill [blue!60,nearly transparent]  
(1.8,6.4/4+5/16) --
(1.25,3/4*1.25+1/4+5/16) -- (1.25+0.25*0.6,3/4*1.25+0.25*0.6*3/4+1/4 ) --(2,7/4)-- cycle;

\fill [blue!60,nearly transparent]  (1,3/4*1.25+1/4+5/16) -- (1.25,3/4*1.25+1/4+5/16) -- (1.25,1.88) -- (1,1.88)-- cycle;

\fill [blue!60,nearly transparent] (1,3/4*1.25+1/4+5/16) -- (1.25,3/4*1.25+1/4+5/16) -- (1.25-0.25*0.6,3/4*1.25+0.25*0.6*3/4+1/4 )-- cycle;

\fill [blue!60,nearly transparent](1,3/4*1.25+1/4+5/16) coordinate (a) arc (180:233.1:0.25) -- cycle;

\fill [blue!60,nearly transparent] (1.25,3/4*1.25+1/4+5/16) -- (1.25+0.25*0.6,3/4*1.25+0.25*0.6*3/4+1/4 )--(1.25,4.75/4)--(1.25-0.25*0.6,3/4*1.25+0.25*0.6*3/4+1/4 )--cycle;

\end{tikzpicture}
\caption{An oversimplified intuition}
\label{fig: oversimplified intuition}
  \rule{0in}{1.2em}$^\dag$\scriptsize  This graph shows that under the same stringent tests $\classifier_A^+$ and $\classifier_B^+$, the fixed-order procedure that offers $\classifier_A^+$ first selects a larger set of potential qualified agents than the random-order procedure without disclosure and randomization probability $\frac12$.
  However, the optimal tests may not be the same for the two procedures. Hence, This intuition is oversimplified.
\end{figure}
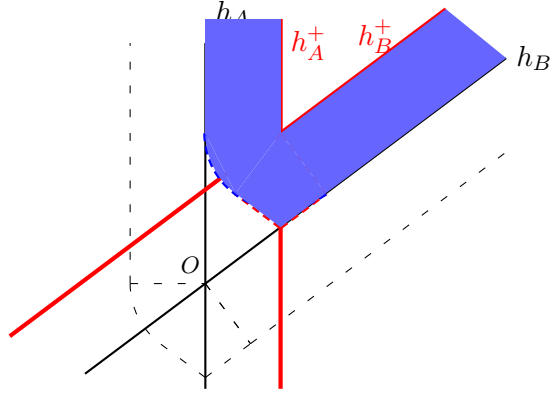

We propose a new constructive argument to prove this result.
The idea is that fixing any feasible random order mechanism, we construct another feasible fixed-order mechanism that dominates the random-order mechanism.
We distinguish two cases.
In \cref{lem:feasible-informed-rand-distance cost}, we deal with the class of random-order mechanisms with disclosure.
In \cref{lem:feasible-uninformed-rand-distance cost}, we deal with the class of random-order mechanisms without disclosure.
This is a more involved case because the set of tests feasible for random-order mechanisms without disclosure is usually strictly larger.
It requires extra steps to construct the feasible fixed order mechanisms.


\begin{lemma}\label{lem:feasible-informed-rand-distance cost}
 Suppose the  random-order mechanism  with disclosure $(\tilde\classifier_A,\tilde\classifier_B,q,\widetilde h_1)$ is feasible. 
 Then the two fixed-order mechanisms $(\tilde\classifier_A,\tilde\classifier_B,1)$ and $(\tilde\classifier_A,\tilde\classifier_B,0)$ are feasible and one of them is no worse than the  random-order mechanism with disclosure.
\end{lemma}

This is the simpler case because the set of tests feasible for random-order mechanisms with disclosure is the same as that for fixed-order mechanisms.
The challenge of this proof comes from that fixing the tests, neither of the two fixed-order mechanisms is directly comparable to any random-order mechanisms with disclosure.
This is because there is no nested structure between the set of types that can get selected by any of the two fixed-order mechanisms and that by any random-order mechanism with disclosure.
For example, see \cref{fig: same graph ran_informed fixed 2}. 
Our proof strategy is to introduce a mixed mechanism that randomizes over the two fixed-order mechanisms.
After that, we first apply a similar set inclusion argument to show that the mixed  mechanism dominates the random-order mechanism with disclosure.
And then we apply a probability argument to show that one of the fixed-order mechanisms used in the mixed mechanism dominates the mixed mechanism.


Next, consider any random order mechanism without disclosure $(\tilde \classifier_A,\tilde \classifier_B, q, \varnothing)$. 
Let $\tilde O$ be the intersection point of the boundary lines of $ \tilde\classifier_A$ and $\tilde \classifier_B$. Let $\classifier_A'$ and $\classifier_B'$ be the tests that contain $\tilde O$ on their boundary lines and have the same normal vectors $\weights_A$ and $\weights_B$ as the half plane $\classifier_A$ and $\classifier_B$ respectively. 

\begin{lemma}\label{lem:feasible-uninformed-rand-distance cost}
 Suppose the random-order mechanism  without disclosure $(\tilde\classifier_A,\tilde\classifier_B,q,\nullset)$ is feasible.  
 Then two fixed order mechanisms $(\tilde\classifier_A,\classifier_B',1)$ and $(\classifier_A',\tilde\classifier_B,0)$ are also feasible.
 Moreover, one of these two fixed order mechanisms is no worse than the random-order mechanism without disclosure.
\end{lemma}

As we have eluded earlier, the set of tests feasible for random-order mechanisms without disclosure is usually larger. \cref{fig:shifted non-parallel random} provides an example.\footnote{One might hope that for those non-parallel tests, we can show that they are dominated by some parallel tests. 
However, this is not the case.
There is usually no way to rank the  tests that are feasible for random-order mechanisms without disclosure, even if we fix the randomization probability $q$.
 One example is \cref{fig:theta<30 random}. When the angle of the tests increase, the set of types being selected does not contain the set of types being selected under smaller angle of tests.
 Hence the constructive argument we propose here is required to compare any random-order mechanism without disclosure and any fixed-order mechanism.} 
Under the random order mechanism without disclosure $ (\tilde\classifier_A,\tilde \classifier_B,\frac12,\nullset)$, every selected agent uses a one-step strategy to pass both tests together. Hence the set of attributes being selected under this mechanism is the blue region and $\tilde\classifier_A\cap\tilde \classifier_B$.
Since the blue region is contained in the qualified region, this  mechanism is feasible.
However, since $\tilde\classifier_i,i\in\{A,B\}$ is not parallel to $\classifier_i,i\in\{A,B\}$, the corresponding fixed-order mechanisms are not feasible.
To see this, consider the unqualified attributes  $\orifeatures\in \widetilde h_B$ that are very closed to $\classifier_A$ but do not satisfy $\classifier_A$,  as shown in \cref{fig:shifted non-parallel random}.
Because the tests are non-parallel, we can always find such $\orifeatures$ so that its cost to pass $\widetilde h_A$ is less than one.
This implies that the fixed-order mechanism with $\widetilde h_B$ as the first test is not  feasible.
Similarly, we can argue that the fixed-order mechanism with $\widetilde h_A$ as the first test is not  feasible.

\begin{figure}[h]  
\centering 
\begin{tikzpicture}[xscale=4,yscale=4]

\draw [thick] (1,0.65) -- (1,1.8);
\node [right] at (1, 1.9 ) {$\classifier_A$};
\node [right] at (1, 1.75 ) {$+$};

\draw [domain=0.6:2, thick] plot (\x, {3/4*\x+1/4});
\node [right] at (2, 1.75 ) {$\classifier_B$};
\node [left] at (2,1.75) {$+$};

\node [right, red] at (1.26, 1.84 ) {$\widetilde\classifier_A$};
\draw [domain=1.17:1.29, ultra thick, red] plot (\x, {10*(\x-5/4)+3/2});
\draw [domain=0.5:1.76, ultra thick, red] plot (\x, {3.5/4*(\x-5/4)+3/2});
\node [left, red] at (1.6, 1.84 ) {$\widetilde\classifier_B$};

\node [left, font=\tiny] at (0.98, {3.5/4*(0.98-5/4)+3/2}) {\footnotesize$\orifeatures$};
\node [black] at (0.98, {3.5/4*(0.98-5/4)+3/2}) {\textbullet};

\draw[red, densely dashed] (1,3/4*1.25+1/4+5/16) arc (180:360-53.1:0.25);

\fill [blue!60,nearly transparent]  
((1.76, {3.5/4*(1.76-5/4)+3/2}) --
(1.25,3/4*1.25+1/4+5/16) -- (1.25+0.25*0.6,3/4*1.25+0.25*0.6*3/4+1/4 ) --(2,{3.5/4+1})-- cycle;

\fill [blue!60,nearly transparent]  (1,3/4*1.25+1/4+5/16) -- (1.25,3/4*1.25+1/4+5/16) -- (1.28, {10*(1.28-5/4)+3/2}) -- (1.04, {10*(1.28-5/4)+3/2})-- cycle;

\fill [blue!60,nearly transparent] (1,3/4*1.25+1/4+5/16) -- (1.25,3/4*1.25+1/4+5/16) -- (1.25+0.25*0.6,3/4*1.25+0.25*0.6*3/4+1/4 )-- cycle;

\fill [blue!60,nearly transparent](1,3/4*1.25+1/4+5/16) coordinate (a) arc (180:360-53.1:0.25) -- cycle;

\end{tikzpicture}
\caption{Feasible non-parallel tests for random-order mechanism without disclosure}
\label{fig:shifted non-parallel random}
\end{figure}
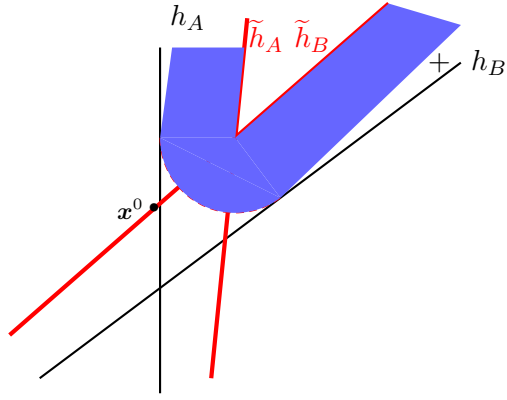

Here is a sketch of the proof for this lemma.
First we show that the two fixed order mechanisms $(\tilde\classifier_A,\classifier_B',1)$ and $(\classifier_A',\tilde\classifier_B,0)$ are feasible (See \cref{fig:construction uninformed random}). 
By combining one non-parallel test with one parallel test and a careful choice of the order of the tests, we make sure that no unqualified agent is selected under each of these fixed-order mechanisms.

Second, we use a similar argument to show dominance: we construct a mixed mechanism that randomizes over the two feasible fixed-order mechanisms $(\tilde\classifier_A,\classifier_B',1)$ and $(\classifier_A',\tilde\classifier_B,0)$, and then show that this mixed mechanism is weakly better than the random-order mechanism  without disclosure using the set inclusion argument. 
It follows naturally that  one of the two fixed-order mechanisms dominates the random-order mechanism without disclosure using a probability argument. All omitted proofs are in \cref{appendix: distance cost}.

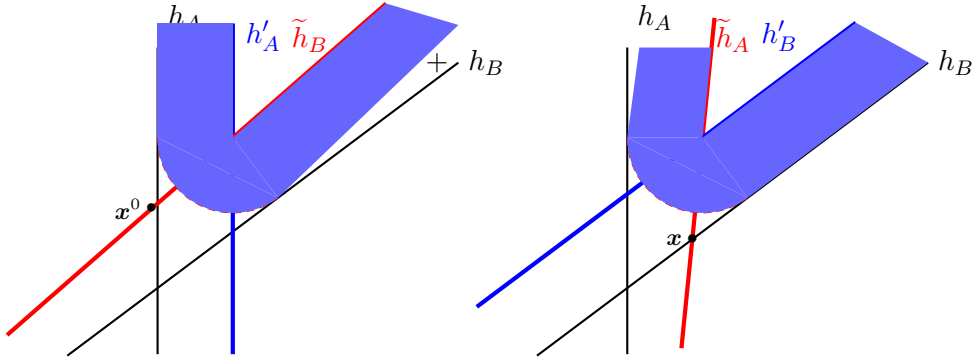
\begin{figure}[h]  
\centering 
\begin{subfigure}{0.45\textwidth}
\begin{tikzpicture}[xscale=4,yscale=4]

\draw [thick] (1,0.78) -- (1,1.8);
\node [right] at (1, 1.9 ) {$\classifier_A$};
\node [right] at (1, 1.75 ) {$+$};

\draw [domain=0.7:2, thick] plot (\x, {3/4*\x+1/4});
\node [right] at (2, 1.75 ) {$\classifier_B$};
\node [left] at (2,1.75) {$+$};

\draw [ultra thick, blue] (1.25,0.78) -- (1.25,1.88);
\node [right, blue] at (1.26, 1.84 ) {$\classifier_A'$};
\draw [domain=0.5:1.76, ultra thick, red] plot (\x, {3.5/4*(\x-5/4)+3/2});
\node [left, red] at (1.6, 1.84 ) {$\widetilde\classifier_B$};

\node [left, font=\tiny] at (0.98, {3.5/4*(0.98-5/4)+3/2}) {\footnotesize$\orifeatures$};
\node [black] at (0.98, {3.5/4*(0.98-5/4)+3/2}) {\textbullet};

\draw[red, densely dashed] (1,3/4*1.25+1/4+5/16) arc (180:360-53.1:0.25);

\fill [blue!60,nearly transparent]  
((1.76, {3.5/4*(1.76-5/4)+3/2}) --
(1.25,3/4*1.25+1/4+5/16) -- (1.25+0.25*0.6,3/4*1.25+0.25*0.6*3/4+1/4 ) --(2,{3.5/4+1})-- cycle;

\fill [blue!60,nearly transparent]  (1,3/4*1.25+1/4+5/16) -- (1.25,3/4*1.25+1/4+5/16) -- (1.25,1.88) -- (1,1.88)-- cycle;

\fill [blue!60,nearly transparent] (1,3/4*1.25+1/4+5/16) -- (1.25,3/4*1.25+1/4+5/16) -- (1.25+0.25*0.6,3/4*1.25+0.25*0.6*3/4+1/4 )-- cycle;

\fill [blue!60,nearly transparent](1,3/4*1.25+1/4+5/16) coordinate (a) arc (180:360-53.1:0.25) -- cycle;

\end{tikzpicture}
\end{subfigure}
\begin{subfigure}{0.45\textwidth}
\begin{tikzpicture}[xscale=4,yscale=4]

\draw [thick] (1,0.78) -- (1,1.8);
\node [right] at (1, 1.9 ) {$\classifier_A$};
\node [right] at (1, 1.75 ) {$+$};

\draw [domain=0.7:2, thick] plot (\x, {3/4*\x+1/4});
\node [right] at (2, 1.75 ) {$\classifier_B$};
\node [left] at (2,1.75) {$+$};

\node [right, red] at (1.26, 1.84 ) {$\widetilde\classifier_A$};
\draw [domain=1.18:1.29, ultra thick, red] plot (\x, {10*(\x-5/4)+3/2});
\draw [domain=0.5:1.76, ultra thick, blue] plot (\x, {3/4*(\x-5/4)+3/2});
\node [left, blue] at (1.6, 1.84 ) {$\classifier_B'$};

\node [left, font=\tiny] at (45/37, {3/4*(45/37-1)+1}) {\footnotesize$\features$};
\node [black] at (45/37, {3/4*(45/37-1)+1}) {\textbullet};

\draw[red, densely dashed] (1,3/4*1.25+1/4+5/16) arc (180:360-53.1:0.25);

\fill [blue!60,nearly transparent]  
((1.76, {3/4*(1.76-5/4)+3/2}) --
(1.25,3/4*1.25+1/4+5/16) -- (1.25+0.25*0.6,3/4*1.25+0.25*0.6*3/4+1/4 ) --(2,{3/4+1})-- cycle;

\fill [blue!60,nearly transparent]  (1,3/4*1.25+1/4+5/16) -- (1.25,3/4*1.25+1/4+5/16) -- (1.28, {10*(1.28-5/4)+3/2}) -- (1.04, {10*(1.28-5/4)+3/2})-- cycle;

\fill [blue!60,nearly transparent] (1,3/4*1.25+1/4+5/16) -- (1.25,3/4*1.25+1/4+5/16) -- (1.25+0.25*0.6,3/4*1.25+0.25*0.6*3/4+1/4 )-- cycle;

\fill [blue!60,nearly transparent](1,3/4*1.25+1/4+5/16) coordinate (a) arc (180:360-53.1:0.25) -- cycle;

\end{tikzpicture}
\end{subfigure}
\caption{Constructing feasible mixed mechanism for random-order mechanisms without disclosure}
\label{fig:construction uninformed random}
\end{figure}

\bigskip


\subsection{Cheap-talk}\label{subsec:cheap talk manipulation}
Third, we consider the case where the principal can tailor the tests to different types of agent, i.e., cheap-talk communication is allowed upfront.
We show that when cheap-talk communication is allowed,  a fixed-order sequential mechanism achieves the first best.

\begin{theorem}\label{thm: optimal max qualified cheap talk}
Suppose $M=\Featurespace$. 
    For any distribution $\dist$, there exists a feasible fixed-order sequential mechanism that achieves the first best.
\end{theorem}

The proof is constructive.
We first introduce some useful notations. 
Let $O$ be the intersection point of the boundary lines of $\classifier_A$ and $\classifier_B$. 
Recall that  $\classifier_i^+, i\in \{A,B\} $ is a test obtained by shifting $\classifier_i$ along its normal vector $\weights_i$ by a distance of $1/\eta$.
Denote  the intersection point of the boundary lines of $\classifier_A^+ $ and $\classifier_B^+ $ by $O^+$.
Let $\widehat\classifier_A$ be the half plane obtained by rotating $\classifier_A^+ $ around point $O^+$ towards $\classifier_B^+ $ until the boundary line is exactly the angle bisector of $\classifier_A^+ $ and $\classifier_B^+$, which means the boundary of $\widehat \classifier_A$ goes through $O$.

Consider the fixed-order procedure $(\classifier_A^+,\classifier_B^+,1)$ and the fixed-order procedure $(\classifier_A^+,\widehat\classifier_B,0)$.
Let $A$ be the intersection of boundaries of $\classifier_B$ and $\classifier_A^+$. Let $B$ be the intersection of 
boundaries of $\classifier_A$ and $\classifier_B^+$.
\cref{fig:cheap-talk} plots the set if qualified types that are selected by the two fixed-order procedures.
The fixed-order procedure $(\classifier_A^+,\classifier_B^+,1)$ that uses two shifted parallel tests select all qualified types except for those in 
the yellow triangle $\Updelta \text{OAB}$ (left panel).
However, any type in 
the yellow triangle $\Updelta \text{OAB}$ is selected by 
 the fixed-order procedure $(\classifier_A^+,\widehat\classifier_B,0)$ that uses one non-parallel test (right panel). Hence offering a menu of these two procedures can select all qualified types and no unqualified type.

\begin{figure}
    \centering
    \begin{subfigure}[b]{0.45\linewidth}
         \begin{tikzpicture}[xscale=5.6,yscale=5.6]

\fill [yellow!50,nearly transparent]  (1,1.5) --(1,1)-- (1.25,1.25*0.75+0.25)--(1.25,1.5)  -- cycle;

\fill [white](1,3/4*1.25+1/4+5/16) coordinate (a) arc (180:233.1:0.25) -- cycle;
\fill [white]  (1,3/4*1.25+1/4+5/16) -- (1.25,3/4*1.25+1/4+5/16) -- (1.25-0.25*0.6,3/4*1.25+0.25*0.6*3/4+1/4 )-- cycle;
\fill [white] (1.25,3/4*1.25+1/4+5/16) -- (1.25+0.25*0.6,3/4*1.25+0.25*0.6*3/4+1/4 )--(1.25,4.75/4)--(1.25-0.25*0.6,3/4*1.25+0.25*0.6*3/4+1/4 )--cycle;

\draw [thick] (1,1) -- (1,1.8);
\node [right] at (1, 1.9 ) {$\classifier_A$};
\node [right] at (1, 1.75 ) {$+$};

\draw [domain=1:2, thick] plot (\x, {3/4*\x+1/4});
\node [right] at (1.8, 0.75*1.8+0.25 ) {$\classifier_B$};
\node [left] at (1.8, 0.75*1.8+0.25 ) {$+$};

\node [above, orange] at (0.95, 1 ) {\footnotesize$O$};

\draw [ultra thick,  red] (1.25,0.75*1.25+0.25) -- (1.25,1.88);
\node [right, red] at (1.25, 1.8 ) {$\classifier_A^{+}$};

\draw [domain=1:1.8, ultra thick, red] plot (\x, {3/4*\x+1/4+5/16});
\node [left, red] at (1.66, 3/4*1.68+1/4+5/16 ) {$\classifier_B^{+}$};

\node [left, orange] at (1,1.5) {$A$};
\node [right, orange] at (1.25,1.25*0.75+0.25) {$B$};

\draw[red, densely dashed, ultra thick] (1,3/4*1.25+1/4+5/16) arc (180:233.1:0.25);

\draw [densely dashed, red, ultra thick] (1.25,1.25*0.75+0.25) -- (1.25-0.25*0.6,3/4*1.25+1/4+5/16-0.25*0.8);

\end{tikzpicture}
    \end{subfigure}
    \begin{subfigure}[b]{0.45\linewidth}
        \begin{tikzpicture}[xscale=5.6,yscale=5.6]

\node [above] at (0.8, {tan(30)*0.8-0.25/cos( 30)} ) {$\classifier_B$};
\node [below, purple] at (0.8, {tan(30)*0.8-0.25/cos( 30)} ) {$D$};

\draw [domain={0.25*sin(60)}:0.5, densely dashed, red, ultra thick] plot (\x, {tan(60)*\x-0.25/cos(60)});
\node[above, purple] at (0.5, {tan(60)*0.5-0.25/cos(60)}) {$E$};

\draw[densely dashed, red, ultra thick] ({0.25*sin(60)},{-0.25*cos(60)}) arc (-30:-60:0.25);

\draw [thick] (-0.25,{-0.25/tan(30)}) -- (-0.25,0.45);
\draw [domain=-0.25:0.8, thick] plot (\x, {tan(30)*\x-0.25/cos( 30)});

\node [above] at (-0.25, 0.45) {$\classifier_A$};
\node [right] at (-0.25,0.4 ) {$+$};

\draw [domain=-0.25:0.25, ultra thick, red] plot (\x, {tan(60)*\x});
\node [above, red] at (0.25, {tan(60)*0.25} ) {$\widehat\classifier_B$};
\node [left, red] at (0.2, {tan(60)*0.2}) {$+$};

\draw [ultra thick, blue] (0,{-0.25/cos( 30)}) -- (0,0.45);
\node [above, blue] at (0, 0.45 ) {$\classifier_A^+$};

\node[right, purple] at  ({0.25*sin(30)},{-0.25*cos(30)}) {$C$};

\node [left] at (0.8, {tan(30)*0.8-0.25/cos( 30)}) {$+$};

\fill [purple!60,nearly transparent]  (0.5,{tan(60)*0.5-0.25/cos(60)}) -- ({(0.25/cos(60)-0.25/cos(30))/(tan(60)-tan(30))},{tan(30)*(0.25/cos(60)-0.25/cos(30))/(tan(60)-tan(30))-0.25/cos( 30)}) -- (0.8, {tan(30)*0.8-0.25/cos( 30)}) -- cycle;
\node [above, thick, purple] at (0.6, {tan(30)*0.62}) {loss};

\end{tikzpicture}
    \end{subfigure}
    \caption{A menu of testing procedures that achieves the first best}
    \label{fig:cheap-talk}
\end{figure}
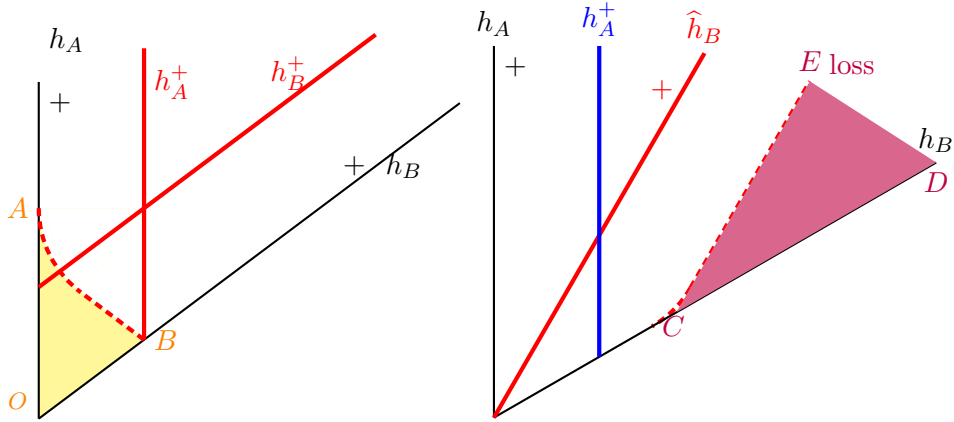

The formal proof is provided in \cref{appendix:cheap talk manipulation}. The key idea is that the union of the set of types selected by each fixed-order mechanism exactly coincides with the qualified region.

\bigskip

\subsection{Comparison to simultaneous mechanisms}\label{subsec: simultaneous manipulation}
Second, we ask whether conducting both tests together could improve upon conducting one test at a time.
In other words, we ask whether simultaneous mechanisms could work better than sequential mechanisms.
\citet{zigzag} have studied the comparison between fixed-order mechanisms and simultaneous mechanisms under the manipulation setting and the same cost function.
 We state a similar result and in \cref{thm:opt_manipulation} we generalize it to a broader class of cost functions.\footnote{We view our \cref{lem:fix-simultaneous-distance cost} as a correction of Theorem 4.4 in \citet{zigzag}. 
The main issue there is that they first show the optimal tests for both simultaneous mechanisms and fixed order mechanisms must use shifted parallel tests.
\cref{fig:cheap-talk} serves as a counterexample to illustrate a fixed-order mechanism that uses shifted non-parallel tests could dominate one that uses shifted parallel tests under some distribution.
}
\begin{proposition}\label{lem:fix-simultaneous-distance cost}
    The optimal fixed-order mechanism dominates the optimal simultaneous mechanism. Moreover, it uses stringent tests $\widetilde\classifier_A, \widetilde\classifier_B$, i.e., $\widetilde\classifier_A\cap \widetilde\classifier_B \subsetneq \classifier_A\cap \classifier_B$.
\end{proposition}

The proof idea is to first show that the optimal simultaneous mechanism uses shifted tests that are parallel to the principal's true criteria.
And then we show that one fixed-order mechanism that uses the same tests weakly increases the chance of selecting an agent.

To understand this proposition, suppose the test $\classifier_{A} \cap \classifier_{B}$ is used in the simultaneous mechanism. 
This implies that any agent with true attributes $\orifeatures\in \classifier_{A} \cap \classifier_{B}$
will not manipulate and will be selected. In addition, agents with true
attributes $\orifeatures\notin \classifier_{A} \cap \classifier_{B}$ will be selected only if they pay a
 cost to adopt attributes that are in the qualified region. Since
the benefit of doing so (i.e., getting selected) is 1, only those with cost
no greater than 1 will be willing to do so. The set of candidates who are
willing to do such preparation is 
\[
\manipulation=\left\{ \orifeatures\notin h_{A}\cap h_{B}:\min_{\features\in h_{A}\cap
h_{B}}\onecost(\orifeatures,\features)\leq 1\right\} \text{.}
\]%
This set $\manipulation$ comprises all points whose distance to either the two edges of
the qualified region $h_{A}\cap h_{B}$ is no greater than 1.

The simultaneous mechanism that uses tests $%
h_{A}\cap h_{B}$ selects some unqualified agents, so does not satisfy
the principal's objective. In order not to select any unqualified agents,
the principal must choose a more stringent test. For example,   $h_{A}^{+}\cap h_{B}^{+}$, where $h_{i}^{+}$ is
obtained by a parallel shift of $h_{i}$ by a distance of $1/\eta$. It turns out $h_{A}^{+}\cap
h_{B}^{+}$ is the optimal simultaneous mechanism.
This is because shifted non-parallel tests are either not feasible or not optimal under any simultaneous procedure (See \cref{fig:non-parallel tests not optimal}).

\begin{figure}
\centering 
    \begin{subfigure}{0.45\textwidth}
\centering 
\begin{tikzpicture}[xscale=4,yscale=4]

\draw [thick] (1,0.65) -- (1,1.8);
\node [right] at (1, 1.9 ) {$\classifier_A$};
\node [right] at (1, 1.75 ) {$+$};

\draw [domain=0.6:2, thick] plot (\x, {3/4*\x+1/4});
\node [right] at (2, 1.75 ) {$\classifier_B$};
\node [left] at (2,1.75) {$+$};

\draw [ultra thick, red] (1.25,0.65) -- (1.25,1.88);
\node [right, red] at (1.25, 1.8 ) {$\classifier_A^{+}$};
\draw [domain=0.35:1.8, ultra thick, red] plot (\x, {2.5/4*(\x-5/4)+3/2});
\node [left, red] at (1.64, 3/4*1.66+1/4+5/16 ) {$\widehat\classifier_B$};

\draw[red, densely dashed] (1,3/4*1.25+1/4+5/16) arc (180:360-53.1:0.25);

\fill [blue!60,nearly transparent]  
(1.8,{2.5/4*(1.8-5/4)+3/2}) --
(1.25,3/4*1.25+1/4+5/16) -- (1.25+0.25*0.6,3/4*1.25+0.25*0.6*3/4+1/4 ) --(2,6.5/4)-- cycle;

\fill [blue!60,nearly transparent]  (1,3/4*1.25+1/4+5/16) -- (1.25,3/4*1.25+1/4+5/16) -- (1.25,1.88) -- (1,1.88)-- cycle;

\fill [blue!60,nearly transparent] (1,3/4*1.25+1/4+5/16) -- (1.25,3/4*1.25+1/4+5/16) -- (1.25+0.25*0.6,3/4*1.25+0.25*0.6*3/4+1/4 )-- cycle;

\fill [blue!60,nearly transparent](1,3/4*1.25+1/4+5/16) coordinate (a) arc (180:360-53.1:0.25) -- cycle;

\end{tikzpicture}
\caption{Not feasible non-parallel tests}
\end{subfigure}
\begin{subfigure}{0.45\textwidth} 
\centering 
\begin{tikzpicture}[xscale=4,yscale=4]

\draw [thick] (1,0.65) -- (1,1.8);
\node [right] at (1, 1.9 ) {$\classifier_A$};
\node [right] at (1, 1.75 ) {$+$};

\draw [domain=0.6:2, thick] plot (\x, {3/4*\x+1/4});
\node [right] at (2, 1.75 ) {$\classifier_B$};
\node [left] at (2,1.75) {$+$};

\draw [ultra thick, red] (1.25,0.65) -- (1.25,1.88);
\node [right, red] at (1.25, 1.8 ) {$\classifier_A^{+}$};
\draw [domain=0.35:1.76, ultra thick, red] plot (\x, {3.5/4*(\x-5/4)+3/2});
\node [left, red] at (1.63, {3.5/4*(1.65-5/4)+3/2} ) {$\widetilde \classifier_B$};


\draw[red, densely dashed] (1,3/4*1.25+1/4+5/16) arc (180:360-53.1:0.25);

\fill [blue!60,nearly transparent]  
((1.76, {3.5/4*(1.76-5/4)+3/2}) --
(1.25,3/4*1.25+1/4+5/16) -- (1.25+0.25*0.6,3/4*1.25+0.25*0.6*3/4+1/4 ) --(2,{3.5/4+1})-- cycle;

\fill [blue!60,nearly transparent]  (1,3/4*1.25+1/4+5/16) -- (1.25,3/4*1.25+1/4+5/16) -- (1.25,1.88) -- (1,1.88)-- cycle;

\fill [blue!60,nearly transparent] (1,3/4*1.25+1/4+5/16) -- (1.25,3/4*1.25+1/4+5/16) -- (1.25+0.25*0.6,3/4*1.25+0.25*0.6*3/4+1/4 )-- cycle;

\fill [blue!60,nearly transparent](1,3/4*1.25+1/4+5/16) coordinate (a) arc (180:360-53.1:0.25) -- cycle;

\end{tikzpicture}
 \caption{Not optimal non-parallel tests}  
 \end{subfigure}
 \caption{Shifted non-parallel tests are not optimal}
 \label{fig:non-parallel tests not optimal}
\end{figure}
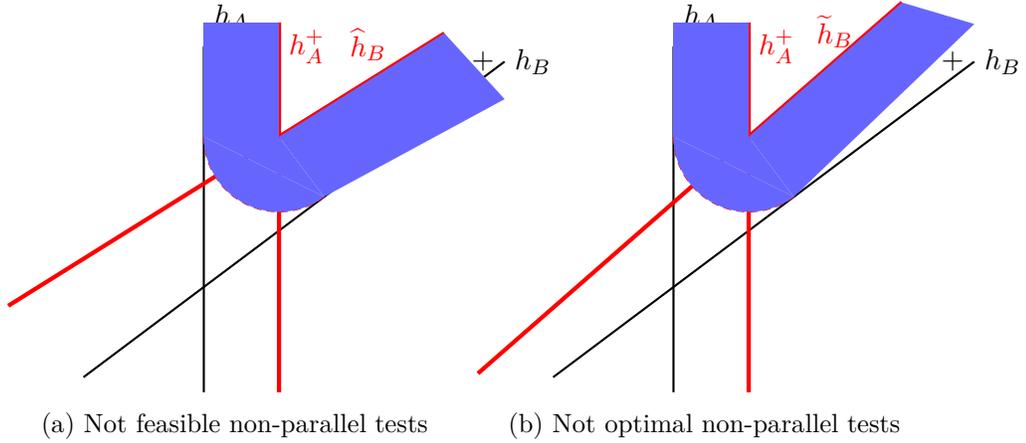

Suppose now the principal uses a fixed-order sequential mechanism that offers test $h_{A}$ at the first stage, and (2) among those
selected in (1), those who can pass test $h_{B}$ will be eventually
selected. We can characterize the set of attributes that get selected under such a fixed-order mechanism (\cref{fig: zig zag}) using the geometric intuition offered in \cref{subsec:example}.

In order not to select any unqualified agent, the principal needs to choose
more stringent tests instead of $h_{A}$ and $h_{B}$. One obvious candidate is the same stringent tests used in the optimal simultaneous mechanism: $%
h_{A}^{+}$ and $h_{B}^{+}$.
It remains to check that all agents selected by the fixed-order sequential mechanism are qualified under tests $%
h_{A}^{+}$ and $h_{B}^{+}$.
It turns out that this is indeed the case here (see \cref{fig: oversimplified intuition}).\footnote{In general, there are tests that are feasible under simultaneous mechanisms but are not feasible under fixed-order sequential mechanisms.}

\citet{zigzag} showed that the optimal simultaneous mechanism  $h_{A}^{+}\cap h_{B}^{+}$ is dominated by the fixed-order sequential mechanism $%
h_{A}^{+}$ and $h_{B}^{+}$; however, there is an error in the proof of their main result (Theorem 4.4), which we correct here. 
Observe that by using the fixed-order sequential mechanism $%
h_{A}^{+}$ and $h_{B}^{+}$, the principal selects more agents than
by using the simultaneous mechanism $h_{A}^{+}\cap h_{B}^{+}$. This is
because the fixed-order sequential mechanism enables agents to use two-step
strategies, which are cheaper (sometimes strictly) than the best one-step
strategy. 
We conclude that the optimal simultaneous mechanism is dominated by some sequential mechanisms. 
Moreover, in order to avoid unqualified agents, any feasible sequential mechanisms must use stringent tests.
 \section{Investment}\label{subsec: distance investment}
In this section, we consider the investment setting. Notice that the incentive of the agent is the same as before. The main difference between this setting and the manipulation setting is the incentive of the principal. All omitted proofs in this section are provided in \cref{appendix: distance investment}.

\subsection{Optimal sequential mechanism}
First, we  show that under a mild condition, there exists a random-order mechanism without disclosure that can achieve the first best (\cref{thm: true effort sequential}).

\begin{theorem}\label{thm: true effort sequential}
    If $\theta\geq 30^{\circ}$, for any distribution $\dist$, 
   then the   random order mechanism without disclosure $(\classifier_A,\classifier_B,\frac12,\nullset)$ achieves the first best.
\end{theorem}

\begin{lemma}\label{claim same behavior}
    If $\theta\geq 30^{\circ}$, then in $(\classifier_A,\classifier_B,\frac12,\nullset)$, any attributes $\features\notin \classifier_A\cap\classifier_B$ that get accepted by this mechanism use a one-step strategy.
\end{lemma}

\begin{proof}[Proof of \cref{thm: true effort sequential}]
   By \cref{claim same behavior}, any agent that gets selected adopt some attributes in $\classifier_A\cap\classifier_B$ under the random-order mechanism $(\classifier_A,\classifier_B,\frac12,\nullset)$.
    We show that the described mechanism selects all potential qualified agents and zero unqualified agent, which achieves the first best.
   First, any qualified agent gets selected with zero cost under the described mechanism.
    Second, notice that under the investment setting, every unqualified agent who can improve to some qualified attributes with cost less than one become qualified.
    Such an agent is selected by the described mechanism.
    Third, any other attributes can only get selected by the described mechanism with a cost strictly larger than one, which is profitable.
    So these attributes remain unqualified and are not selected.
\end{proof}

\cref{thm: true effort sequential} hinges that in  the investment setting, the key for any sequential mechanism to achieve the first best is that every agent prefers using some one-step strategy in the sequential mechanism.
This observation will be made concrete in \cref{lem: optimal max qualified improving effort}.

How do other sequential mechanisms perform when the agent invests?
First, if $\theta<90^{\circ}$, there does not exist any sequential mechanism with fixed order that can achieve the first best.
Consider the fixed-order sequential mechanism that first offers $\classifier_A$ and then $\classifier_B$.
Under this mechanism, some attributes prefer using a two-step strategy and they do not become qualified after the investment. 
For example, in \cref{subfig:BR sequential fixed},  the agent with attributes $\orifeatures$ (on the boundary line of $\classifier_A$) uses the following strategy: he does not change his attributes before the first test, and invests to become $\secondfeatures$ (its projection on the boundary line of $\classifier_B$) before the second test.
Although this agent passes both tests and gets selected under the current mechanism, he is still unqualified after investment.
To prevent selecting any unqualified agent, the principal needs to use at least one stringent test under this mechanism.
For example, offering $\tilde{\classifier}_A$ first and then $\classifier_B$ as shown in \cref{fig:feasible fixed improvement}. 

    \begin{figure}[t]
\centering
\begin{subfigure}[b]{0.45\linewidth}
\begin{tikzpicture}[xscale=7.8,yscale=7.8]

\draw [domain=0.62:1.4, thick] plot (\x, {3/4*\x+1/4});
\node [left] at (1.3, 1.25 ) {$ \classifier_B$};
\draw [thick] (1,0.62) -- (1,1.3);
\node [right] at (1, 1.25) {$\classifier_A$};
\node [above,font=\tiny] at (0.97, 1 ) {\footnotesize$O$};
\node [left] at (1.2,1.15) {$+$};
\node [right] at (1, 1.15) {$+$};

\draw [domain=1+0.25*0.6:1.45, loosely dashed] plot (\x, {3/4*\x+1/4-5/16}); 
\draw [loosely dashed] (1,1) -- (1+0.25*0.6,1-0.25*0.8);

\draw [loosely dashed] (1-0.25,1) -- (1-0.25,1.3);
\draw [loosely dashed] (1-0.25,1) -- (1,1);

\draw[blue, ultra thick] (1-0.25,1) arc (180:233.1:0.25);
\draw[blue, ultra thick] (1,1) -- ++(180:0.25);
\draw[blue, ultra thick] (1,1) -- ++(233.1:0.25) ;

\draw [blue, ultra thick] (1,1) -- (1-0.25*0.6,1-0.25*0.8);

\draw [ blue, ultra thick] (1,11/16) -- (1-0.25*0.6,1-0.25*0.8);

\draw [blue, ultra thick] (1,11/16) -- (1+0.25*0.6,1-0.25*0.8);
\draw [blue, ultra thick] (1,1) -- (1+0.25*0.6,1-0.25*0.8);
\draw [blue, ultra thick, <-] (1+0.01*0.6,1-0.01*0.8) -- (1+0.25*0.6,1-0.25*0.8);

\node [right, font=\tiny] at (1,11/16) {$\orifeatures=\firstfeatures$};

\draw [densely dashed, ->]  (1,11/16) -- (1-0.15, 1-0.15*3/4);
\draw [ultra thick, green]  (1,1) -- (1-0.15, 1-0.15*3/4);
\node [left, font=\tiny] at (1-0.15, 1-0.15*3/4) {$\secondfeatures$};

\end{tikzpicture}
\caption{BR in fixed order mech: $\classifier_A\rightarrow\classifier_B$ }
\label{subfig:BR sequential fixed}
\end{subfigure}
  \begin{subfigure}[b]{0.45\linewidth}
\begin{tikzpicture}[xscale=7.6,yscale=7.6]

\draw [domain=0.66:1.4, thick] plot (\x, {3/4*\x+1/4});
\node [left] at (1.36, 1.28 ) {$\classifier_B$};
\draw [thick] (1,0.66) -- (1,1.3);
\node [right] at (1, 1.28) {$\classifier_A$};
\node [above,font=\tiny] at (0.97, 1 ) {\footnotesize$O$};

\draw [domain=1+0.25*0.6:1.45, densely dashed] plot (\x, {3/4*\x+1/4-5/16}); 
\draw [densely dashed] (1,1) -- (1+0.25*0.6,1-0.25*0.8);

\draw [densely dashed] (1-0.25,1) -- (1-0.25,1.3);
\draw [densely dashed] (1-0.25,1) -- (1,1);

\draw[densely dashed] (1-0.25,1) arc (180:233.1:0.25);
\draw[densely dashed] (1,1) -- ++(180:0.25);
\draw[densely dashed] (1,1) -- ++(233.1:0.25) ;

\draw [densely dashed] (1,1) -- (1-0.25*0.6,1-0.25*0.8);
\draw [ densely dashed] (1,11/16) -- (1-0.25*0.6,1-0.25*0.8);
\draw [densely dashed] (1,11/16) -- (1+0.25*0.6,1-0.25*0.8);
\draw [densely dashed] (1,1) -- (1+0.25*0.6,1-0.25*0.8);

\draw [ultra thick, red] (1+0.25*0.6,0.66) -- (1+0.25*0.6,1.3);

\draw [<->, red, dotted] (1,1.16) -- (1+0.25*0.6,1.16);
\node [above, red] at (1+0.07, 1.16 ) {$\frac{\cos{\theta}}{\eta}$};
\node [right,red] at (1+0.25*0.6, 1.28) {$\tilde\classifier_A$};

\draw[blue, ultra thick] (0.9,3/4*1.15+1/4) arc (180:233.1:0.25);
\draw[blue, ultra  thick] (1.15,3/4*1.15+1/4) -- ++(180:0.25);
\draw[blue, ultra  thick] (1.15,3/4*1.15+1/4) -- ++(233.1:0.25) ;

\draw [blue, ultra  thick] (1.15,3/4*1.15+1/4-5/16) -- (1.15-0.25*0.6,3/4*1.15+1/4-0.25*0.8);

\draw [blue, ultra thick] (1.15,1.15*0.75+0.25) -- (1.15+0.25*0.6,3/4*1.15+1/4-0.25*0.8);

\draw [blue, ultra thick] (1.15,3/4*1.15+1/4-5/16) -- (1.15+0.25*0.6,3/4*1.15+1/4-0.25*0.8);

\draw [densely dashed]  (1,1) -- (1.15,3/4*1.15+1/4);
\draw [ultra thick, green]  (1,1) -- (1.15,3/4*1.15+1/4);

\end{tikzpicture}
\caption{Feasible mechanism: tests $\tilde\classifier_A,\classifier_B$} 
\label{fig:feasible fixed improvement}
  \end{subfigure}
\caption{One feasible sequential mechanism with fixed order}

\end{figure}
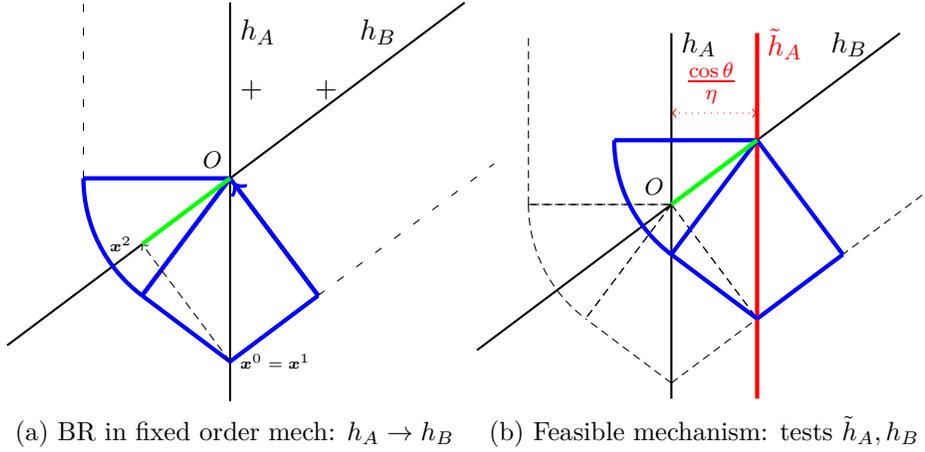

Moreover, a random-order mechanism that discloses the first test result to the agent is usually worse than some fixed order mechanism.
We formally prove this observation in \cref{prop:investment fixed order better than informed}.

What happens when $\theta\leq 30^{\circ}$?
When $\theta< 30^{\circ}$, depending on the distribution, both random-order mechanisms without disclosure and fixed-order mechanisms could be optimal.
To provide intuition, consider the mechanism $(\classifier_A,\classifier_B,\frac12,\nullset)$ and $\theta<30^{\circ}$ as in  \cref{fig:theta<30 random}.
$\theta<30^{\circ}$ can be interpreted as the two requirements set by the principal are significantly in conflict with each other .
Hence there exists some attributes that  satisfy neither $\classifier_A$ and $\classifier_B$ would find it very costly to satisfy both tests simultaneously, but the small angle allows them to use a two-step strategy in the random-order mechanism without disclosure.
For example, $\orifeatures$ in  \cref{fig:theta<30 random}.
However, in the investment setting, when these attributes use a two-step strategy, their eventual attributes are still unqualified under non-stringent tests $\classifier_A$ and $\classifier_B$.

To make this mechanism feasible, two stringent tests are required.
For example, the two tests $\classifier_A^*$ and $\classifier_B^*$ shifted by a distance of $\frac{1}{2\eta}$ as in \cref{fig:theta<30 random}.
This incurs a loss: some unqualified attributes that could have improved to the qualified region would no longer be accepted under the feasible random-order mechanism without disclosure.
In contrast, as we have argued above, a feasible fixed-order mechanism only requires using one stringent test (See \cref{fig:feasible fixed improvement}).
After using stringent tests, the set of agent that will be accepted under a fixed-order mechanism and that under a random-order mechanism without disclosure no longer have a nested structure.
Hence depending on the distribution, either could be better.

       \begin{figure}[t]
\centering
\begin{tikzpicture}[xscale=4.3,yscale=4.3,
    pics/legend entry/.style={code={%
        \draw[pic actions] 
        (-0.25,0.25) -- (0.25,0.25);}}]]

\draw [domain=0.85:1.26, thick] plot (\x, {tan(deg(0.4*pi))*(\x-1)+1});
\draw [thick] (1,0.54) -- (1,1.8);

\draw [domain={1+0.25*cos(deg(0.1*pi))}:1.5, loosely dashed] plot (\x, {tan(deg(0.4*pi))*(\x-1)+1-0.809}); 
\draw [loosely dashed] (1-0.25,1) -- (1-0.25,1.76);
\draw[red, ultra thick] (1-0.25,1) arc (180:240:0.25);
\draw[red, ultra thick] (1,1) -- ++(180:0.25);

\draw[red, ultra thick] ({1+0.25*cos(deg(0.1*pi))},{1-0.25*sin(deg(0.1*pi))}) arc (360-18:360-78:0.25);
\draw[red, ultra thick] (1,1) -- ++(360-18:0.25);

\draw [red, ultra thick] (1-0.125,{1-0.125*sqrt(3)}) -- (1-0.125,0.61625);
\draw [domain=0.875:0.957, red, ultra thick]  plot(\x,{tan(deg(0.3*pi))*(\x-0.875)+0.61625});

\draw [domain=0.956:1,red, ultra thick]  plot(\x,{-tan(deg(0.4*pi))*(\x-1)+1-0.125/sin(deg(0.1*pi))});
\draw [domain=1:1.053,red, ultra thick]  plot(\x,{tan(deg(0.4*pi))*(\x-1)+1-0.125/sin(deg(0.1*pi))});

\draw [blue, ultra thick] (1.125,0.9) -- (1.125,2.2);

\draw [domain=0.85:1.26, ultra thick, blue] plot (\x, {tan(deg(0.4*pi))*(\x-1)+1+0.4045});
\node [right, blue] at (1.16, 1.85) {$\classifier_B^*$};
\node [left, blue] at (1.13, 1.85 ) {$\classifier_A^*$};

\node [left] at (1-0.125,0.61625){\footnotesize$\orifeatures$};
\node [black] at (1-0.125,0.61625) {\textbullet};

\node [left] at (1,1.45){\footnotesize$\widetilde\orifeatures$};
\node [black] at (1,{1+0.4045}) {\textbullet};
\end{tikzpicture}
\caption{$\theta< 30^{\circ}$: random order mechanism without disclosure} \label{fig:theta<30 random}  
\end{figure}
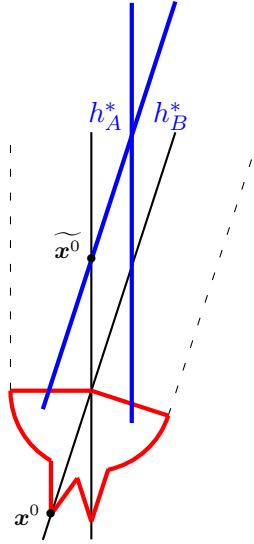

\subsection{Comparison to simultaneous mechanisms}
Next, we ask: do simultaneous mechanisms work better than sequential mechanisms in the investment setting?
The answer is yes.

\begin{proposition}\label{thm:optimal investment}
    For any distribution $\dist$ and any cost function $\cost$, the optimal simultaneous  mechanism uses two tests that coincide with the true requirement: $\classifier_A$ and $\classifier_B$. Moreover, it achieves the first best.
\end{proposition}

The optimal simultaneous mechanism achieves the first best.
This is because the agent is forced to pass both tests together under a simultaneous mechanism.
When the tests coincide with the true criteria, every selected agent is qualified or becomes qualified.

The simultaneous procedure can also be understood as the most stringent procedure because the agent has to pass both tests together.
Using stringent tests does not help in the investment setting, but using a stringent procedure encourages most investment.

\section{Robustness}

\subsection{Robustness to tests and qualified region}\label{sec: perfect tests}

In this subsection, we assume that the principal only uses perfect tests, i.e., for any agent with attributes $\features$, a perfect test outputs $\features$.
This is motivated by the literature on persuasion with evidence \citep{glazer2004optimal,sher2014persuasion}.
In the persuasion games with evidence, (1) the agent's cost of changing attributes is infinitely high, and (2) the principal can use \emph{perfect} tests, i.e., principal can observe agent's attributes in each round of communication.
The main difference is that in our manipulation setting, the agent can fabricate their type at a cost.

All results in this subsection is derived under the manipulation setting and program \ref{max qualified}.
All proof is omitted from this section and provided in \cref{appendix: perfect tests}.

First we show that when the principal can use perfect tests and an upfront cheap-talk communication is feasible, she can achieve the first best.

\begin{proposition}\label{thm: sequential perfect tests}
Suppose $M=\Featurespace$.
     For any distribution $\dist$, there exists a feasible fixed-order sequential mechanism that achieves the first best.   
     Moreover, this mechanism is adaptive, i.e., the second test depends on the first test.
\end{proposition}

The next corollary points out that a perfect test can convey information as in a cheap-talk communication. 

\begin{corollary}\label{cor: sequential perfect tests}
Suppose $M=\nullset$.
     Suppose the principal lets the agent to choose the first test, based on which the principal chooses the second test. For any distribution $\dist$, there exists a feasible fixed-order sequential mechanism that achieves the first best.   
\end{corollary}

 \cref{cor: sequential perfect tests} is a direct implication of \cref{thm: sequential perfect tests}. Under perfect tests, when the principal allows the agent to choose the first test, every agent chooses the first test that coincides with his true type.
 Based on the first test, the principal 
 can implement the same mechanism constructed in \cref{thm: sequential perfect tests}.
 
Next, we show that \cref{thm: sequential perfect tests} is not confined to the specific shape of the qualified region. As a matter of fact, as long as the qualified region is convex, the principal achieve the first best by using a sequential mechanism that offers tests in an adaptive manner.

\begin{proposition}\label{thm: perfect tests arbitrary qualified region}
When the (true) qualified region $\qualregion$ is convex,  a sequential mechanism that uses two tests  achieves the first best.
\end{proposition}

Finally, we show that under a condition on the shape of the qualified region, a single-test mechanism, an analogy to a simultaneous mechanism, is strictly dominated by the above described sequential mechanism.

\begin{proposition}\label{thm: single-test mechanism perfect test}
     There exists a single-test mechanism that achieves the first best if and only if the qualified region $\qualregion$ is the convex hull of the union of balls with radius $1/\mc$.
\end{proposition}

\cref{thm: single-test mechanism perfect test} implies that whenever there is a sharp corner (for instance, $\qualregion=\classifier_A\cap \classifier_B$, where $\classifier_i, i\in \{A,B\}$ is a half plane) in the true qualified region, a single-test mechanism cannot accept every qualified agent without accepting any unqualified agent. 
In other words, the best sequential mechanism (with two tests) strictly out performs a single-test mechanism under objective \ref{max qualified} and when effort is manipulation.
Only in those cases where the true qualified region has only round corner(s),  a single-test mechanism is equally good as the best sequential mechanism (with two tests).
Such examples include circles, rectangles with round corners, etc..
\subsection{Robustness to cost functions}\label{subsec: seq general cost}

In this subsection, we consider a larger class of cost functions. 

\begin{assumption}[translation invariance]\label{def: translation invariant}
    The cost function $c : \bbR^6 \to \bbR$ is \emph{translation invariant} if for any attributes $\orifeatures, \firstfeatures,\secondfeatures$ and any shift vector $\genericfeatures\in \bbR^2$, $c(\orifeatures, \firstfeatures,\secondfeatures) = c(\orifeatures+\genericfeatures, \firstfeatures+\genericfeatures,\secondfeatures+\genericfeatures)$.
\end{assumption}

\begin{assumption}[absolute homogeneity]\label{def: absolute homogenous c}
    The cost function $c : \bbR^6 \to \bbR$ is \emph{absolute homogeneous} if for any attributes $\orifeatures, \firstfeatures,\secondfeatures$, and any $\alpha\in \bbR$, $c(\alpha\orifeatures, \alpha\firstfeatures,\alpha\secondfeatures) = |\alpha|c(\orifeatures, \firstfeatures,\secondfeatures)$.
\end{assumption}
This assumption can also be easily extended to homogeneity of any degree $m$ for $m\in \bbR_+$, i.e., $c(\alpha\orifeatures, \alpha\firstfeatures,\alpha\secondfeatures) = |\alpha|^m c(\orifeatures, \firstfeatures,\secondfeatures)$.

\begin{assumption}[triangle inequality]\label{def: trainagle inequality}
    The cost function $c : \bbR^6 \to \bbR$ satisfies triangle inequality if for any attributes $\features,  \boldsymbol{y}, \genericfeatures$, 
    $$\cost(\features, \boldsymbol{y}, \genericfeatures)\geq \onecost(\features,  \genericfeatures).$$ 
\end{assumption}

\begin{assumption}[monotonicity]\label{def: monotone}
    The cost function $c : \bbR^6 \to \bbR$ is monotone if for any attributes $\features, \boldsymbol{y}, \genericfeatures$, 
    $$\cost(\features, \boldsymbol{y}, \genericfeatures)\geq \onecost(\features,  \boldsymbol{y}).$$ 
\end{assumption}

\begin{assumption}[regular]\label{def: existence of minimum cost}
    The cost function $\onecost: \bbR^4 \to \bbR$ is regular if for any attributes $\features$ and any half plane $\classifier$, there exists some $\boldsymbol{y}\in \classifier$ such that 
    $$\onecost(\features, \boldsymbol{y})=\inf_{\genericfeatures\in \classifier}\onecost(\features,\genericfeatures).$$    
\end{assumption}

If the cost function only depends on the differences between the $i$-th attributes and $(i-1)$-th attributes, for $i\in \{1,2\}$, i.e., 
$\cost(\orifeatures,\firstfeatures,\secondfeatures)=c(\firstfeatures-\orifeatures,\secondfeatures -\orifeatures, \secondfeatures-\firstfeatures)$, then the cost function is translation invariant.

Triangle inequality is reminiscent of the upward triangle inequality in \citet{perez2022test}.
The main difference is that in \citet{perez2022test}, the three attributes lie on the same line, while here, the three attributes lie on the same plane.
It means that the cost of changing from any initial attributes $\orifeatures$ to any $\secondfeatures$ would increase if the agent takes an extra middle step.
This is consistent with our applications where changing attributes require either physical effort or forgone profitability.

Monotonicity rules out transient cost or negative cost. 
Consider a scenario where the agent has to purchase extra equipment to pass the first test but afterwards he is free to sell the equipment.
This would result in $\cost(\orifeatures, \firstfeatures,\secondfeatures)<\onecost(\orifeatures, \firstfeatures)$, which violates monotonicity.
This is a reasonable condition in the manipulation setting. 
In the hiring example, $\firstfeatures$ and $\secondfeatures$ represent how the profile of a candidate appears to be in a job interview.
The cost of changing the candidate's profile represents the cost of actual effort.
Once the effort is exerted, it is sunk.
In the bank regulation example, $\firstfeatures$ and $\secondfeatures$ represent the balance sheet positions of a bank.
Adjusting the balance sheet usually incurs transaction costs.

We cover the extreme case where the cost of modifying attributes is infinity, i.e., the hard evidence environment. In this extreme case, no agent is able to change their attributes. Our results show that the fixed-order mechanism with the original two tests $\classifier_A$ and $\classifier_B$ is still optimal. Moreover, it accepts all qualified agents and no unqualified agent.

Here are two classes of additive cost functions that satisfy \cref{def: translation invariant}-\ref{def: existence of minimum cost}.
\begin{itemize}
    \item  $
c(\orifeatures,\firstfeatures,\secondfeatures)= \onecost(\orifeatures,\firstfeatures) + \onecost(\firstfeatures,\secondfeatures) 
$, where $\onecost(\cdot,\cdot)$ is a metric.
    \item $
\cost(\orifeatures,\firstfeatures,\secondfeatures) =\onecost(\orifeatures,\firstfeatures) + \onecost(\orifeatures,\secondfeatures),
 $ where  $\onecost(\cdot,\cdot)$ is a metric.
\end{itemize}

\subsubsection{Manipulation}\label{subsubsec: general cost manipulation}

\begin{proposition}\label{thm:opt_manipulation}
    Suppose the cost function $c$ satisfies \cref{assump: cost function one step}-\ref{def: existence of minimum cost}.
    Consider manipulation setting.
    For any distribution $\dist$, there exists a fixed-order sequential mechanism that outperforms the optimal simultaneous mechanism. Moreover, this fixed-order sequential mechanism 
    uses two stringent tests $\widetilde\classifier_A, \widetilde\classifier_B$, i.e., $\widetilde\classifier_A\cap \widetilde\classifier_B \subsetneq \classifier_A\cap \classifier_B$
\end{proposition}

To prove this theorem, we first pin down the optimal tests under that the optimal simultaneous mechanism in \cref{prop:optimal simultaneous manipulation general cost}.
Second, we show that the fixed-order mechanisms that use the same tests as in the optimal simultaneous mechanism are also feasible 
 and select more potential qualified agents than the optimal simultaneous mechanism (\cref{lem:fix-simul}).
These two together imply that the optimal mechanism must be sequential.
Third, we show that any feasible mechanism must use stringent tests (\cref{lem:stringent tests}).
Hence we conclude that the optimal mechanism has the described properties.
All omitted proofs in this subsection can be found in \cref{appendix:general costs}.

 \begin{proposition}\label{prop:optimal simultaneous manipulation general cost}
 Consider the manipulation setting.
    Suppose the one-step cost function $\onecost : \bbR^4 \to \bbR$  satisfies \cref{def: translation invariant} ( translation invariant),\cref{def: absolute homogenous c} (absolute homogeneous) and \cref{def: existence of minimum cost}.
    For any distribution $\dist$, the optimal simultaneous mechanism uses two stringent tests $\classifier_A^{+},\classifier_B^{+}$ where $\classifier_i^{+}\subset \classifier_i$ for $i\in \{A,B\}$.
    \end{proposition}

Next, we compare the optimal simultaneous mechanisms and the fixed-order mechanisms that use the same tests. 
\begin{lemma}\label{lem:fix-simul}
    Given the optimal simultaneous mechanism  $(\classifier_A^{+},\classifier_B^{+})$, the two  fixed-order mechanisms $(\classifier_A^{+},\classifier_B^{+},1)$ and $(\classifier_A^{+},\classifier_B^{+},0)$ are also \emph{feasible} and  are weakly better than the simultaneous mechanism. 
\end{lemma}

\begin{lemma}\label{lem:stringent tests}
    Any feasible mechanism uses stringent tests, i.e., the two tests the mechanism announces $\tilde\classifier_A$ and $\tilde\classifier_B$  satisfies $\tilde\classifier_A\cap\tilde\classifier_B\subsetneq \classifier_A\cap\classifier_B$.
\end{lemma}

\subsubsection{Investment}\label{subsubsec:investment}
We have shown that in the investment setting, the simultaneous mechanism using the true criteria as the tests achieves the first best (\cref{thm:optimal investment}). 
This is true for any cost function.
Despite its power, we continue to study the optimal sequential mechanism under investment for two reasons.
First, practically speaking, it is sometimes hard to implement simultaneous mechanisms because of physical constraints.
Second, it helps us to understand better how the two classes of mechanisms work under investment.
In these scenarios, it is useful to know what the optimal sequential mechanism is like.

Therefore, in this subsection, we generalize the result on optimal sequential mechanism in the investment setting (\cref{thm: true effort sequential}) to a broader class of cost functions. 
All proofs are omitted in this subsection and will be provided in \cref{appendix: seq general cost}.

In general, analyzing sequential mechanism under investment is more challenging than under manipulation setting.
As in the manipulation setting, we first need to characterize the agent's best response under a given mechanism.
Moreover, under investment, we also need to keep track of the eventual attributes the agent adopts, in order to decide whether the agent becomes qualified.
The second step can easily get very complicated even when there are nice topological properties coming from the cost function.

Nonetheless, in this subsection, we first provide a condition (\cref{condition:one-step}) under which the agent behaves in the same way under random-order mechanism and under simultaneous mechanism, implying that any simultaneous mechanism can be implemented by some random-order mechanism.

The following condition characterizes any agent's behavior, i.e., best response, in a mechanism.
\begin{conditionp}{O}\label{condition:one-step}
     Every agent prefers some one-step strategy to any two-step strategy.
\end{conditionp}

This condition holds trivially in any simultaneous mechanism.
However, it is not trivial to check when this condition holds in a sequential mechanism.
In general, whether this condition holds in a sequential mechanism is linked to  the cost function and properties of the mechanism.
We will provide examples later to illustrate such a connection.
Equivalently, this condition can be rephrased as the following.

\begin{conditionp}{O'}\label{condition:one-step'}
     No agent prefers using any two-step strategy.
\end{conditionp}

\begin{proposition}\label{lem: optimal max qualified improving effort}
    Consider  the investment setting and program \ref{max qualified}.
    Suppose the cost function satisfies \cref{assump: cost function one step} and \ref{def: translation invariant} (translation invariant).   
    For any distribution $\dist$, if there exists a feasible  random-order mechanism without disclosure $(\tilde \classifier_A, \tilde \classifier_B,q,\nullset)$ in which \cref{condition:one-step} holds, then the optimal sequential mechanism is an  random-order mechanism without disclosure using two tests $\classifier_A,\classifier_B$ and the same randomization probability, i.e., $( \classifier_A,  \classifier_B,q,\nullset)$. 
    
    Moreover, it achives the first best.
\end{proposition}

Although this condition holds for a relatively large class of cost functions, it is not easy to interpret using primitives unless we are equipped with a concrete cost function.
The existence of a feasible  random-order mechanism without disclosure in which \cref{condition:one-step} holds usually depends on: (1) the cost function,  the mechanism, including (2) the angle between the two tests (usually linked to the angle of the qualified region), and (3) the randomization probability and the disclosure decision.
However, to characterize the condition on cost function and the properties of the mechanism that guarantees \cref{condition:one-step} holds is very challenging.
This is because we usually need to completely characterize agent's best response, which requires finding and comparing the best one-step and two-step strategy.
This becomes very difficult when there are not enough topological properties being imposed on the cost function.
Hence, we only provide examples under the additive Euclidean cost function to illustrate the connection between the primitives and \cref{condition:one-step}.

 Under this cost function, there always exists a feasible uninformed random order mechanism in which \cref{condition:one-step} holds if and only if $\theta\geq 30^{\circ}$.
From the previous analysis, we can see that any two-step strategy requires the agent to first adopt attributes that satisfy one test but failing another and later adopt different attributes that satisfy the second test but failing the first.
Holding other primitives fixed, such two-step strategies become less appealing when the randomization probability over the order of the tests is closer to $0.5$ because the randomization creates strategic uncertainty, i.e., it is harder for the agent to predict which test is more likely to be the first one.
Moreover, the no disclosure decision also makes two-step strategies less appealing.
In any sequential mechanism with randomization probability $0.5$ and no disclosure, such strategic uncertainty is maximized.

Holding the randomization probability and disclosure decision fixed, when the two tests are far away from each other, two-step strategy becomes costly and could be less appealing than some one-step strategy by adopting some attributes that satisfy both tests at the same time.
Intuitively, the two tests being far away from each other is measured by both the angle between the two tests and the cost function, which is being held fixed here.
 Recall that $\theta$ is the angle of the qualified region.
$\theta\geq30^{\circ}$ can therefore be understood as the two requirements $\classifier_A$ and $\classifier_B$ are far away from each other. 
When $ \theta\geq 30^{\circ}$, we can easily compute that under the random order mechanism with randomization probability $\probprincipal=0.5$ and no disclosure, no agent prefers two-step strategies.
However, when $\theta< 30^{\circ}$, some agent prefers using some two-step strategy even under the uninformed random order mechanism with randomization probability $\probprincipal=0.5$.

In the remaining section, we impose stronger properties on the cost function to derive similar sharper characterization of the optimal sequential mechanism as in the Euclidean distance cost. We assume that the cost function is additive.

\begin{assumption}[additive]\label{def: additive}
A cost function $c(\cdot,\cdot,\cdot)$ is additive if 
    for any initial attributes $\orifeatures$ and any two-step strategy $\strategies=(\firstfeatures, \secondfeatures)$,
$$
c(\orifeatures,\firstfeatures, \secondfeatures) = \onecost(\orifeatures,\firstfeatures) + \onecost(\firstfeatures, \secondfeatures).
$$
\end{assumption}




\begin{theorem}\label{thm:optimal sequential metric investement}
     Consider  the investment setting and program \ref{max qualified}.
     Suppose the cost function satisfies \cref{assump: cost function one step}-\ref{def: additive}.  
     For any given $\classifier_A$, there exists an angle $\theta^*(\classifier_A) \in (0, 180^{\circ})$ such that when the angle between the two half planes $\classifier_A$ and $\classifier_B$ is greater than $\theta^*(\classifier_A)$, for 
     any distribution $\dist$,    
    the optimal sequential mechanism is an random-order mechanism without disclosure with two tests $\classifier_A$ and $\classifier_B$. 
    
    Moreover, it achieves the first best.
\end{theorem}

The next result shows that in the investment setting, random-order mechanisms with disclosure is never optimal. The proof uses the same mixed mechanism argument as before.

\begin{proposition}\label{prop:investment fixed order better than informed}
    Consider  the investment setting and program \ref{max qualified}.
     Suppose the cost function satisfies \cref{assump: cost function one step}-\ref{def: additive}.
    For any distribution $\dist$, the optimal fixed-order mechanism is weakly better than the optimal random-order mechanism with disclosure.
\end{proposition}

\subsection{Robustness to objective}\label{sec:Pii}
In this subsection, we consider the following alternative objective.
\begin{equation}\tag{$\mathcal{P}_{II}$}\label{min unqualified}
    \begin{aligned}
        \min & \quad \pr[\text{selecting unqualified agent}] \\
        s.t. & \quad \text{ selecting all qualified agent}.\\
    \end{aligned}
\end{equation}

\paragraph{Manipulation}
Let $\density$ denote the density function of the distribution over initial type $\dist$.
Let $\nabla \density$ denote the gradient vector of $\density$.

\begin{proposition}\label{prop:Pii manipulation}
    Consider the manipulation setting and program \ref{min unqualified}.
    Suppose the cost function is induced by the Euclidean distance and is additive.
    Suppose $\nabla \density \cdot \weights_A \leq 0$ and $\nabla \density \cdot \weights_B \geq 0$.
    Then there exists a fixed-order sequential mechanism that outperforms the optimal simultaneous mechanism.
\end{proposition}

The alternative objective complicates the analysis under manipulation.
In general, which mechanism is optimal depends on the distribution.
Here we provide a relatively simple and easy-to-interpret condition $\nabla \density \cdot \weights_A \leq 0$ and $\nabla \density \cdot \weights_B \geq 0$, under which we can show a fixed-order mechanism dominates all simultaneous mechanisms. 
This condition quantifies the tradeoff between the two criteria.
It can be understood as the more qualified type according to one criterion (along 
the direction of $\weights_B$) has decreasing density while the more qualified type according to another criterion (along 
the direction of $\weights_A$) has increasing density. 
Under this condition, we can easily pin down the optimal tests under the optimal simultaneous mechanism.
Then, we can construct a fixed-order sequential mechanism that dominates it.
We defer the proof to \cref{appendix: alternative objective}.

\paragraph{Investment} The same simultaneous mechanism continues to achieve the first best under the alternative objective in the investment setting. 
\begin{proposition}\label{prop:Pii investment}
    Consider  the investment setting and program \ref{min unqualified}.
    For any distribution $\dist$ and any cost function, the optimal  simultaneous mechanism is one that uses the true requirements and  achieves the first best.
\end{proposition}

The proof of \cref{prop:Pii investment} is very similar to the proof of \cref{thm:optimal investment}.

\section{Other applications}\label{sec:applications}
\paragraph{Collective decision making with diverging opinions}

Consider a search committee deciding whether to hire a candidate.
To be more concrete, consider an academic department hiring a scientific researcher, where candidates usually need to present their research plan for the next five to ten years and supporting evidence that such a plan is feasible and promising, on top of past work to demonstrate competence. 
Each member in the committee cares both about the objective aspect of the candidate, i.e., competence, and the subjective aspect of the candidate, i.e., how risky or ambitious his research plan is.
Everyone likes to hire a competent candidate, but some members prefer more ambitious candidates while others prefer less ambitious candidates.
Consider an unanimous or equal veto decision rule. 
We can view the whole committee as the principal.
And the committee's qualified region is defined by the two members with the most extreme subjective preferences, say Mr. R who prefers ambitious candidates and Mr. L who prefers less ambitious candidates (See \cref{fig: search committee}).

\begin{figure}[h] 
\centering 
\begin{tikzpicture}[xscale=7.8,yscale=7.8]
\draw [thick,<->] (-0.1,0) -- (0.8,0);
\draw [thick, ->] (6/25,0) -- (6/25,0.8);
\node [below] at (0.8,0) {more ambitious};
\node [below] at (-0.2,0) {less ambitious};
\node [above] at (6/25,0.8) {competence};

\draw [domain=0:8/15, thick, red] plot (\x, {1.5*\x});
\draw [domain=0:0.6, thick, blue] plot (\x, {-\x +0.6});

\path[pattern color=yellow,pattern=north west lines] (0,0.8) -- (0,0.6) -- (6/25,6/25*1.5)-- (8/15,0.8); 
\node [thick,font=\tiny] at (0.1,0.65) {\footnotesize$\text{approval zone}$};

\fill [red!60,nearly transparent]  (8/15,0.8) -- (0,0) -- (8/15,0) -- cycle;
\node [thick,font=\tiny, red] at (0.53,0.3) {\footnotesize$\text{Mr. L's rejection zone}$};

\fill [blue!60,nearly transparent]  (0,0.6) -- (0,0) -- (0.6,0)-- cycle;
\node [thick,font=\tiny, blue] at (0,0.25) {\footnotesize$\text{Mr. R's rejection zone}$};

\end{tikzpicture}
\caption{A search committee with diverging preferences}
\label{fig: search committee}
  \end{figure}
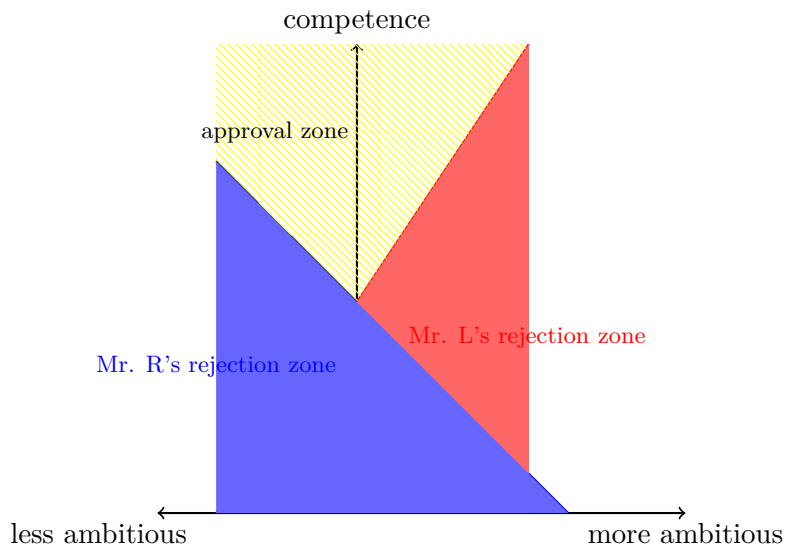
  
A test is  a joint assessment on the candidate's multiple attributes by each member, which could be a conversation, a discussion, a speech, a presentation, or an interview.
Here only the two members with the most extreme preferences matter, i.e., Mr. L and Mr. R.
A candidate passes Mr. L's (or Mr. R's) test if he passes the overall assessment or judgement made by Mr. L (or Mr. R). 

 A simultaneous mechanism is an interview with both Mr. L and Mr. R in the same room.
 A sequential mechanism with fixed order consists of  two separate interviews Mr. L and Mr. R in a  fixed order.

The candidate could inflate or deflate his attitude towards taking risk when talking to different interviewers.
Hence this fits into the manipulation setting.
The cost of changing apparent performance is arguably path dependent for multiple reasons.
A candidate usually needs to provide resumes and personal statement that are assessed by all interviewers. This could constrain how far the candidate can misrepresent themselves.

\paragraph{Banking Regulation}
After the recent global financial crisis in 2008, one major responsibility of central bank is to stabilize the banking system. 
The central bank sets multiple regulatory objectives based on commercial banks' financial statements. 
However, some of these metrics are in conflict with each other.
Consider the leverage ratio and liquidity ratio.
The leverage ratio is the ratio of Tier 1 capital to total asset.
The liquidity ratio is the ratio of liquid asset to net cash outflows over a 30-day stress period.
These two ratios are interconnected through liquid asset.
There are two reasons for the interconnectedness of the two ratios.
First, the total asset is the sum of liquid asset and illiquid asset.
Second, liquid asset is the item that is easiest to change in a short period of time among the items appearing on the bank's financial statements.

The central bank publicly sets the minimum requirement on leverage ratio and liquidity ratio. To visualize the central banks' requirements, we project the requirements on the two ratios on the space of Tier 1 capital and liquid asset. 
Suppose the minimum requirement on leverage ratio is $3\%$. 
Then the minimum requirement on leverage ratio is equivalent to requiring $\text{Tier 1 capital }\geq 3\% \times \text{ liquid asset } + 3\% \times\text{ other asset}$, which is represented by the red line in \cref{fig: banking regulation}.
The liquidity ratio requirement equivalent to requiring $\text{ liquid asset }\geq \text{ net cash outflow }$, which is represented by the blue line in \cref{fig: banking regulation}.
Notice that both of the requirements are bank specific.
Hence \cref{fig: banking regulation} represents the central bank's regulatory objective to one bank.
In Europe, the Basel committee on Banking Supervision (BCBS) decides that commercial banks need to report their leverage ratio at the quarter end and report their liquidity ratio at the end of each month.
In contrast, in the US, big commercial banks are required to report both their leverage ratio and liquidity ratio daily.

We argue that 
the European regulatory framework is essentially a sequential one with raised standards,\footnote{Encountering European banks' `window dressing' behaviors, European central bank decided to raise regulatory requirements on selected banks. In particular, `for six banks, a P2R leverage ratio add-on was applied on top of the 3\% leverage ratio requirement.' ( \citet{ECBraiserequirement})} while the 
US regulatory framework is essentially a simultaneous procedure. 
We further argue that both practices could be compatible to our model if the European central bank prioritizes reducing gaming behaviors from commercial banks and the US central bank prioritizes boosting investment from commercial banks.\footnote{An alternative possibility is that European central bank believes that  commercial banks are primarily  gaming, while the US central bank believes that  commercial banks are primarily investing.}

\begin{figure}[h] 
\centering 
\begin{tikzpicture}[xscale=7.8,yscale=7.8]
\draw [thick,->] (-0.1,0) -- (0.8,0);
\draw [thick, ->] (0,-0.1) -- (0,0.8);
\node [below] at (0.86,0) {liquid asset};
\node [above] at (0,0.8) {Tier 1 capital};

\draw [domain=0:0.8, thick, red] plot (\x, {0.2*\x+0.3});
\draw [ thick, blue] (0.5,0) --(0.5,0.8);

\node [below] at (0.46,0) {net cash outflow};
\node [right] at (0.5,0.8) {liquidity ratio};
\node [right] at (0.8, 0.46) {leverage ratio };
\node [left,font=\tiny] at (0,-0.05) {\footnotesize$O$};
\path[pattern color=purple,pattern=north west lines] (0.8,0.46)--(0.5,0.4)--(0.5,0.8); 

\end{tikzpicture}
\caption{Regulation on a commercial bank}
\label{fig: banking regulation}
  \end{figure}
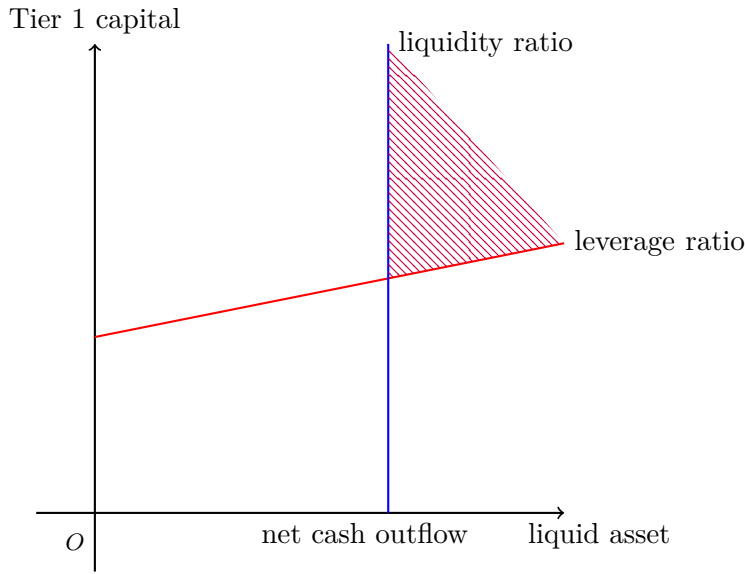

In practice, more European banks are reported to manipulate their financial statements by temporarily changing their liquid asset holding in a short period of time. We consider such temporary change manipulation.
In contrast, US banks are not detected to be involved into this kind of short term manipulation behaviors. 
Instead, when US banks change their financial statements, they tend to maintain such a change for a long time. We consider such long term maintenance of change an investment.

Empirical evidence shows that under the current regulatory framework, European banks are involved into window dressing behaviors by contracting repurchase agreement transactions right before the quarter end, and expand repurchase agreement  transactions right after (\citet{window_dressing} and \citet{ECB_window}).
A repurchase agreement usually involves two parties: a borrower of liquidity and a lender.
From the perspective of a borrow, the agreement is referred to as a repo while from the perspective of a lender, the same agreement is called a reverse repo.
A repo transaction could impact the balance sheet of the borrower while a reverse repo transaction has no impact on the lender's balance sheet.
When a bank enters a repo transaction, they receive cash, which enters the total asset and hence decreases the leverage ratio.
Hence contracting repo transactions can  temporarily boost a bank's leverage ratio, while expanding repo transactions enables banks to borrow liquidity, which could enhance their liquidity ratio.

In the US, banks are required to report their daily and quarterly average leverage (liquidity) ratio.
This type of practice prevents banks from temporarily improving their leverage (liquidity) ratio.

\paragraph{Regulation on merger or acquisition approval}
For a merger or acquisition to happen, the proposed firm needs to first submit a public filing to the Securities and Exchange Commission (SEC) and disclose information to investors so that they can evaluate the profitability of the merger or acquisition transaction.
SEC's regulatory goal is to inform shareholders about the potential profitability consequence of the merger or acquisition.
Hence the proposed firm usually need to justify that the merger or acquisition will increase profit, either because it introduces production efficiency by reducing cost significantly or because it enables the firm to set higher prices through larger market power.\footnote{The SEC mandates that public companies disclose detailed financial information related to M\&A transactions to ensure transparency and protect investors. These disclosures often include financial metrics such as projections of future performance, required returns, and cost of capital \citep{SEC_Regulation_SK, SEC_Form_S4}. In M\&A transactions, financial decisions are often based on the comparison between the expected return and the cost of capital (the hurdle rate). This financial analysis determines whether the deal adds value to the acquiring company’s shareholders \citep{jensen1983market}.}
Notice that although increasing market power is a sever anti-competitiveness concern, it is usually not an issue with the SEC since the SEC only cares about the investors.
The regulation from the SEC is represented by the blue line in \cref{fig: FTC-SEC}.
Basically, SEC requires that a company needs to reduce costs and/or increase revenue in some way above a hurdle rate.

After the firm's proposal gets approved by the SEC, the firm needs to further justify to the Federal Trade Commission (FTC) that the merger or acquisition does not trigger antitrust intervention. 
The FTC's regulatory goal is in conflict with that of the SEC. 
FTC usually approves a merger or an acquisition if the increase in market power is offset by the cost reduction, or improved production efficiency, so that overall prices do not actually increase. 
The FTC's regulatory goal is represented by the red line in \cref{fig: FTC-SEC}.
Depending on how much costs are  passed through to the consumers, the slope of the red line could be steeper or flatter.
The region where a merger or an acquisition that gets approved (qualified region) is colored in yellow in \cref{fig: FTC-SEC}.

A test chosen by the SEC is usually a public filing where the firm needs to justify the benefits of the merger or acquisition to the shareholders,
while a test chosen by the FTC is usually a threshold on how much estimated increase in market power  compared to the estimated  cost reduction can be tolerated, for example, price increase due to increased power cannot exceed $5\%$ of cost reduction.
 
Usually, firms could inflate the increased benefit from the merger or acquisition to the SEC to pass the SEC's scrutiny.
In contrast, when presenting to the FTC, firms could 
could exaggerate the uncertainty on potential price increase.
If firms are sued by the FTC, they could also hire experts to argue that the market is larger and hence  to deflate the potential increase of price or the harm to the consumers, or they could hire experts to manipulate cost information to argue that cost reduction is much higher than estimated. \footnote{Sometimes, firms can also hide their data or throw shade on the legitimacy of their data to influence the FTC's estimation.}
However, when using such strategies, they cannot contradict to their public filing submitted to the SEC. So the cost function of the firms exhibit some path dependence.

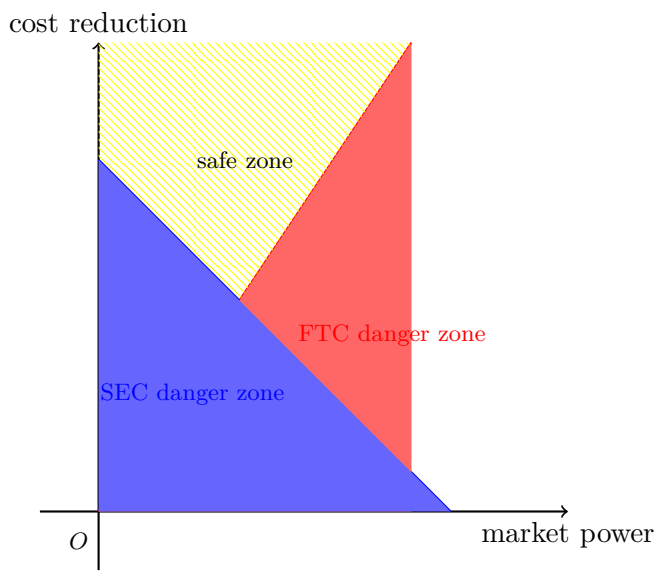
\begin{figure}[h] 
\centering 
\begin{tikzpicture}[xscale=7.8,yscale=7.8]
\draw [thick,->] (-0.1,0) -- (0.8,0);
\draw [thick, ->] (0,-0.1) -- (0,0.8);
\node [below] at (0.8,0) {market power};
\node [above] at (0,0.8) {cost reduction};

\draw [domain=0:8/15, thick, red] plot (\x, {1.5*\x});
\draw [domain=0:0.6, thick, blue] plot (\x, {-\x +0.6});

\node [left,font=\tiny] at (0,-0.05) {\footnotesize$O$};
\path[pattern color=yellow,pattern=north west lines] (0,0.8) -- (0,0.6) -- (6/25,6/25*1.5)-- (8/15,0.8); 
\node [thick,font=\tiny] at (0.25,0.6) {\footnotesize$\text{safe zone}$};

\fill [red!60,nearly transparent]  (8/15,0.8) -- (0,0) -- (8/15,0) -- cycle;
\node [thick,font=\tiny, red] at (0.5,0.3) {\footnotesize$\text{FTC danger zone}$};

\fill [blue!60,nearly transparent]  (0,0.6) -- (0,0) -- (0.6,0)-- cycle;
\node [thick,font=\tiny, blue] at (0.16,0.2) {\footnotesize$\text{SEC danger zone}$};

\end{tikzpicture}
\caption{Regulation on firms during merger and acquisition}
\label{fig: FTC-SEC}
  \end{figure}

 \section{Discussions}\label{sec:discussion}

\paragraph{Dimensionality and degenerate sequential mechanisms.}
Our model cannot be reduced to a problem where agent has a one-dimension attribute and the principal offers two selection cutoffs.
Consider a setup where the agent has a one-dimensional attribute $\feature\in \bbR$.
The principal would like to select any agent with attribute $\feature\in [b_1,b_2]$, where $b_1<b_2$.
Suppose now the principal announces that any agent with attribute belonging to $[\tilde b_1,\tilde b_2]$ will be selected, with $\tilde b_1<\tilde b_2$. 
The corresponding tests would be $\tilde\classifier_A: \feature\geq \tilde b_1$ and $\tilde\classifier_B: \feature\leq \tilde b_2$.
Then no matter what the order of the tests is, the agent's best response is fixed: for any $\feature < \tilde b_1$, $\strategies(\feature)=(\tilde b_1,\tilde b_1)$; for any $\feature > \tilde b_2$, $\strategies(\feature)=(\tilde b_2,\tilde b_2)$; for any $\tilde b_1\leq \feature \leq \tilde b_2$, $\strategies(\feature)=(\feature,\feature)$.

Next we introduce the concept of degenerate testing mechanisms. 
\begin{definition}[degenerate mechanism]
	A sequential mechanism $(\tilde\classifier_A,\tilde\classifier_B,\probprincipal,\disclose)$ is said to be \emph{degenerate} if A's best response is the same regardless of the order of the tests.
\end{definition}

We summarize the above observation in the following proposition.
\begin{proposition}
	If the agent's attribute is one-dimension, then any testing mechanism $(\tilde\classifier_A,\tilde\classifier_B,\probprincipal,\disclose)$ is degenerate.
\end{proposition}

\citet{perez2022test} studies a test design problem where the agent's attribute is one dimensional.
The above observation provides a motivation to study the case where agent has multi dimensional attributes and how the order of the tests play a role in such settings. 
As a starting point, we study the simplest case where agent has two dimensional attributes and principal has two tests.

\paragraph{Modeling tests.}
We interpret a test as an evaluation of whether one (or some) aspect(s) of the agent meets the the principal's requirement.
We assume that the agent does not which aspect(s) he is being evaluated in each test.
We argue that this is a reasonable assumption.
For example, in a math test, the students taking the test do not know whether they are being evaluated on their mathematical thinking skill or their mastery of mathematical tools.
This slightly differs from the way we use the word test in English.
When we say there is a math test, the math test is a subject that includes at least the above mentioned two \emph{aspects of ability} being evaluated.
Hence if we interpret a test as an evaluation of aspect(s) of ability, it is unlikely that the students know how each aspect of ability is evaluated (or graded) unless they are told.

\paragraph{Simultaneous mechanisms are not special case of sequential mechanisms.}
Although in simultaneous mechanisms the order of the tests does not matter, they cannot be viewed as special cases of sequential mechanisms, i.e.,  $\simultaneous\not\subset \sequential$.
First of all, in any game specified by a simultaneous mechanism, the agent only has one information node.
Therefore, a simultaneous mechanism cannot be achieved by any fixed order mechanism or random order mechanism with disclosure.
Second, even though in the game specified by a random order mechanism without disclosure, the agent also has one information node, his strategy space is much larger.
This is because in any simultaneous mechanism, the agent is essentially restricted to use \emph{one-step} strategies, while in any random order mechanism without disclosure, the agent can use either \emph{one-step} strategies or \emph{two-step} strategies. 
An alternative way to see this is that in any random order mechanism without disclosure, the agent is evaluated by a \emph{linear} test each time.
However, in any simultaneous mechanism, the agent is evaluated by the intersection of two \emph{linear} tests each time. 
More formally, let $\twolineartest=\{\tilde\classifier_1\cap\tilde\classifier_2:\tilde\classifier_i\in \lineartest, i=1,2\}$ be the space of tests that can be used in any simultaneous mechanism.
Then $\lineartest\subset \twolineartest$.
Therefore, by changing the timing to offer the tests, simultaneous mechanisms are essentially enlarging the space of tests that can be used.

\paragraph{Dynamic nature of sequential mechanisms.}
To further illustrate why the game induced by any sequential mechanism is not static, let's consider the cost function $\cost(\orifeatures,\firstfeatures,\secondfeatures)=\metric(\orifeatures,\firstfeatures)+\metric(\firstfeatures,\secondfeatures)$, where $\metric(\cdot,\cdot)$ is the Euclidean distance.
Consider two fixed order mechanisms that use the same two tests $\tilde\classifier_A$ and $\tilde\classifier_B$ and they only differ in terms of the order of the two tests.
In \cref{subfig:two step A-B}, the mechanism offers $\tilde\classifier_A$ as the first test, while in \cref{subfig:two step B-A}, the mechanism offers $\tilde\classifier_B$ as the first test.
Consider an agent with attributes $\orifeatures$ as in \cref{subfig:two step A-B}.
Notice that $\orifeatures$ pass test $\tilde\classifier_A$ but fail test $\tilde\classifier_B$.
Under the fixed order mechanism $(\tilde\classifier_A,\tilde\classifier_B,1)$, the optimal strategy of such attributes is to choose $\firstfeatures=\orifeatures$ and a different $\secondfeatures$ that are the projection of $\orifeatures$ on the boundary line of $\tilde\classifier_B$ (See \cref{subfig:two step A-B}).
Such a strategy is the least costly for the agent and guarantees that the agent is selected under this fixed order mechanism.
However, under another fixed order mechanism $(\tilde\classifier_B,\tilde\classifier_A,1)$, the optimal strategy of such attributes is to choose $\firstfeatures=\secondfeatures$ that are the intersection of  the boundary line of $\tilde\classifier_A$ and $\tilde\classifier_B$ (See \cref{subfig:two step A-B}).
This is because now $\tilde\classifier_B$ is the first test the agent needs to pass and $\orifeatures$  are so far away from  $\tilde\classifier_B$ that it is less costly to pass the two tests at the same time.
This example shows that a sequential mechanism is naturally dynamic and as a result, the agent's strategy is also dynamic.

\begin{figure}[h]  \label{fig: sequential is dynamic}
	\centering 
	\begin{subfigure}[b]{0.45\linewidth}
		\begin{tikzpicture}[xscale=7.8,yscale=7.8]
		
		\draw [domain=0.68:1.4, thick] plot (\x, {3/4*\x+1/4});
		\node [left] at (1.38, 1.3 ) {$\tilde \classifier_B$};
		\draw [thick, red] (1,0.68) -- (1,1.3);
		\node [right] at (1, 1.3) {$\tilde\classifier_A$};
		
		\node [left] at (1.33,1.25) {$+$};
		\node [right] at (1, 1.25) {$+$};
		
		\node[black] at (1, 1-0.2 ) {\textbullet};
		\node [ right] at (1, 1-0.2 ) {$\orifeatures\textcolor{blue}{=\firstfeatures}$};
		
		\node [blue] at  (1-12/125,1-9/125) {\textbullet};
		\node [left, blue] at (1-12/125,1-9/125 ) {$\secondfeatures$};
		
		\draw [densely dashdotted, blue] (1, 1-0.2 ) --(1-12/125,1-9/125);
		\end{tikzpicture}
		\caption{Two-step strategy: $\tilde\classifier_A\rightarrow\tilde\classifier_B$}  
		\label{subfig:two step A-B}
	\end{subfigure}
	\begin{subfigure}[b]{0.45\linewidth}
		\begin{tikzpicture}[xscale=7.8,yscale=7.8]
		
		\draw [domain=0.68:1.4, thick, red] plot (\x, {3/4*\x+1/4});
		\node [left] at (1.38, 1.3 ) {$\tilde \classifier_B$};
		\draw [thick] (1,0.68) -- (1,1.3);
		\node [right] at (1, 1.3) {$\tilde\classifier_A$};
		
		\node [left] at (1.33,1.25) {$+$};
		\node [right] at (1, 1.25) {$+$};
		
		\node[blue] at (1, 1-0.2 ) {\textbullet};
		\node [ right] at (1, 1-0.2 ) {$\orifeatures$};
		
		
		\node [blue] at  (1,1) {\textbullet};
		\node [blue,right] at  (1,1) {$\firstfeatures=\secondfeatures$};

		\draw [densely dashdotted, blue] (1,1-0.2) -- (1,1);
		
		\end{tikzpicture}
		\caption{Two-step strategy: $\tilde\classifier_B\rightarrow\tilde\classifier_A$}
		\label{subfig:two step B-A}
	\end{subfigure}
	\caption{Dynamic nature of sequential mechanisms}
	\rule{0in}{1.2em}$^\dag$\scriptsize Suppose $\cost(\orifeatures,\firstfeatures,\secondfeatures)=\metric(\orifeatures,\firstfeatures)+\metric(\firstfeatures,\secondfeatures)$, where $\metric(\cdot,\cdot)$ is the Euclidean distance. The left panel shows the optimal strategy under the fixed order mechanism that offers test $\tilde\classifier_A$ as the first test, while the right panel shows the optimal strategy under the fixed order mechanism that offers test $\tilde\classifier_B$ as the first test.
\end{figure}
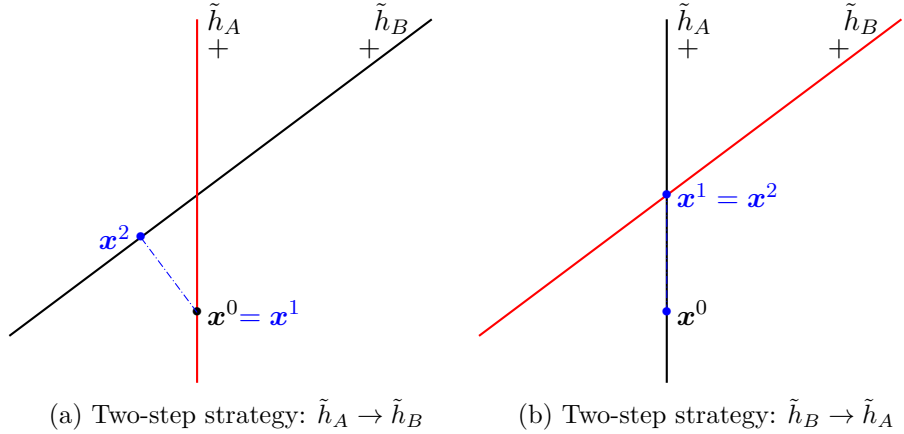  

\paragraph{More tests do not help.}
In the manipulation setting, adding more tests only makes a procedure more stringent and it is counter-productive. 
Suppose the principal offers test $\classifier_1$ and $\classifier_2$ repeatedly. For simplicity, suppose the first test is $\classifier_1$, the second is $\classifier_2$ and the third is $\classifier_1$. For those types that find it profitable to pass both tests together under a fixed-order procedure with two tests, adding an extra test does not change their incentives.
For those types that find it profitable to pass one test at a time, adding a third test only make it more costly to pass one test at a time.
Hence, the benefit of the fixed-order procedure is diminished under more tests.

In the investment setting, under the optimal simultaneous mechanism, adding one more test cannot improve the outcome since using two tests already achieve the first best.

\paragraph{Endogenous agent's technology.}
We assume that the agent's technology, whether he manipulates or invests, is exogenous. 
This usually can be explained by the different monitoring strength of the economic environment.
In the banking application, US banks are tested daily and therefore there is no room for manipulation.
In contrast, European banks are tested monthly or quarterly.
The low monitoring intensity allows banks to temporarily change their balance sheet.

Suppose instead now when the agent changes his type from $\orifeatures$ to $\firstfeatures$, he can choose between  manipulation, which costs $0.2\cdot\onecost(\orifeatures,\firstfeatures)$ and investment, which costs $0.5\cdot\onecost(\orifeatures,\firstfeatures)$.
Then the only equilibrium is such that every type chooses manipulation and the principal chooses the optimal fixed-order sequential mechanism under manipulation, given that the principal's objective is \ref{max qualified}. 
This is because (1) whichever mechanism the principal chooses, the agent prefers manipulation over investment, and (2) no type can credibly commits to investment. 
As a result, although every type would have been better-off by choosing investment, it is not an equilibrium.

\paragraph{Informational robustness.}
When the agent manipulates , the optimal tests are chosen based on the information about the agent's cost.
Whenever there is a small amount of uncertainty in the agent's cost function,  the principal needs to choose even more stringent tests so as not to select any unqualified agent in the worst case where the agent's cost happens to be least costly one.
In contrast, when the agent invests, the optimal tests coincide with the qualified region. 
Hence the optimal mechanism in the investment setting is more robust to the agent's information.

\section{Relation to the literature}
\paragraph{Literature on test design facing a strategic agent}
Our paper is closely related to the literature on test design where the agent can take hidden action to affect the outcome of the test (\citet{perez2022test, deb2018optimal}).
\citet{perez2022test} study  a setting where the principal employs a test, which is an exogenous Blackwell experiment, to provide information for the true, one-dimensional state of the world, but the agent can costly falsify the inputs into the test.
The agent only cares about getting approved by the principal while the principal only wants to approve when the state of the world is high, analogous to our qualified region.
They show that the optimal test design involves `productive falsification', i.e., it sets the passing standard higher than the true qualified threshold so that even initially qualified agent has to exert positive effort to get approved.
We recover the similar finding in our manipulation setting, which is similar to their setting except that in our setting the agent's attributes are two-dimensional and the principal uses two tests to select the agent, while they study one-dimensional attribute and one test.
We also study a setting where the agent invests, which is not studied by them.

\citet{deb2018optimal} study a setting where the principal selects a sequence of tasks to learn about the agent's true one-dimensional type, while the agent can take hidden action to affect the outcome of the task, which is jointly determined by the agent's private type and hidden action.
The principal decides whether to pass the agent at the end based on the history of performance.
The principal cares about making the right decision, i.e., passing the good type and failing the bad type, while the agent only cares about passing the test.
They identify a sufficient condition under which the principal selects the most informative task. When the condition fails to hold, they show cases where the principal prefers using less informative tasks.
Although we both study sequential testing, our focus is very different from them. 
They study the optimal choice of the informativeness of the tasks, which they call adaptive testing, while we assume that the informativeness of the tests is exogenous and deterministic.

\citet{frankel2022improving,ballscoring} study the optimal design of linear scoring rules in a setting where the principal wants to infer the agent's fundamental type from a score that reflects both the agent's fundamental type and gaming type.
The fundamental type in their settings is similar to the initial attributes in our manipulation setting, but our model is very different from them.

\paragraph{Literature on persuasion with hard evidence}
\citet{glazer2004optimal,glazer2006study,sher2014persuasion}  study persuasion problems with hard evidence where the agent's private information also has two dimensions. 
 The agent (he) always prefers being accepted, whereas the principal (she) prefers accepting only when the state of the world satisfies some conditions.
 In their settings, the agent needs to provide hard evidence to persuade the principal.
 The principal chooses a persuasion rule to minimize the probability of an error.
 Their problems can be viewed as an extreme case of ours, where (1) the agent's cost of changing attributes is infinitely high, and (2) the principal can use \emph{perfect} tests, i.e., principal can observe agent's attributes in each round of communication, while we (1) allow the agent to change his attributes with some cost, and (2) consider a  more restrictive principal's objective due to tractability. 
\citet{glazer2004optimal} shows that it is beneficial for the receiver to randomize the requested evidence after allowing the sender to send a cheap talk message.
 \citet{glazer2006study} shows that neither the commitment to the decision nor randomization have value.
\citet{sher2014persuasion} shows that (1) commitment has no value, and (2) under some conditions, randomization has no value, implying that static persuasion rules coincide with dynamic persuasion rules.
In contrast, we show that (1) randomization has no value regardless of whether or not cheap talk message is allowed, and (2) the optimal simultaneous (or single-test) mechanism, which induce a static game, can be either better or worse than the optimal sequential mechanism, which induce a dynamic game, depending on  whether the agent manipulates or invests.
\citet{Carroll_Egorov} consider a richer class of payoff functions for the agent and they focus on the possibility of full learning by randomization.

\paragraph{Literature on contracting with gaming behaviors}
We are related to a large literature on how to design contracts to induce productive effort and deter gaming. 
\citet{holmstrom1987aggregation} is one seminal paper that studies the design of contract when agents can exert effort among multiple activities.
\citet{ederer2018gaming} study how randomization over compensation scheme could deter gaming, i.e., the agent diverts effort away from valuable but difficult to measure activities to less costly and easy to measure activities. 
\citet{li2021learning} study the optimal relational contract when the agent can privately learn which activity is more important and suggest that intermittent replacement of existing measures could deter such a learning-by-shirking effect.
We differ from these papers by studying how the principal can use test design, instead of using monetary payment to either provide incentives in the investment setting, or conduct screening in the manipulation setting.
Moreover, we also differ from these paper by considering a distinct type of gaming behaviors: the kind of gaming behaviors considered by the above mentioned papers are closer to the phenomenon of multi-tasking as initiated by \citet{holmstrom1991multitask}, while we consider the gaming behaviors that exploit the rules of classification to improve apparent performance as nicely summarized in \citet{ederer2018gaming}.
In contrast to \citet{ederer2018gaming}, randomization over the order of tests is never beneficial to deter gaming, or manipulation in our model.

\paragraph{Literature on multiple attribute search}
 We are related to the literature on multiple attribute search (\citet{olszewski2016search, sanjurjo2017search}) in the sense that the principal can be viewed as a searcher and the agent's attributes are analogous to the attributes.
 In this literature, the consensus is that when there is no cost for the searcher, sequential search is weakly better than simultaneous search.
One important difference is that the agent in our setting can take strategic actions, while the object is passive in the search literature.
Moreover, we find that even when there is no cost for the principal to offer test, simultaneous procedures could be better when the agent invests.

 \paragraph{Computer science literature on strategic classification in machine learning}
 We also contribute to a computer science literature that studies how to design algorithms facing strategic agents, stemming from \citet{hardt2016strategic}.
 
Our work is most closely related to  \citet{zigzag}, who study fixed-order sequential mechanisms when the agent can manipulate his type according to an additive Euclidean cost function, and explore how the agent can exploit the sequential ordering of tests to achieve a favorable outcome at a limited cost.

Building on \citet{zigzag}, we conduct a comprehensive mechanism design analysis that encompasses random sequential mechanisms---which are central to two key results in our paper (\cref{thm: optimal max qualified} and \cref{thm: true effort sequential})---and complements \citet{zigzag}'s finding that a fixed-order sequential mechanism outperforms simultaneous mechanisms.\footnote{We also identify an error in Theorem 4.4 of \citet{zigzag} and offer a corrected version of their result in \cref{lem:fix-simul} and in \cref{thm:opt_manipulation} of our paper.}
 
 Other papers (\citet{ahmadi2022classificationstrategicagentsgame,miller2020strategicclassificationcausalmodeling,haghtalab2020maximizingwelfareincentiveawareevaluation,kleinberg2019classifiersinduceagentsinvest}) have also studied settings where the agent can both manipulate and invest at the same time. 
 The only exception is \citet{ahmadi2022classificationstrategicagentsgame}, who also compare the setting where the agent can manipulate to that where the agent can invest. 
 However, these papers have very different focus. They answer questions primarily related to computational hardness and algorithmic approximation, while we study the optimal mechanisms.


\section{Conclusion}
We study the optimal testing procedures in two extreme environments: the agent improves either  apparent performance (manipulate) or actual performance (invest) but never both. We show that the optimal testing procedure is sequential with fixed order when the agent manipulates and it is simultaneous when the agent invests.
We apply our model to explain the different banking regulatory practices in Europe and in the US, the joint regulation of merger and acquisition across departments, and the collective decision making involving subjective opinions.
Our paper points out several potential paths of future research.
First, the comparison between the incentive schemes or institutional designs facing a manipulating agent and an investing agent is fruitful. 
Second, we consider a setting where a principal has multiple requirements which can potentially contradict each other, but it is also interesting to relax the commitment assumption to study conflicting principals.
Third, one can also consider an agent with richer technology space, i.e., an agent that can both manipulate and invest.

\bibliography{reference}

\begin{thebibliography}{29}
\newcommand{\enquote}[1]{``#1''}
\expandafter\ifx\csname natexlab\endcsname\relax\def\natexlab#1{#1}\fi

\bibitem[\protect\citeauthoryear{Ahmadi, Beyhaghi, Blum, and Naggita}{Ahmadi et~al.}{2022}]{ahmadi2022classificationstrategicagentsgame}
\textsc{Ahmadi, S., H.~Beyhaghi, A.~Blum, and K.~Naggita} (2022): \enquote{On classification of strategic agents who can both game and improve,} .

\bibitem[\protect\citeauthoryear{Ball}{Ball}{Forthcoming}]{ballscoring}
\textsc{Ball, I.} (Forthcoming): \enquote{Scoring Strategic Agents,} \emph{American Economic Journal: Microeconomics}.

\bibitem[\protect\citeauthoryear{Bassi, Behn, Grill, and Waibel}{Bassi et~al.}{2024}]{window_dressing}
\textsc{Bassi, C., M.~Behn, M.~Grill, and M.~Waibel} (2024): \enquote{Window dressing of regulatory metrics: evidence from repo markets,} \emph{Journal of Financial Intermediation}, 58, 101086.

\bibitem[\protect\citeauthoryear{Carroll and Egorov}{Carroll and Egorov}{2019}]{Carroll_Egorov}
\textsc{Carroll, G. and G.~Egorov} (2019): \enquote{Strategic Communication With Minimal Verification,} \emph{Econometrica : journal of the Econometric Society.}, 87.

\bibitem[\protect\citeauthoryear{Cohen, Sharifi-Malvajerdi, Stangl, Vakilian, and Ziani}{Cohen et~al.}{2023}]{zigzag}
\textsc{Cohen, L., S.~Sharifi-Malvajerdi, K.~Stangl, A.~Vakilian, and J.~Ziani} (2023): \enquote{Sequential strategic screening,} in \emph{International Conference on Machine Learning}, PMLR, 6279--6295.

\bibitem[\protect\citeauthoryear{Deb and Stewart}{Deb and Stewart}{2018}]{deb2018optimal}
\textsc{Deb, R. and C.~Stewart} (2018): \enquote{Optimal adaptive testing: Informativeness and incentives,} \emph{Theoretical Economics}, 13, 1233--1274.

\bibitem[\protect\citeauthoryear{Ederer, Holden, and Meyer}{Ederer et~al.}{2018}]{ederer2018gaming}
\textsc{Ederer, F., R.~Holden, and M.~Meyer} (2018): \enquote{Gaming and strategic opacity in incentive provision,} \emph{The RAND Journal of Economics}, 49, 819--854.

\bibitem[\protect\citeauthoryear{Egelhof, Martin, and Zinsmeister}{Egelhof et~al.}{2024}]{ECB_window}
\textsc{Egelhof, J., A.~Martin, and N.~Zinsmeister} (2024): \enquote{Regulatory incentives and quarter-end dynamics in The repo market,} .

\bibitem[\protect\citeauthoryear{{European Central Bank}}{{European Central Bank}}{December 19 2023}]{ECBraiserequirement}
\textsc{{European Central Bank}} (December 19 2023): \enquote{Speech by Andrea Enria, Chair of the Supervisory Board of the ECB, at the press conference on the 2023 SREP results and the supervisory priorities for 2024-26,} .

\bibitem[\protect\citeauthoryear{Frankel and Kartik}{Frankel and Kartik}{2022}]{frankel2022improving}
\textsc{Frankel, A. and N.~Kartik} (2022): \enquote{Improving information from manipulable data,} \emph{Journal of the European Economic Association}, 20, 79--115.

\bibitem[\protect\citeauthoryear{Glazer and Rubinstein}{Glazer and Rubinstein}{2004}]{glazer2004optimal}
\textsc{Glazer, J. and A.~Rubinstein} (2004): \enquote{On optimal rules of persuasion,} \emph{Econometrica}, 72, 1715--1736.

\bibitem[\protect\citeauthoryear{Glazer and Rubinstein}{Glazer and Rubinstein}{2006}]{glazer2006study}
---\hspace{-.1pt}---\hspace{-.1pt}--- (2006): \enquote{A study in the pragmatics of persuasion: a game theoretical approach,} \emph{Theoretical Economics}, 1, 395--410.

\bibitem[\protect\citeauthoryear{Green and Laffont}{Green and Laffont}{1986}]{green1986partially}
\textsc{Green, J.~R. and J.-J. Laffont} (1986): \enquote{Partially verifiable information and mechanism design,} \emph{The Review of Economic Studies}, 53, 447--456.

\bibitem[\protect\citeauthoryear{Haghtalab, Immorlica, Lucier, and Wang}{Haghtalab et~al.}{2020}]{haghtalab2020maximizingwelfareincentiveawareevaluation}
\textsc{Haghtalab, N., N.~Immorlica, B.~Lucier, and J.~Z. Wang} (2020): \enquote{Maximizing Welfare with Incentive-Aware Evaluation Mechanisms,} .

\bibitem[\protect\citeauthoryear{Hardt, Megiddo, Papadimitriou, and Wootters}{Hardt et~al.}{2016}]{hardt2016strategic}
\textsc{Hardt, M., N.~Megiddo, C.~Papadimitriou, and M.~Wootters} (2016): \enquote{Strategic classification,} in \emph{Proceedings of the 2016 ACM conference on innovations in theoretical computer science}, 111--122.

\bibitem[\protect\citeauthoryear{Holmstrom and Milgrom}{Holmstrom and Milgrom}{1987}]{holmstrom1987aggregation}
\textsc{Holmstrom, B. and P.~Milgrom} (1987): \enquote{Aggregation and linearity in the provision of intertemporal incentives,} \emph{Econometrica: Journal of the Econometric Society}, 303--328.

\bibitem[\protect\citeauthoryear{Holmstrom and Milgrom}{Holmstrom and Milgrom}{1991}]{holmstrom1991multitask}
---\hspace{-.1pt}---\hspace{-.1pt}--- (1991): \enquote{Multitask principal--agent analyses: Incentive contracts, asset ownership, and job design,} \emph{The Journal of Law, Economics, and Organization}, 7, 24--52.

\bibitem[\protect\citeauthoryear{Jensen and Ruback}{Jensen and Ruback}{1983}]{jensen1983market}
\textsc{Jensen, M.~C. and R.~S. Ruback} (1983): \enquote{The market for corporate control: The scientific evidence,} \emph{Journal of Financial economics}, 11, 5--50.

\bibitem[\protect\citeauthoryear{Kleinberg and Raghavan}{Kleinberg and Raghavan}{2019}]{kleinberg2019classifiersinduceagentsinvest}
\textsc{Kleinberg, J. and M.~Raghavan} (2019): \enquote{How Do Classifiers Induce Agents To Invest Effort Strategically?} .

\bibitem[\protect\citeauthoryear{Li, Mukherjee, and Vasconcelos}{Li et~al.}{2021}]{li2021learning}
\textsc{Li, J., A.~Mukherjee, and L.~Vasconcelos} (2021): \enquote{Learning to game the system,} \emph{The Review of Economic Studies}, 88, 2014--2041.

\bibitem[\protect\citeauthoryear{Li and Qiu}{Li and Qiu}{2023}]{li2023screening}
\textsc{Li, Y. and X.~Qiu} (2023): \enquote{Screening Signal-Manipulating Agents via Contests,} \emph{arXiv preprint arXiv:2302.09168}.

\bibitem[\protect\citeauthoryear{Miller, Milli, and Hardt}{Miller et~al.}{2020}]{miller2020strategicclassificationcausalmodeling}
\textsc{Miller, J., S.~Milli, and M.~Hardt} (2020): \enquote{Strategic Classification is Causal Modeling in Disguise,} .

\bibitem[\protect\citeauthoryear{Olszewski and Wolinsky}{Olszewski and Wolinsky}{2016}]{olszewski2016search}
\textsc{Olszewski, W. and A.~Wolinsky} (2016): \enquote{Search for an object with two attributes,} \emph{Journal of Economic Theory}, 161, 145--160.

\bibitem[\protect\citeauthoryear{Perez-Richet and Skreta}{Perez-Richet and Skreta}{2022}]{perez2022test}
\textsc{Perez-Richet, E. and V.~Skreta} (2022): \enquote{Test design under falsification,} \emph{Econometrica}, 90, 1109--1142.

\bibitem[\protect\citeauthoryear{Perez-Richet and Skreta}{Perez-Richet and Skreta}{2024}]{perez2024score}
---\hspace{-.1pt}---\hspace{-.1pt}--- (2024): \enquote{Score-based mechanisms,} \emph{arXiv preprint arXiv:2403.08031}.

\bibitem[\protect\citeauthoryear{Sanjurjo}{Sanjurjo}{2017}]{sanjurjo2017search}
\textsc{Sanjurjo, A.} (2017): \enquote{Search with multiple attributes: Theory and empirics,} \emph{Games and Economic Behavior}, 104, 535--562.

\bibitem[\protect\citeauthoryear{Sher}{Sher}{2014}]{sher2014persuasion}
\textsc{Sher, I.} (2014): \enquote{Persuasion and dynamic communication,} \emph{Theoretical Economics}, 9, 99--136.

\bibitem[\protect\citeauthoryear{{U.S. Securities and Exchange Commission}}{{U.S. Securities and Exchange Commission}}{2023}]{SEC_Regulation_SK}
\textsc{{U.S. Securities and Exchange Commission}} (2023): \enquote{Regulation S-K,} \url{https://www.sec.gov/rules-regulations/staff-guidance/compliance-disclosure-interpretations/divisionscorpfinguidanceregs-kinterphtm}.

\bibitem[\protect\citeauthoryear{{U.S. Securities and Exchange Commission}}{{U.S. Securities and Exchange Commission}}{2024}]{SEC_Form_S4}
---\hspace{-.1pt}---\hspace{-.1pt}--- (2024): \enquote{Form S-4 - Registration Statement Under the Securities Act of 1933,} \url{https://www.sec.gov/files/forms-4.pdf}.

\end{thebibliography}
\newpage

\appendix
\section{Omitted proof in \cref{sec: distance cost} (manipulation)}\label{appendix: distance cost}

\subsection{Preparation: characterization of Agent's Response}\label{sec: characterization BR}

In this subsection, we characterize the agent's best response in any random-order mechanism without disclosure.

Consider a random-order mechanism without disclosure $(\tilde \classifier_A, \tilde \classifier_B, q, \varnothing)$. 
We use $\manipulation_{\probprincipal}(\tilde\classifier_A,\tilde\classifier_B)$ to denote the set of attributes that the agent does not initially pass the sequence of tests but could profitably manipulate his attribute so as to pass the selection procedure. 
To prove the main theorem, we show the following coverage property of the manipulation sets for random-order mechanisms without disclosure with different probability $q$. 
\begin{proposition}\label{prop: coverage of Mq}
     For any fixed classifiers $\tilde\classifier_A$ and $\tilde\classifier_B$, the manipulation sets for random-order mechanisms without disclosure $(\tilde\classifier_A,\tilde\classifier_B,q,\nullset)$ satisfy:
    \begin{enumerate}
        \item $\manipulation_{\probprincipal}(\tilde \classifier_A, \tilde \classifier_B) \subset \manipulation_{0} (\tilde \classifier_A, \tilde \classifier_B)$ for any $0\leq \probprincipal \leq 1/2$;
        \item $\manipulation_{\probprincipal} (\tilde \classifier_A, \tilde \classifier_B) \subset \manipulation_{1} (\tilde \classifier_A, \tilde \classifier_B)$ for any $1/2 \leq \probprincipal \leq  1$;
    \end{enumerate}  
\end{proposition}

We first formally characterize the zig-zag strategy of agents in the sequential mechanism. For any attributes $\features$ and any classifier $\tilde \classifier_A$, we use $\symmetric_{\tilde\classifier_A} \features$ to denote its symmetric point with respect to the boundary of $\tilde\classifier_A$. We use $\Pi_{\tilde\classifier_A}\features$ to denote its projection onto $\tilde\classifier_A$, which is the closest point to $\features$ in $\tilde \classifier_A$. We call a two-step strategy $(\features,\features_1,\features_2)$ a zig-zag strategy if three attributes $\features$, $\features_1$, and $\features_2$ are not on a line. 

\begin{lemma}[zig-zag strategy]\label{lem:zig-zag}
    Consider any attributes $\features \in \tilde \classifier_A^\compl \cup \tilde \classifier_B^\compl$.  
    If  $\Pi_{\classifier_B}(\symmetric_{\classifier_A}\features)$ does not satisfy $\tilde\classifier_A$, then the best strategy to first satisfy $\tilde\classifier_A$ is a zig-zag strategy $(\features,\features_1, \features_2)$, where $\features_2 = \Pi_{\classifier_B}(\symmetric_{\classifier_A}\features)$and $\features_1$ is the intersection point of $(\symmetric_{\classifier_A}\features)\features_2$ and the boundary of $\tilde\classifier_A$.
\end{lemma}

\begin{proof}[Proof of \cref{lem:zig-zag}]
    We first show that this zig-zag strategy is well-defined. Since $\features \in \tilde \classifier_A^C \cup \tilde \classifier_B^C$, we have $\symmetric_{\classifier_A}\features \in \tilde\classifier_A$. 
    If $\features_2 = \Pi_{\classifier_B}(\symmetric_{\classifier_A}\features)$ does not satisfy $\tilde\classifier_A$, then the line segment $(\symmetric_{\classifier_A}\features)\features_2$ intersects with $\tilde\classifier_A$. 
    Thus, this zig-zag strategy $(\features,\features_1,\features_2)$ is well-defined.
    
    For the zig-zag strategy $(\features,\features_1,\features_2)$, the cost of this strategy is $\cost(\features,\features_1,\features_2) = \onecost(\features,\features_1) + \onecost(\features_1,\features_2)$. Since $\features_1$ is on the boundary of $\tilde\classifier_A$, by symmetry, we have $\onecost(\symmetric_{\tilde\classifier_A}\features,\features_1) = \onecost(\features,\features_1)$. Since $\features_1$ is on the line $(\symmetric_{\tilde\classifier_A}\features)\features_2$, we have
    $$    \cost(\features,\features_1,\features_2) = \onecost(\symmetric_{\tilde\classifier_A}\features,\features_1) + \onecost(\features_1,\features_2) = \onecost(\symmetric_{\tilde\classifier_A}\features, \features_2).
    $$
    
    Consider any zig-zag strategy $(\features, \features',\features'')$ that first satisfies $\tilde\classifier_A$ and then $\tilde\classifier_B$.
    Since the attributes $\features$ do not satisfy $\tilde\classifier_1$, its symmetric point $\symmetric_{\tilde\classifier_A} \features$ satisfy $\tilde\classifier_A$.
    Since $\features'$ satisfies $\tilde \classifier_A$, by symmetry, we have $\onecost(\symmetric_{\tilde\classifier_A}\features, \features') \leq \onecost(\features, \features')$, which implies $\onecost(\symmetric_{\tilde\classifier_A}\features, \features') \leq \onecost(\features, \features')$. 
    By the triangle inequality of the Euclidean distance, the cost of this strategy is at least
    $$
    \cost(\features,\features',\features'') \geq \onecost(\symmetric_{\tilde\classifier_A}\features, \features') + \onecost(\features',\features'') \geq \onecost(\symmetric_{\tilde\classifier_A}\features, \features'').
    $$
    Since $\features_2$ is the projection of $\symmetric_{\tilde\classifier_A} \features$ onto the boundary of $\tilde\classifier_B$, we have $$\cost(\features,\features',\features'') \geq \onecost(\symmetric_{\tilde\classifier_A}\features, \features'') \geq \onecost(\symmetric_{\tilde\classifier_A}\features, \features_2) = \cost(\features,\features_1,\features_2),$$
    which completes the proof.
\end{proof}

\begin{remark}
    Consider any sequential mechanism with two tests $\tilde \classifier_A$ and $\tilde \classifier_B$ such that the angle between two tests is in $[90^{\circ}, 180^{\circ})$.
    For any attributes $\features \in \tilde \classifier_A^\compl \cup \tilde \classifier_B^\compl$, we have  $\Pi_{\classifier_B}(\symmetric_{\classifier_A}\features)$ satisfies $\tilde\classifier_A$. Thus, for such an agent, the best strategy to first satisfy $\tilde\classifier_A$ is not a zig-zag strategy. In this case, it is easy to show that the best strategy is a one-step strategy. An alternative algebraic proof is provided by Theorem 3.7 in \citet{zigzag}.
\end{remark}

In the following analysis, we only consider the sequential mechanism with two tests $\tilde \classifier_A$ and $\tilde \classifier_B$ such that the angle between two tests is in $(0,90^{\circ})$.
We now characterize the manipulation set $\manipulation_{\probprincipal}(\tilde\classifier_A,\tilde\classifier_B)$ for random-order mechanism without disclosure $(\tilde\classifier_A, \tilde\classifier_B,q,\varnothing)$.
Let $L_A$ and $L_B$ be the boundary lines of classifiers $\tilde \classifier_A$ and $\tilde \classifier_B$ respectively. 
Let $O$ be the intersection point of $L_A$ and $L_B$.
Let $\Line_A^+=\Line_A\cap \classifier_B$  and $\Line_B^+=\Line_B\cap \classifier_A$ be the part of $L_A$ and $L_B$ in the qualified region, respectively.
Let $\setperp_A(\tilde \classifier_A, \tilde \classifier_B)=\{\orifeatures\notin \tilde\classifier_A\cap \tilde\classifier_B: \min_{\genericfeatures\in\Line_A^+}\onecost(\orifeatures,\genericfeatures)\leq 1 \}$ be the set of candidates whose true attributes are not qualified but have cost less than one to adopt attributes on $\Line_A^+$.
Similarly, let $\setperp_B(\tilde \classifier_A, \tilde \classifier_B)=\{\orifeatures\notin \tilde \classifier_A\cap \tilde \classifier_B: \min_{\genericfeatures\in\Line_B^+}\onecost(\orifeatures,\genericfeatures)\leq 1 \}$ be the set of candidates whose true attributes are not qualified but have cost less than one to adopt attributes on $\Line_B^+$.

Let $\Omega(\tilde\classifier_A,\tilde\classifier_B) = \bbR^2 \setminus ((\tilde\classifier_A\cap\tilde\classifier_B)\cup \setperp_A(\tilde \classifier_A, \tilde \classifier_B)\cup \setperp_B(\tilde \classifier_A, \tilde \classifier_B))$.
Without loss of generality, we assume the unit normal vector of $\tilde\classifier_A$ is $\weights_A = (1,0)$.
Let $\util_{AB}$ be the agent's best utility among strategies that first pass only $\tilde \classifier_A$ but not $\tilde \classifier_B$ and then pass $\tilde \classifier_B$. 
We define the set $\setonetwo_{\probprincipal}(\tilde\classifier_A,\tilde\classifier_B)=\{\features\in\Omega(\tilde\classifier_A,\tilde\classifier_B): \util_{AB}\geq 0\}$.
Similarly, let $\util_{BA}$ be the agent's best utility among all strategies that first pass only $\tilde\classifier_B$ and then pass $\tilde\classifier_A$. 
Define the set $\settwoone_{\probprincipal}(\tilde\classifier_A,\tilde\classifier_B)=\{\features\in\Omega(\tilde\classifier_A,\tilde\classifier_B): \util_{BA}\geq 0\}$.
Let the agent's utility of directly moving to point $O$ be $\util_{0}$.
We define the set $\setO_{\probprincipal}(\tilde\classifier_A,\tilde\classifier_B)=\{\features\in\Omega(\tilde\classifier_A,\tilde\classifier_B): \util_{0}\geq 0 \}$.
When there is no ambiguity about the two classifiers $\tilde \classifier_A$ and $\tilde \classifier_B$ used in the mechanism, we use $\setO$, $\setonetwo$, and $\settwoone$ to denote these sets. 
Then, we have the manipulation set $\manipulation_{\probprincipal}(\tilde \classifier_A,\tilde\classifier_B)=\setperp_A \cup \setperp_B \cup \setonetwo_{\probprincipal}\cup \settwoone_{\probprincipal}\cup \setO_{\probprincipal}$.


First, it is easy to see that the set $\setO_{\probprincipal}$ is invariant with the probability $q$.

\begin{lemma}\label{lmm: Bq invariant with q}
    $ \setO_{\probprincipal} = \setO =\arc{OAB}$ for any $\probprincipal\in [0,1]$.
\end{lemma}

\begin{proof}
    The ball $B(O,1/\eta)$ contains all attributes $\features$ such that the cost $\onecost(\features,O) \leq 1$. We have $\setO_{\probprincipal} = B(O, 1/\eta) \cap \Omega(\tilde\classifier_A,\tilde\classifier_B) = \arc{OAB}$.
\end{proof}

Next, we characterize the sets $\setonetwo_{\probprincipal}$ and $\settwoone_{\probprincipal}$. Let the point $\pointonetwo_{\probprincipal}$ be the point whose distance to point $O$ is $\probprincipal/\eta$ and whose projection on $\tilde\classifier_B$ is point $O$. Then $\pointonetwo_{\probprincipal} = -\frac{\probprincipal}{\mc}\cdot \weights_B$, where $\weights_B$ is the unit normal vector of $\tilde\classifier_B$.
Note that for the agent with attributes $\pointonetwo_{\probprincipal}$, the cost for first passing $\tilde\classifier_A$ but not $\tilde \classifier_B$ and then passing $\tilde\classifier_B$ is $\probprincipal$.
Let its symmetric point with respect to the boundary of $\tilde\classifier_A$ be $\pointonetwo_{\probprincipal}' = \symmetric_{\tilde\classifier_A} \pointonetwo_{\probprincipal}$.
Let the point $\tilde \pointonetwo_{\probprincipal}  = (0,-\frac{\probprincipal}{\mc\sin{\theta}})$ be the point that falls on the boundary of $\tilde\classifier_A$ with distance $\probprincipal/\eta$ to $\tilde\classifier_B$.

\begin{lemma}\label{lem: Cq}
    $\setonetwo_{\probprincipal} = O \pointonetwo_{\probprincipal} \tilde \pointonetwo_{\probprincipal} \pointonetwo_{\probprincipal}'$.
\end{lemma}

\begin{proof}
    We first consider the agent with attributes $\features$ in the triangle region $O\pointonetwo_{\probprincipal} \tilde \pointonetwo_{\probprincipal}$. These attributes already satisfy the classifier $\tilde\classifier_A$. Since both $\pointonetwo_{\probprincipal}$ and $\tilde \pointonetwo_{\probprincipal}$ have distance $\probprincipal/\eta$ to $\tilde\classifier_B$, we have the line $\pointonetwo_{\probprincipal}\tilde \pointonetwo_{\probprincipal}$ is parallel to the boundary of $\tilde\classifier_B$. Thus, all attributes in $O\pointonetwo_{\probprincipal} \tilde \pointonetwo_{\probprincipal}$ have distance at most $\probprincipal/\eta$ to $\tilde\classifier_B$. For the agent with these attributes, the utility for providing $\features$ in the first test and then moving to pass $\tilde \classifier_B$ is non-negative. 

    Then, we consider the triangle region $O\pointonetwo_{\probprincipal}' \tilde \pointonetwo_{\probprincipal}$. Note that for each attributes $\features$ in $O\pointonetwo_{\probprincipal}' \tilde \pointonetwo_{\probprincipal}$, its symmetric point with respect to the boundary of $\tilde\classifier_A$ is in $O\pointonetwo_{\probprincipal} \tilde \pointonetwo_{\probprincipal}$. Thus, by Lemma~\ref{lem:zig-zag}, the agent with attributes in $O\pointonetwo_{\probprincipal}' \tilde \pointonetwo_{\probprincipal}$ also has the utility $u_{AB} \geq 0$.
\end{proof}

Similarly, let $\pointtwoone_{1-\probprincipal} = -\frac{1-\probprincipal}{\mc}\cdot \weights_A$ be the point whose distance to point $O$ is $1-\probprincipal$ and whose projection on $\tilde\classifier_A$ is point $O$. 
Let its symmetric point over the boundary of $\tilde\classifier_B$ be $\pointtwoone_{1-\probprincipal}' = \symmetric_{\tilde\classifier_B} \pointtwoone_{1-\probprincipal}$.
Let point $\tilde \pointtwoone_{1-\probprincipal}  = (-\frac{ 1-\probprincipal}{\mc},-\frac{1-\probprincipal}{\mc\tan{\theta}})$ be the point that falls on the boundary of $\tilde\classifier_B$ with distance $\frac{1-\probprincipal}{\eta}$ to $\tilde\classifier_A$. With a similar analysis as in Lemma~\ref{lem: Cq}, we characterize $\settwoone_{\probprincipal}$ as follows.

\begin{lemma}\label{lem: Dq}
    $\settwoone_{\probprincipal} = O \pointtwoone_{1-\probprincipal} \tilde \pointtwoone_{1-\probprincipal} \pointtwoone_{1-\probprincipal}'$.
\end{lemma}

    

\begin{lemma}\label{lmm:Cq < Dq+Bq when q is small}
   $\setonetwo_{\probprincipal}\subset \settwoone_{0}\cup \setO$ if $\probprincipal\leq \frac{1}{2\cos{\theta}}$.
\end{lemma}

\begin{proof}

Let $G = (0, -\frac{1}{\mc\sin{2\theta}})$ be the point where $\tilde \pointtwoone_{1} \pointtwoone_{1}'$ intersects with the boundary of $\tilde\classifier_A$.
When  $\probprincipal\leq \sin{\theta}$, we have $\frac{\probprincipal}{\mc\sin{\theta}}\leq \frac{1}{\mc\sin{2\theta}}$.
Since the point $\tilde \pointonetwo_{\probprincipal} = (0, - \frac{\probprincipal}{\mc\sin{\theta}})$, we have $O\tilde \pointonetwo_{\probprincipal}\subset OG$.

It is easy to see that $O\pointonetwo_{\probprincipal}\subset \setO$ for any $\probprincipal\leq 1$. Since $\pointonetwo_{\probprincipal}'$ is the symmetric point of $\pointonetwo_{\probprincipal}$ with respect to $\tilde\classifier_1$, we have $O\pointonetwo_{\probprincipal}'\subset \setO$ for any $\probprincipal\leq 1$. 

Note that $\settwoone_{0}\cup \setO$ is a convex set. Since four points $O$, $\pointonetwo_{\probprincipal}$, $\pointonetwo_{\probprincipal}'$, and $\tilde \pointonetwo_{\probprincipal}$ are contained in $\settwoone_{0}\cup \setO$, By Lemma~\ref{lem: Cq}, we have that $\setonetwo_{\probprincipal}\subset \settwoone_{\probprincipal}\cup \setO$.
\end{proof}

Similarly, we can show the following coverage property for $\settwoone_{\probprincipal}$.

\begin{lemma}\label{lmm:Dq < Cq+Bq when q is large}
   $\settwoone_{\probprincipal}\subset \setonetwo_{1}\cup \setO$ if and only if $\probprincipal\geq 1 - \frac{1}{2\cos{\theta}}$.
\end{lemma}

\begin{proof}
    Let $E = (-\frac{1}{2\eta\cos\theta}, -\frac{1}{2\eta\sin{\theta}})$ be the point where $\tilde \pointonetwo_{1} \pointonetwo_{1}'$ intersects with the boundary of $\tilde\classifier_B$. Note that the point $\tilde \pointtwoone_{1-q} = (-\frac{1-q}{\eta},-\frac{1-q}{\eta\tan\theta}) = 2\cos \theta (1-q) E$. When $q \geq 1-\frac{1}{2\cos \theta}$, we have $2\cos \theta (1-q) \leq 1$ , which implies $O\tilde \pointtwoone_{1-q} \subset OE$.
    Since $\pointtwoone_{1-q}$ and $\pointtwoone_{1-q}'$ is in $\setO$ and $\setonetwo_{1}\cup \setO$ is convex, we have $\settwoone_{\probprincipal}\subset \setonetwo_{1}\cup \setO$.
\end{proof} 


Now, we prove the coverage proposition of the manipulation sets for sequential mechanisms.

\begin{proof}[Proof of \cref{prop: coverage of Mq}]
    We first consider the case where $q \leq 1/2$. 
    By \cref{lmm:Cq < Dq+Bq when q is small}, we know that $\manipulation_{\probprincipal}=\setperp_A \cup \setperp_B \cup \setonetwo_q \cup \settwoone_{\probprincipal}\cup \setO \subset \setperp_A \cup \setperp_B \cup\settwoone_0 \cup \setO \subset \manipulation_0$.

    We then consider the case where $q \geq 1/2$. 
    By \cref{lmm:Dq < Cq+Bq when q is large}, we know that $\manipulation_{\probprincipal}=\setperp_A \cup \setperp_B \cup\setonetwo_{\probprincipal}\cup \settwoone_{\probprincipal} \cup \setO \subset \setperp_A \cup \setperp_B \cup\setonetwo_1 \cup \setO \subset \manipulation_1$.
\end{proof}

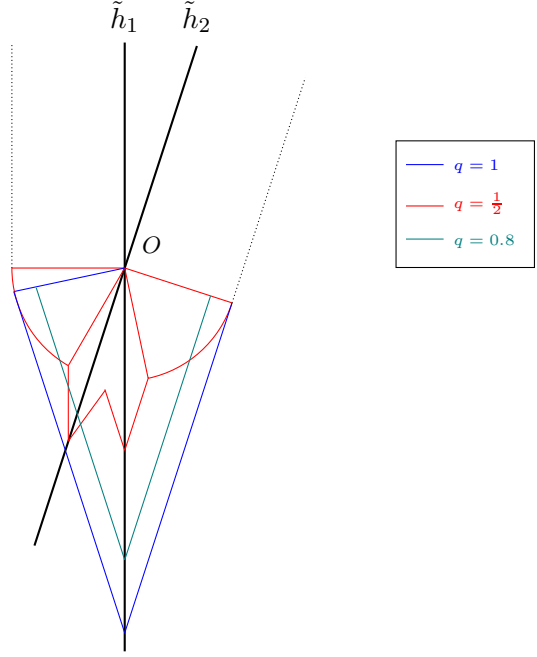
\begin{figure}[t]
\centering
\begin{tikzpicture}[xscale=6,yscale=6,
    pics/legend entry/.style={code={%
        \draw[pic actions] 
        (-0.25,0.25) -- (0.25,0.25);}}]]

\draw [domain=0.8:1.16, thick] plot (\x, {tan(deg(0.4*pi))*(\x-1)+1});
\node [above] at (1.16, 1.5 ) {$\tilde\classifier_2$};
\draw [thick] (1,0.15) -- (1,1.5);
\node [above] at (1, 1.5 ) {$\tilde\classifier_1$};

\draw [domain={1+0.25*cos(deg(0.1*pi))}:1.4, densely dotted] plot (\x, {tan(deg(0.4*pi))*(\x-1)+1-0.809}); 
\draw [densely dotted] (1-0.25,1) -- (1-0.25,1.5);
\draw[red] (1-0.25,1) arc (180:240:0.25);
\draw[red] (1,1) -- ++(180:0.25);
\draw[red] (1,1) -- ++(240:0.25) ;

\draw[red] ({1+0.25*cos(deg(0.1*pi))},{1-0.25*sin(deg(0.1*pi))}) arc (360-18:360-78:0.25);
\draw[red] (1,1) -- ++(360-18:0.25);
\draw[red] (1,1) -- ++(360-78:0.25) ;

\draw [red] (1-0.125,{1-0.125*sqrt(3)}) -- (1-0.125,0.61625);
\draw [domain=0.875:0.957, red]  plot(\x,{tan(deg(0.3*pi))*(\x-0.875)+0.61625});

\draw [domain=0.956:1,red]  plot(\x,{-tan(deg(0.4*pi))*(\x-1)+1-0.125/sin(deg(0.1*pi))});
\draw [domain=1:1.053,red]  plot(\x,{tan(deg(0.4*pi))*(\x-1)+1-0.125/sin(deg(0.1*pi))});

\draw[blue] (1,1) -- ++(192:0.25) ;
\draw [domain=1:{1+0.25*cos(deg(0.1*pi))}, blue] plot (\x, {tan(deg(0.4*pi))*(\x-1)+1-0.809});
\draw [domain=1-0.246:1,blue]  plot(\x,{-tan(deg(0.4*pi))*(\x-1)+1-0.809});

\draw [domain=1:{1+0.25*0.8*cos(deg(0.1*pi))}, teal] plot (\x, {tan(deg(0.4*pi))*(\x-1)+1-0.8*0.25/sin(deg(0.1*pi))});
\draw [domain=1-0.196:1,teal]  plot(\x,{-tan(deg(0.4*pi))*(\x-1)+1-0.8*0.25/sin(deg(0.1*pi))});

\node [above] at (1.06, 1.01 ) {\footnotesize$O$};

\matrix [draw, above right] at (1.6,1) {
 \pic[blue]{legend entry}; &  \node[blue,font=\tiny] {$\probprincipal=1$}; \\
 \pic[red]{legend entry}; &  \node[red,font=\tiny] {$\probprincipal=\frac12$}; \\
 \pic[teal]{legend entry}; &  \node[teal,font=\tiny] {$\probprincipal=0.8$}; \\
};

\end{tikzpicture}
\caption{ random-order mechanisms without disclosure vs fixed-order mechanism}
\end{figure}

\begin{figure}[t]
\centering
\begin{subfigure}[b]{0.4\linewidth}
\begin{tikzpicture}[xscale=3.5,yscale=3.5,
    pics/legend entry/.style={code={%
        \draw[pic actions] 
        (-0.25,0.25) -- (0.25,0.25);}}]]


\draw [domain=0.8:1.16, thick] plot (\x, {tan(deg(0.4*pi))*(\x-1)+1});
\node [above] at (1.16, 1.5 ) {$\classifier_2$};
\draw [thick] (1,0.15) -- (1,1.5);
\node [above] at (1, 1.5 ) {$\classifier_1$};

\draw [domain={1+0.25*cos(deg(0.1*pi))}:1.4, densely dotted] plot (\x, {tan(deg(0.4*pi))*(\x-1)+1-0.809}); 
\node [above] at (1.4, 1.5 ) {$\classifier_2^-$};
\draw [densely dotted] (1-0.25,1) -- (1-0.25,1.5);
\node [above] at (1-0.25, 1.5) {$\classifier_1^-$};
\draw[blue] (1-0.25,1) arc (180:198:0.25);
\draw[blue] (1,1) -- ++(180:0.25);

\draw[teal] ({1-0.25*cos(deg(0.1*pi))},{1-0.25*sin(deg(0.1*pi))}) arc (198:198+36:0.25);
\draw[teal] (1,1) -- ++(198+36:0.25) ;

\draw[teal] ({1+0.25*cos(deg(0.1*pi))},{1-0.25*sin(deg(0.1*pi))}) arc (360-18:360-54:0.25);
\draw[teal] (1,1) -- ++(360-18:0.25);
\draw[teal] (1,1) -- ++(360-54:0.25) ;

\draw[blue] (1,1) -- ++(192:0.25) ;
\draw [domain=1:{1+0.25*cos(deg(0.1*pi))}, blue] plot (\x, {tan(deg(0.4*pi))*(\x-1)+1-0.809});
\draw [domain=1-0.246:1,blue]  plot(\x,{-tan(deg(0.4*pi))*(\x-1)+1-0.809});

\draw [domain=1:{1+0.25*0.8*sin(47)}, teal] plot (\x, {tan(deg(0.4*pi))*(\x-1)+1-0.8*0.25/sin(deg(0.1*pi))});
\draw [domain={1-0.25*0.8*cos(43)}:1,teal]  plot(\x,{-tan(deg(0.4*pi))*(\x-1)+1-0.8*0.25/sin(deg(0.1*pi))});

\node [above,font=\tiny] at (1.06, 1.01 ) {$O$};

\node [right,font=\tiny] at ({1+0.25*cos(deg(0.1*pi))},{1-0.25*sin(deg(0.1*pi))} ) {$\pointonetwo_{\probprincipal=1}$};
\node [left,font=\tiny] at ({1-0.25*cos(deg(0.1*pi))},{1-0.25*sin(deg(0.1*pi))}) {$\pointonetwo_{\probprincipal=1}'$};
\node [below,font=\tiny] at (1, {1-0.25/sin(deg(0.1*pi))} ) {$\tilde \pointonetwo_{\probprincipal=1}$};

\matrix [draw, above right] at (1.5,0.2) {
 \pic[blue]{legend entry}; &  \node[blue,font=\tiny] {$\probprincipal=1$}; \\
 \pic[teal]{legend entry}; &  \node[teal,font=\tiny] {$\probprincipal=0.8$}; \\
};

\end{tikzpicture}
\caption{Random order $\probprincipal=0.8$} \label{fig:q>1-h(theta)}  
\end{subfigure}
\begin{subfigure}[b]{0.4\linewidth}
\begin{tikzpicture}[xscale=3.5,yscale=3.5,
    pics/legend entry/.style={code={%
        \draw[pic actions] 
        (-0.25,0.25) -- (0.25,0.25);}}]]

\draw [domain=0.74:1.16, thick] plot (\x, {tan(deg(0.4*pi))*(\x-1)+1});
\node [above] at (1.16, 1.5 ) {$\classifier_2$};
\draw [thick] (1,0.15) -- (1,1.5);
\node [above] at (1, 1.5 ) {$\classifier_1$};

\draw [domain={1+0.25*cos(deg(0.1*pi))}:1.4, densely dotted] plot (\x, {tan(deg(0.4*pi))*(\x-1)+1-0.809}); 
\node [above] at (1.4, 1.5 ) {$\classifier_2^-$};
\draw [densely dotted] (1-0.25,1) -- (1-0.25,1.5);
\node [above] at (1-0.25, 1.5) {$\classifier_1^-$};
\draw[teal] (1-0.25,1) arc (180:180+37:0.25);
\draw[teal] (1,1) -- ++(180:0.25);
\draw[teal] (1,1) -- ++(180+37:0.25) ;

\draw[red] ({1+0.25*cos(deg(0.1*pi))},{1-0.25*sin(deg(0.1*pi))}) arc (360-18:360-36:0.25);
\draw[red] (1,1) -- ++(360-18:0.25);
\draw[teal] ({1+0.25*cos(deg(0.2*pi))},{1-0.25*sin(deg(0.2*pi))}) arc (360-36:360-73:0.25);
\draw[teal] (1,1) -- ++(360-73:0.25);

\draw [teal] (1-0.25*0.8,{1-0.25*0.8*cos(44)}) -- (1-0.25*0.8,{1-0.25*0.8*tan(deg(0.4*pi))});
\draw [domain=1-0.25*0.8:{1+0.25*0.29}, teal]  plot(\x,{tan(deg(0.3*pi))*(\x-1+0.25*0.8)+1-0.25*0.8*tan(deg(0.4*pi))});

\draw[red] (1,1) -- ++(360-36:0.25) ;
\draw[red] (1-0.25,1) -- (1-0.25, {1- 0.25*tan(deg(0.4*pi))}) ;
\draw [domain=1-0.25:{1+0.25*cos(36)}, red] plot (\x, {tan(deg(0.3*pi))*(\x-1+0.25)+1-0.25*tan(deg(0.4*pi))});

\node [above,font=\tiny] at (1.06, 1.01 ) {$O$};

\node [left,font=\tiny] at ({1-0.25},1) {$\pointtwoone_{1-\probprincipal=1}$};
\node [right,font=\tiny] at ({1+0.25*cos(deg(0.2*pi))},{1-0.25*sin(deg(0.2*pi))}) {$\pointtwoone_{1-\probprincipal=1}'$};
\node [below,font=\tiny] at (1-0.25, {1-0.25*tan(deg(0.4*pi))} ) {$\tilde \pointtwoone_{1-\probprincipal=1}$};

\matrix [draw, above right] at (1.5,0) {
 \pic[red]{legend entry}; &  \node[red,font=\tiny] {$\probprincipal=0$}; \\
 \pic[teal]{legend entry}; &  \node[teal,font=\tiny] {$\probprincipal=0.2$}; \\
};

\end{tikzpicture}
\caption{Random order $\probprincipal=0.2$} \label{fig:q<h(theta)}  
\end{subfigure}

\end{figure}

\subsection{Proof of main results in \cref{subsec:seq manipulation}}


\begin{proof}[Proof of \cref{lem:gain non-parallel tests}]
   From \cref{sec: characterization BR} we know that the set of attributes selected is $\manipulation_1=\setperp_A \cup \setperp_B \cup\setonetwo_1 \cup \setO $.
   Apply \cref{lmm: Bq invariant with q} and \cref{lem: Cq}. 
\end{proof}

\begin{proof}[Proof of \cref{lem:loss non-parallel tests}]
    The proof uses similar argument as the previous one and hence is omitted.
\end{proof}

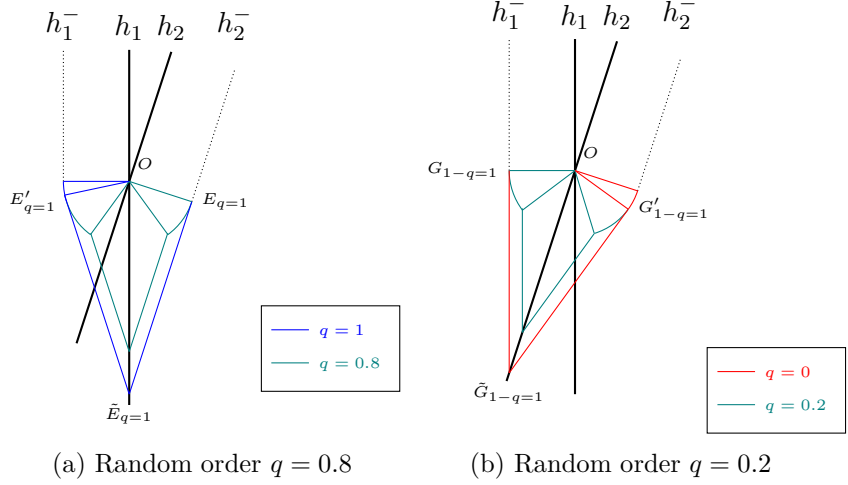
\begin{proof}[Proof of \cref{lem:feasible-informed-rand-distance cost}]
First, we  show that the fixed order mechanism $(\tilde\classifier_A,\tilde\classifier_B,1)$ is feasible.
That is, we want to show that for $\features$ that are accepted by $(\tilde\classifier_A,\tilde\classifier_B,1)$, they are also qualified.
We partition the set of attributes that are accepted by $(\tilde\classifier_A,\tilde\classifier_B,1)$ into two subsets $F_1$ and $F_2$. 
We show that for any attributes in the first subset $F_1$ are also accepted by the informed random order mechanism (Step 1).
Since the informed random order mechanism is feasible, we know that any attributes in the first subset $F_1$ are qualified.
As for the second subset $F_2$, we show that it is contained in a convex set, whose extreme points are in $F_1$ and hence are qualified(Step 2).
Since the qualified region is also convex, we can infer that the second subset $F_2$ is contained in the qualified region.
Hence any attributes in the second subset are also qualified.
The feasibility of $(\tilde\classifier_A,\tilde\classifier_B,0)$ can be shown analogously.
Lastly, we show that one of the two fixed order mechanisms is no worse than the informed random-order mechanism (Step 3).

\paragraph{Step 1} We first show that the set of attributes that (1) are accepted by informed random order mechanism  $(\tilde\classifier_A,\tilde\classifier_B,q,\test_1)$ and (2) satisfy either $\tilde\classifier_A$ or $\tilde\classifier_B$, contains the set of attributes that (1) are accepted by the fixed order mechanism  $(\tilde\classifier_A,\tilde\classifier_B,1)$ and (2) satisfy either $\tilde\classifier_A$ or $\tilde\classifier_B$.
We call the former set set $I$, the latter set set $F_1$.

Consider any $\features$ that are in set $F_1$, i.e.,  any $\features$ that (1) are accepted by the fixed order mechanism  $(\tilde\classifier_A,\tilde\classifier_B,1)$ and (2) satisfy either $\tilde\classifier_A$ or $\tilde\classifier_B$.
We distinguish three cases.

\textbf{Case 1:} Consider $\features$ that satisfy both $\tilde\classifier_A$ and $\tilde\classifier_B$.
Under the strategy $\strategies=(\features,\features)$, 
such attributes are accepted by both the informed random order mechanism and the fixed order mechanism.

\textbf{Case 2:} Consider $\features$ that satisfy $\tilde\classifier_A$ but not $\tilde\classifier_B$.
We can infer that either (1) the best response is a one-step strategy, i.e.,  $\firstfeatures=\secondfeatures$ that satisfy both $\tilde\classifier_A$ and $\tilde\classifier_B$, or (2) the best response  is a two-step strategy, i.e., $\firstfeatures$  satisfy $\tilde\classifier_A$ but not $\tilde\classifier_B$, and $\secondfeatures$ satisfy $\tilde\classifier_B$ but not $\tilde\classifier_A$.
If the former case is true, then such attributes are also accepted by the informed random order mechanism and hence are also in set $I$.
If the latter case is true, then we must have $\firstfeatures=\features$ because of triangle inequality.
Moreover, we can infer that such a strategy is profitable, i.e., $1-c(\features,\features,\secondfeatures)\geq 0$.
This implies that when such attributes use the same strategy in the informed random order mechanism, the expected utility is $q[1-c(\features,\features,\secondfeatures)]\geq 0$ and they get accepted by the informed random order mechanism with probability at least $q$.
Therefore, such attributes are also in set  $I$.

\textbf{Case 3:} Consider $\features$ that satisfy $\tilde\classifier_B$ but not $\tilde\classifier_A$.
Similarly, we can infer that either (1) it is a one-step strategy, i.e.,  $\firstfeatures=\secondfeatures$ that satisfy both $\tilde\classifier_A$ and $\tilde\classifier_B$, or (2) it is a two-step strategy, i.e., $\firstfeatures$  satisfy $\tilde\classifier_A$ but not $\tilde\classifier_B$, and $\secondfeatures$ satisfy $\tilde\classifier_B$ but not $\tilde\classifier_A$.
If the former case is true, then such attributes are also accepted by the informed random order mechanism and hence are also in set  $I$.
If the latter case is true, then we can infer that such a strategy is profitable, i.e., $1-c(\features,\firstfeatures,\secondfeatures)\geq 0$.
This implies that such attributes can use a strategy $(\features,\firstfeatures,\firstfeatures)$ in the informed random order mechanism, the expected utility is $(1-q)[1-\onecost(\features,\firstfeatures)]\geq 0$, because $\onecost(\features,\firstfeatures)\leq c(\features,\firstfeatures,\secondfeatures)$ by monotonicity.
Hence such attributes are  accepted by the informed random order mechanism with probability at least $1-q$.
Therefore, such attributes are also in set  $I$.

\paragraph{Step 2} By the first characterization, we know that $\setonetwo_1\cup \setO$ contains the set of attributes that (1) are accepted by the fixed order mechanism  $(\tilde\classifier_A,\tilde\classifier_B,1)$, and (2) satisfy neither $\tilde\classifier_A$ nor $\tilde\classifier_B$.
Call the latter set set $F_2$.

Notice that $\setonetwo_1\cup \setO$ is convex and the qualified region is convex.
To show that $F_2$ is contained in the qualified region, 
it suffices to show that the extreme points of $\setonetwo_1\cup \setO$ are qualified.
This is true because the extreme points of $\setonetwo_1\cup \setO$ are in set $F_1$.
To see this, denote the intersecting point of the boundary lines of $ \tilde\classifier_A$ and $\tilde\classifier_B$ by $\tilde O$.
The extreme points of $\setonetwo_1\cup \setO$ are $\tilde O$, $\tilde\pointonetwo_1$ and another point in $ \tilde\classifier_B$.

\paragraph{Step 3}
 we introduce a mixed mechanism that mixes two fixed-order mechanisms as follows:
     announce the fixed-order mechanism $(\tilde\classifier_A,\tilde\classifier_B,1)$ with probability $q$; and announce the other fixed-order mechanism $(\tilde\classifier_A,\tilde\classifier_B,0)$ with probability $1-q$. 
     Such a mixed mechanism only accepts qualified agent.
    
    We first show that this mixed fixed-order mechanism is better than the random-order mechanism.  
    More specifically, this mixed fixed-order mechanism accepts all attributes that are accepted by the random-order mechanism with weakly higher probability.
    Suppose the agent is accepted with probability one in the random-order mechanism. Then, this agent must use a one-step strategy and provide attributes that satisfy both $\tilde \classifier_A$ and $\tilde \classifier_B$ in the first test. In this case, the agent can always get accepted in both fixed-order mechanisms by adopting the same strategy. Thus, this agent is accepted with probability one in the mixed mechanism. 
    
    Suppose the agent gets accepted with probability $q$ in the random-order mechanism. This means the agent chooses a two-step strategy that first provides attributes $\firstfeatures$ that satisfy only test $\tilde \classifier_A$ and then provides $\secondfeatures$ that satisfy another test $\tilde \classifier_B$. 
    The expected utility of this agent is $q- d(\features,\firstfeatures)- q\times d(\firstfeatures,\secondfeatures) \geq 0$.  
    Since the cost is non-negative, we have $1- d(\features,\firstfeatures)-  d(\firstfeatures,\secondfeatures) \geq 0$.
    In the fixed-order mechanism $(\tilde\classifier_A,\tilde\classifier_B,1)$, this agent can adopt the same strategy to get utility $1- d(\features,\firstfeatures)-  d(\firstfeatures,\secondfeatures) \geq 0$.
    Thus, this agent gets accepted with probability at least $q$ in the mixed mechanism. 
    Similarly, we can show that for any agent who is accepted by the random-order mechanism with probability $1-q$ is also accepted with probability at least $1-q$ in the mixed mechanism.
    Therefore, the mixed fixed-order mechanism is no worse than the random-order mechanism. 
    
    Finally, it is easy to see that one (the better one) of the two fixed-order mechanisms is no worse than the mixed mechanism.





\end{proof}

\begin{proof}[Proof of \cref{lem:feasible-uninformed-rand-distance cost}]

We first show that two fixed-order mechanisms $(\tilde\classifier_A,\classifier_B',1)$ and $(\classifier_A',\tilde\classifier_B,0)$ are feasible.
Consider the fixed-order mechanism $(\tilde\classifier_A,\classifier_B',1)$.
By the characterization in Section~\ref{sec: characterization BR}, the manipulation set for this mechanism consists of sets $\setO$, $\setonetwo_1(\tilde\classifier_A,\classifier_B')$, $\setperp_A(\tilde\classifier_A,\classifier_B')$, and $\setperp_B(\tilde\classifier_A,\classifier_B')$.
Since the uninformed random-order mechanism $(\tilde\classifier_A,\tilde\classifier_B,q,\test_1)$ is feasible, we have the intersection point $O$ has a distance at least $1/\eta$ to the boundary of both classifiers $\tilde\classifier_A$ and $\tilde\classifier_B$, otherwise there exists unqualified agent can directly move to point $O$ with cost $1$.
Note that the classifier $\classifier_B'$ has the same normal vector $\weights_B$ as the classifier $\classifier_B$. 
Thus, any attributes in $\tilde \classifier_A \cap \classifier_B'$ have a distance at least $1/\eta$ to the boundary of $\classifier_A$ and $\classifier_B$.
Therefore, we have $\setperp_B(\tilde \classifier_A, \classifier_B', 1)$ and $\setO$ are contained in $\classifier_A \cap \classifier_B$.
By Lemma~\ref{lem: Cq}, we have $\setonetwo_1(\tilde \classifier_A, \classifier_B') = OE_1\tilde E_1 E_1'$.
Note that the boundary of $\classifier_B'$ is parallel to the boundary of $\classifier_B$.
Thus, we have attributes $E_1$ and $E_1'$ are
contained in $\classifier_A \cap \classifier_B$.
The attributes $\tilde E_1$ is the intersection of $\tilde\classifier_A$ and $\classifier_B$, which implies the attributes $\tilde E_1$ is also in $\classifier_A \cap \classifier_B$.
Since the qualified region is convex, we have the set $\setonetwo_1 \subset \classifier_A \cap \classifier_B$.
Thus, the fixed-order mechanism $(\tilde\classifier_A,\classifier_B',1)$ is feasible. 
With a similar analysis, we have $(\classifier_A', \tilde\classifier_B, 0)$ is also feasible. 

We next show that one of these two fixed-order mechanisms is no worse than the uninformed random-order mechanism. 
Suppose the boundaries of two tests $\tilde \classifier_A$ and $\tilde \classifier_B$ are parallel to the boundaries of $\classifier_A$ and $\classifier_B$ respectively. 
By Proposition~\ref{prop: coverage of Mq}, in this case, the manipulation set of the uninformed random-order mechanism $(\tilde\classifier_A,\tilde\classifier_B,q,\varnothing)$ is contained in the manipulation set of one fixed-order mechanism.

We now consider the case where the boundaries of $\tilde\classifier_A$ and $\tilde\classifier_B$ are not parallel to the boundaries of $\classifier_A$ and $\classifier_B$.
We show that a mixed mechanism that announces the fixed-order mechanism  $(\tilde\classifier_A,\classifier_B',1)$ with probability $q$ and the fixed-order mechanism $(\classifier_A', \tilde\classifier_B, 0)$ with probability $1-q$ is no worse than the uninformed random order mechanism $(\tilde\classifier_A,\tilde\classifier_B,q,\varnothing)$.
By Lemma~\ref{lem: Cq}, we have $\setonetwo_q(\tilde \classifier_A, \tilde \classifier_B) = OE_q \tilde E_q E_q'$. Since $\tilde E_q$ is on the boundary of $\tilde\classifier_A$, we must have $\tilde E_q$ is in $O \tilde E_1$, otherwise $\tilde E_q$ is not qualified. 
We also know that $E_q$ and $E_q'$ have a distance $q$ to the intersection $O$, which means $E_q$ and $E_q'$ are accepted by the fixed-order mechanism. 
Note that the set of attributes $\calA_1$ accepted by the fixed-order mechanism $(\tilde\classifier_A, \classifier_B', 1)$ is convex. 
Since attributes $O$, $E_q$, $\tilde E_q$, and $E_q'$ are contained in $\calA_1$, we have $\setonetwo_q(\tilde \classifier_A, \tilde \classifier_B) = OE_q\tilde E_qE_q'$ is contained in $\calA_1$.
Similarly, we have the set $\settwoone_q$ is contained in the set of attributes $\calA_2$ accepted by the fixed-order mechanism $(\classifier_A',\tilde \classifier_B,0)$.
If an agent is accepted by the random order mechanism $(\tilde\classifier_A, \tilde \classifier_B,q,\varnothing)$ with probability $1$, then the attributes of this agent must be in $\setO \cup (\tilde\classifier_A \cap \tilde \classifier_B) \cup \setperp_A(\tilde\classifier_A, \tilde \classifier_B) \cup \setperp_B(\tilde\classifier_A, \tilde \classifier_B)$. 
Note that $\setO$, $\tilde\classifier_A \cap \tilde \classifier_B$, $\setperp_A(\tilde\classifier_A, \tilde \classifier_B)$, and $\setperp_B(\tilde\classifier_A, \tilde \classifier_B)$ are all contained by both $\calA_1$ and $\calA_2$. 
Thus, these agents are also accepted with probability $1$ in the mixed fixed-order mechanism. 
If an agent is accepted by the random order mechanism $(\tilde\classifier_A, \tilde \classifier_B,q,t_1)$ with probability $q$, then the attributes of this agent must be in $\setonetwo_q(\tilde\classifier_A,\tilde\classifier_B)$, which is contained in $\calA_1$.
Thus, this agent is accepted with probability at least $q$ in the mixed mechanism.
Similarly, if the agent is accepted by the random order mechanism $(\tilde\classifier_A, \tilde \classifier_B,q,\varnothing)$ with probability $1-q$, then the attributes of this agent must be in $\settwoone_q(\tilde\classifier_A,\tilde\classifier_B)\subset \calA_2$.
Thus, this agent is accepted by the mixed mechanism with probability at least $1-q$.
Since one of the two fixed-order mechanisms is no worse than the mixed mechanism, we get the conclusion.
\end{proof}

We now prove our main theorem. 

\begin{proof}[Proof of Theorem~\ref{thm: optimal max qualified}]
    By \cref{lem:feasible-informed-rand-distance cost} and \cref{lem:feasible-uninformed-rand-distance cost},  we have the best sequential mechanism is a fixed order mechanism. Moreover, to avoid selecting any unqualified agent, we must have $\tilde \classifier_A \cap \tilde \classifier_B \subset \classifier_A \cap \classifier_B$.
\end{proof}

\subsection{Omitted proof in \cref{subsec: simultaneous manipulation}}
\begin{proof}[Proof of \cref{lem:fix-simultaneous-distance cost}]
      This result is implied by \cref{thm:opt_manipulation}.
\end{proof}

\subsection{Omitted proof in \cref{subsec:cheap talk manipulation}}\label{appendix:cheap talk manipulation}

First we state the following two lemmas (partially), which characterize the set of attributes that can get selected by these two fixed-order mechanisms.

\begin{lemma}\label{lem:gain non-parallel tests}
 Under the fixed-order procedure $(\classifier_A^+,\widehat\classifier_B,0)$, we have the following:
 \begin{itemize}
     \item Each type in the triangle $\Updelta\text{OAB}$ has a profitable strategy to get selected.
     \item No unqualified type has a strictly profitable strategy to get selected.
 \end{itemize}
\end{lemma}

\begin{lemma}\label{lem:loss non-parallel tests}
 Under the fixed-order procedure $(\classifier_A^+,\classifier_B^+,1)$, we have the following:
 \begin{itemize}
     \item Each type in $\classifier_A\cap\classifier_B\setminus \Updelta\text{OAB}$ has a profitable strategy to get selected.
     \item Each type in the triangle $\Updelta\text{OAB}$ does not have a profitable strategy to get selected.
     \item No unqualified type has a strictly profitable strategy to get selected.
 \end{itemize}
\end{lemma}

\begin{proof}[Proof of \cref{thm: optimal max qualified cheap talk}]
The fixed-order mechanism is described as the following: 
If the agent reports a type in the triangle $\Updelta \text{OAB}\subset \classifier_A\cap\classifier_B$, then the principal announces the fixed-order mechanism $(\classifier_A^+,\widehat\classifier_B,0)$.
 If the agent reports a type in $ \classifier_A\cap\classifier_B\setminus \Updelta \text{OAB}$, then the principal announces the fixed-order mechanism $(\classifier_A^+,\classifier_B^+,1)$. 
 If the agent reports a type that is not in $ \classifier_A\cap\classifier_B$, then the principal randomly announces either $(\classifier_A^+,\widehat\classifier_B,0)$ or $(\classifier_A^+,\classifier_B^+,1)$.

First, we show that this fixed-order mechanism is feasible. 
By \cref{lem:gain non-parallel tests} and \cref{lem:loss non-parallel tests}, no unqualified type has a profitable strategy to get selected under either $(\classifier_A^+,\widehat\classifier_B,0)$ or $(\classifier_A^+,\classifier_B^+,1)$.
Hence no unqualified type has a profitable strategy to get selected regardless of the reporting strategy.

Next we show that under this fixed-order mechanism, every qualified agent is selected.
By \cref{lem:gain non-parallel tests} and \cref{lem:loss non-parallel tests}, each type in the triangle $\Updelta \text{OAB}$ has a profitable strategy to get selected by the fixed-order mechanism $(\classifier_A^+,\widehat\classifier_B,0)$ but not in the fixed-order mechanism $(\classifier_A^+,\classifier_B^+,1)$.
Hence each type in the triangle $\Updelta \text{OAB}$ has incentive to truthfully report his type and gets selected. Again, by \cref{lem:loss non-parallel tests}, each type in $\classifier_A\cap\classifier_B\setminus \Updelta\text{OAB}$ has a profitable strategy to get selected by the fixed-order mechanism $(\classifier_A^+,\classifier_B^+,1)$ if he reports truthfully.\footnote{Here some types in $\classifier_A\cap\classifier_B\setminus \Updelta\text{OAB}$ might have incentive to misreport if they can get selected by the fixed-order mechanism $(\widehat\classifier_A,\classifier_B^+,1)$ with a lower cost. }
\end{proof}
\section{Omitted proof in  \cref{subsec: distance investment} 
 (investment)}\label{appendix: distance investment}

\begin{proof}[Proof of \cref{claim same behavior}]
    Consider any attributes $\features\notin\classifier_A\cap\classifier_B$. 
    We have shown that 
    if attributes $\features \in \setperp_A(\classifier_A,\classifier_B)$ or if attributes $\features \in \setperp_B(\classifier_A,\classifier_B)$, then the best response of any agent with such attributes is to take a one-step strategy.
    Hence it suffices to show for any attributes $\features \in \setonetwo_{\frac12}\cup\settwoone_{\frac12}$, there exists a one-step strategy  that is better than the best two-step strategy.

    \paragraph{Suppose  $\theta< 90^{\circ}$.}  The only  one-step strategy that could enable the agent to be selected is $\strategies(\features)=O$ (See \cref{fig:acute}).
    Notice that the cost of the one-step strategy is $\onecost(\features,O)=\eta \|\features - O\|_2$.
    Denote the distance between $\features$ and $O$ by $r= \|\features - O\|_2$.
    We distinguish three cases. 
    
    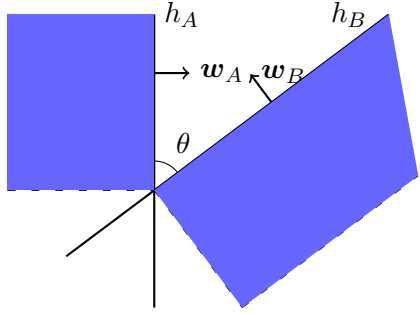
\begin{figure}
    \centering
  \begin{tikzpicture}[xscale=7.8,yscale=7.8]

\draw [domain=0.85:1.4, thick] plot (\x, {3/4*\x+1/4}) ;
\node [left] at (1.38, 1.3 ) {$\classifier_B$};
\draw [thick] (1,0.8) -- (1,1.3) ;
\node [right] at (1, 1.3) {$\classifier_A$};
\draw[black,thick,->] (1,1.2) -- (1.06,1.2)node[right]{$\weights_A$}; 
\draw[black,thick,->] (1.2,1.15) -- (1.2-0.06*0.6,1.15+0.06*0.8)node[right]{$\weights_B$}; 

\draw (1,1)  ++(90-53.1:0.05) arc (90-53.1:90:0.05); \node[right] at(1.02,1.08){$\theta$};

\draw [domain=1+0.25*0.6:1.45, loosely dashed] plot (\x, {3/4*\x+1/4-5/16}); 
\draw [loosely dashed] (1,1) -- (1+0.25*0.6,1-0.25*0.8);

\draw [loosely dashed] (1-0.25,1) -- (1-0.25,1.3);
\draw [loosely dashed] (1-0.25,1) -- (1,1);

\node[font=\tiny] at(0.88,1.1){\footnotesize$\setperp_A(\classifier_A,\classifier_B)$};
\node[font=\tiny] at(1.2,0.95){\footnotesize$\setperp_B(\classifier_A,\classifier_B)$};

\fill [blue!60,nearly transparent]  (0.75,1.3) -- (0.75,1) -- (1,1) -- (1,1.3)-- cycle;

\fill [blue!60,nearly transparent]  (1.4,5.2/4) --  (1,1) -- (1+0.25*0.6,1-0.25*0.8) --(1.45,5.35/4-5/16)-- cycle;
\end{tikzpicture}
\caption{$\theta<90^{\circ}$}
\label{fig:acute}
\end{figure}

    \paragraph{Case 1} Consider any attributes $\features$ that satisfy $\classifier_A$  but do not satisfy $\classifier_B$ (See \cref{fig: case 1}).
    Denote the angle between the vector $O\features$ and $-\weights_A$ by $\alpha$, where $\weight_A$ is the normal vector of $\classifier_A$ and $\alpha\in [0,90^{\circ}-\theta]$.
    In this case, $\features\in\settwoone_{\frac12}$, i.e., such attributes only have a profitable two-step strategy to first pass $\classifier_B$ then $\classifier_A$. 
    The best two-step strategy for $\features$ is $\strategies(\features)=(\features,\secondfeatures)$, where $\secondfeatures$ are the projection of $\features$ on the boundary line of $\classifier_A$.
    The cost of such a strategy is $\eta r \cos{\alpha}$.
    Hence the expected utility from the best two-step strategy is $\frac12-\eta r \cos{\alpha}$.
    Such attributes have a profitable two-step strategy, i.e., $\features\in \setonetwo_{\frac12}\cup\settwoone_{\frac12}$, if and only if $\frac12-\eta r \cos{\alpha}\geq 0$.
    
    It remains to show that 
    $$1-\eta r\geq \frac12-\eta r \cos{\alpha}, \text{ for all } r \text{ s.t. } \frac12-\eta r \cos{\alpha}\geq 0.$$
    It suffices to show 
    $$1\geq \frac{1- \cos{\alpha}}{\cos{\alpha}}\Leftrightarrow \cos{\alpha}\geq \frac12,$$ 
    which is true if $\theta\geq 30^{\circ}$.

\begin{figure}[t]
\centering
\begin{tikzpicture}[xscale=6,yscale=6]

\draw [domain=0.8:1.36, thick] plot (\x, {3/4*\x+1/4});
\node [left] at (1.32, 1.25 ) {$\classifier_B$};
\draw [thick] (1,0.78) -- (1,1.25);
\node [right] at (1, 1.25 ) {$\classifier_A$};
\node [left] at (1.27,1.2) {$+$};
\node [right] at (1, 1.2) {$+$};

\draw [->, thick]  (1,1)--(1-0.25,1);
\node [left] at (1-0.25, 1 ) {\footnotesize$-\weights_A$};

\node at (0.8, 0.9) {\textbullet};
\node [left] at (0.8, 0.9) {\footnotesize$\features$};
\draw[dashed](0.8, 0.9)--(1,1);
\draw (1,1)  ++({180}:0.05) arc (180:210:0.05); \node[left] at (0.96,0.98){\footnotesize$\alpha$};

\node [above] at (1.016, 1 ) {\footnotesize$O$};
\end{tikzpicture}
\caption{Case 1}
\label{fig: case 1}
\end{figure}

    \paragraph{Case 2} Consider any attributes $\features$ that satisfy $\classifier_B$  but do not satisfy $\classifier_A$ (See \cref{fig: case 2}).
    In this case, $\features\in\setonetwo_{\frac12}$, i.e., such attributes only have a profitable two-step strategy to first pass $\classifier_A$ then $\classifier_B$. 
Denote the angle between the vector $O\features$ and $-\weights_B$ by $\alpha$, where $\weight_B$ is the normal vector of $\classifier_B$ and $\alpha\in [0,90^{\circ}-\theta]$.
   The rest of the proof is analogous to case 1.

    \begin{figure}[t]
\centering
\begin{tikzpicture}[xscale=6,yscale=6]

\draw [domain=0.8:1.36, thick] plot (\x, {3/4*\x+1/4});
\node [left] at (1.32, 1.25 ) {$\classifier_B$};
\draw [thick] (1,0.78) -- (1,1.25);
\node [right] at (1, 1.25 ) {$\classifier_A$};
\node [left] at (1.27,1.2) {$+$};
\node [right] at (1, 1.2) {$+$};

 \draw [->,thick] (1,1) -- (1+0.25*0.6,1-0.25*0.8);
 \node [right] at (1+0.25*0.6,1-0.25*0.8) {\footnotesize$-\weights_B$};

\node at (1+0.05,1-0.2) {\textbullet};
\node [below] at (1+0.05,1-0.2){\footnotesize$\features$};
\draw[dashed](1+0.05,1-0.2)--(1,1);

\draw (1,1)  ++({270+atan(1/4)}:0.05) arc ({270+atan(1/4)}:{270+atan(3/4)}:0.05); 
\node[below] at (1.03,0.96){\footnotesize$\alpha$};

\node [above] at (1.016, 1 ) {\footnotesize$O$};
\end{tikzpicture}
\caption{Case 2}
\label{fig: case 2}
\end{figure}

    \paragraph{Case 3} Consider any attributes $\features$ that satisfy neither $\classifier_A$  nor $\classifier_B$.
    We first show that the best two-step strategy to first pass $\classifier_A$ and then $\classifier_B$ is less than the best one-step strategy.
    Similar argument can be used to show that the best two-step strategy to first pass $\classifier_B$ and then $\classifier_A$ is less than the best one-step strategy.
    
    Denote the symmetric attributes of $\features$ over the boundary line of $\classifier_A$ by $\tilde\features$ (See \cref{fig: case 3}).
    Then by symmetry, $r=\|\tilde\features - O\|_2$.
    Denote the projection point of $\tilde\features$ on the boundary line of $\classifier_B$ by $\secondfeatures$.
    Denote the (anti clockwise) angle between the vector $O\tilde\features$ and the boundary line of $\classifier_A$ by $\alpha$ (See \cref{fig: case 3}), where $\alpha\in [0,\theta]$.
    Then by symmetry, the best two-step strategy for $\features$ to first pass $\classifier_A$ and then $\classifier_B$ has the same cost as $\tilde\features$ moving to $\secondfeatures$, which is $\eta \|\tilde\features - O\|_2 \sin{(\theta+\alpha)}=\eta r \sin{(\theta+\alpha)}$.
 Hence the expected utility from the best two-step strategy to first pass $\classifier_A$ and then $\classifier_B$ is $\frac12-\eta r \sin{(\theta+\alpha)}$.
    Such attributes have a profitable two-step strategy to first pass $\classifier_A$ and then $\classifier_B$, i.e., $\features\in \setonetwo_{\frac12}\cup\settwoone_{\frac12}$, if and only if $\frac12-\eta r \sin{(\theta+\alpha)}\geq 0$.
    
    It remains to show that 
    $$1-\eta r\geq \frac12-\eta r \sin{(\theta+\alpha)}, \text{ for all } r \text{ s.t. } \frac12-\eta r \sin{(\theta+\alpha)}\geq 0.$$
    It suffices to show 
    $$1\geq \frac{1- \sin{(\theta+\alpha)}}{\sin{(\theta+\alpha)}}\Leftrightarrow \sin{(\theta+\alpha)}\geq \frac12,$$ 
    which is true if $\theta\geq 30^{\circ}$.
    Analogously, we can show that  if $\theta\geq 30^{\circ}$, then the best two-step strategy to first pass $\classifier_B$ and then $\classifier_A$ is less than the best one-step strategy.
    
    \begin{figure}[t]
\centering
\begin{tikzpicture}[xscale=8,yscale=8]

\draw [domain=0.76:1.36, thick] plot (\x, {3/4*\x+1/4});
\node [left] at (1.32, 1.25 ) {$\classifier_B$};
\draw [thick] (1,0.76) -- (1,1.25);
\node [right] at (1, 1.25 ) {$\classifier_A$};
\node [left] at (1.27,1.2) {$+$};
\node [right] at (1, 1.2) {$+$};

\draw [dotted] (1-0.1,1-0.2) -- (1+0.1,1-0.2);
\node at (1-0.1, 1-0.2 ) {\textbullet};
\node at (1+0.1, 1-0.2 ) {\textbullet};
\node [below] at (1-0.1, 1-0.2 ) {\footnotesize$\features$};
\draw [dotted]  (1+0.1,1-0.2) -- (1+0.1-5.45/24*0.6,1-0.2+5.45/24*0.8);
\node [below] at (1+0.1,1-0.2) {\footnotesize$\tilde\features$};
\node [left] at  (1+0.1-5.5/24*0.6,1-0.2+5.5/24*0.8) {\footnotesize$\secondfeatures$};
\node  at  (1+0.1-5.45/24*0.6,1-0.2+5.45/24*0.8) {\textbullet};

\draw [] (1,1) -- (1+0.1,1-0.2);
\draw (1,1)  ++({270}:0.05) arc (270:{300}:0.05); 
\node[below] at (1.04,0.96){\footnotesize$\alpha$};

\node [above] at (1.016, 1 ) {\footnotesize$O$};
\end{tikzpicture}
\caption{Case 3}
\label{fig: case 3}
\end{figure}

 \paragraph{Suppose  $\theta\geq  90^{\circ}$.}  By Theorem 3.7 in \citet{zigzag}, $\setonetwo_{\frac12}\cup\settwoone_{\frac12}$ is empty.
 Hence every agent prefers one-step strategy.

 \begin{figure}[h]  
\centering 
    \begin{tikzpicture}[xscale=7.8,yscale=7.8]

\draw [domain=0.86:1.26, thick] plot (\x, {-3/4*\x+7/4});
\node [right] at (1.28, 1-0.18 ) {$ \classifier_B$};
\draw [thick] (1,0.8) -- (1,1.3);
\node [right] at (1, 1.3) {$\classifier_A$};
\draw[black,thick,->] (1,1.2) -- (1.06,1.2)node[right]{$\weights_A$}; 
\draw[black,thick,->] (1.2,0.85) -- (1.2+0.06*0.6,0.85+0.06*0.8)node[right]{$\weights_B$}; 

\draw (1,1)  ++({-36.9}:0.05) arc (-36.9:90:0.05); \node[right] at (1.04,1.04){$\theta$};


\draw [domain=0.86:1.1, loosely dashed] plot (\x, {-3/4*\x+23/16});
\draw [loosely dashed] (1,1) -- (1-0.25*0.6,1-0.25*0.8);

\draw [loosely dashed] (1-0.25,1) -- (1-0.25,1.3);
\draw [loosely dashed] (1-0.25,1) -- (1,1);

\node[font=\tiny] at(0.88,1.2){\footnotesize$\setperp_A(\classifier_A,\classifier_B)$};
\node[font=\tiny] at(1.15,0.75){\footnotesize$\setperp_B(\classifier_A,\classifier_B)$};

\fill [blue!60,nearly transparent]  (0.75,1.3) -- (0.75,1) -- (1,1) -- (1,1.3)-- cycle;

\fill [blue!60,nearly transparent]  (1.26, -3/4*1.26+7/4) --  (1,1) -- (1-0.25*0.6,1-0.25*0.8) --(1.1,-3.3/4+23/16)-- cycle;

\end{tikzpicture}
\caption{$\theta\geq 90^{\circ}$}  
\label{fig:obtuse}
\end{figure}
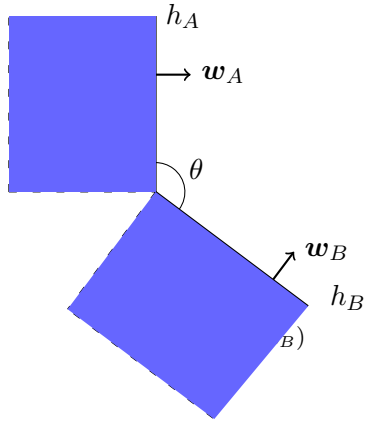

\end{proof}

\begin{proof}[Proof of \cref{thm:optimal investment}]
    We show that the described simultaneous mechanism achieves the upper bound of the value of program \ref{max qualified}. 
    First notice that any qualified agent gets selected with zero cost under the described mechanism.
    Second, we show that the described mechanism accepts every unqualified agent who can improve to another qualified attributes with cost less than one.
    Consider an unqualified agent with attributes $\features\not\in\classifier_A \cap \classifier_B$. 
    Suppose there exists attributes $\features'$ such that (1)  $\features'$ satisfy both tests $\classifier_A$ and $\classifier_B$; and (2)  such an agent can move to $\features'$ with cost $c(\features, \features', \features') \leq 1$.
    Then this unqualified agent has an incentive  to improve his attributes to $\features'$ and get accepted.  
    However, if instead such attributes $\features'$ do not exist and for any attributes $\features''$ in $\classifier_A \cap \classifier_B$, the cost for this agent with attributes $\features$ to improve to $\features''$ is strictly larger than one, then this unqualified agent has no incentive to move to the qualified region $\classifier_A \cap \classifier_B$ in any mechanism.
    Thus, the simultaneous mechanism that uses tests $\classifier_A$ and $\classifier_B$ accepts every agent who is either qualified or can improve their attributes to the qualified region with a cost less than one. 
    This is the upper bound of the value of program \ref{max qualified}.
    Hence the described simultaneous mechanism is optimal.
\end{proof}

\section{Omitted proof in \cref{sec: perfect tests}}\label{appendix: perfect tests}

\begin{proof}[Proof of \cref{thm: sequential perfect tests}]
     Consider the following fixed-order sequential mechanism with two tests.
     For any message $m\in \Featurespace$ sent by the agent, let $\firstfeatures(m)$ be the  first test and $\secondfeatures(m)$ be the second test.
      
     If $m \in \classifier_A^+\cap \classifier_B^+$, then $\firstfeatures(m)=\secondfeatures(m)=m$.
     If $m\not\in \classifier_A \cap \classifier_B$, then $\firstfeatures(m)=\secondfeatures(m)=\features$ for some $\features \in \classifier_A^+\cap \classifier_B^+$.
     Call  $\classifier_A^+\cap \classifier_B^+$ the immediate accepted region.
    If $m\in (\classifier_A \cap \classifier_B) \setminus (\classifier_A^+\cap \classifier_B^+)$, then $\firstfeatures(m)=m$. 
    

    
    Let  $\onecost(\firstfeatures, \classifier_A^+\cap \classifier_B^+) = \min_{\genericfeatures \in \classifier_A^+\cap \classifier_B^+} \eta \cdot \metric(\firstfeatures,\genericfeatures) $ denote the lowest cost from the first test $\firstfeatures$ to the immediate accepted region.
    Then, for any $m\in (\classifier_A \cap \classifier_B) \setminus (\classifier_A^+\cap \classifier_B^+)$, 
     (1) if $\onecost(m, \classifier_A^+\cap \classifier_B^+) \leq 1$, the second test
     $\secondfeatures(m)= \argmin_{\genericfeatures \in \classifier_A^+\cap \classifier_B^+} \eta \cdot \metric(\firstfeatures,\genericfeatures) $; (2) if $\onecost(m, \classifier_A^+\cap \classifier_B^+) > 1$,  the second test is some $\secondfeatures(m)\in \classifier_A \cap \classifier_B$ such that the  cost for the second test $\onecost(\firstfeatures(m),\secondfeatures(m))$ is one.

    \textbf{Step 1: feasibility.}  Consider any agent with attribute $\features \not \in \classifier_A \cap \classifier_B$.
    Then the minimum cost for such $\features$ to any type in $\classifier_A^+\cap \classifier_B^+$ is at least $1$.
    This implies that it is not profitable for such $\features$ to take any first test in $\classifier_A^+\cap \classifier_B^+$.
    Suppose this agent reports some type $m \in (\classifier_A \cap \classifier_B) \setminus (\classifier_A^+\cap \classifier_B^+)$.  If $\onecost(m, \classifier_A^+\cap \classifier_B^+) > 1$, then this agent has to pay a cost one in the second test to get accepted.  Since the cost for the first test is non-negative, the total cost exceeds one.
    If $\onecost(m, \classifier_A^+\cap \classifier_B^+) \leq 1$, then to get accepted, this agent must produce $\secondfeatures \in \classifier_A^+\cap \classifier_B^+$. By the triangle inequality, the total cost for this agent to get accepted is at least $\onecost(\features,\firstfeatures) + \onecost(\firstfeatures,\secondfeatures) \geq \onecost(\features, \secondfeatures) > 1$. Therefore, any unqualified agent has no profitable strategy to get accepted. 

    \textbf{Step 2: first best.} For any qualified agent, to get selected, he  truthfully reports his type.
    Any agent with  $\firstfeatures \in \classifier_A^+\cap \classifier_B^+$ gets accepted with zero cost.
    We only need to show that for  any agent with attributes $\features \in (\classifier_A \cap \classifier_B) \setminus (\classifier_A^+\cap \classifier_B^+)$, the cost to get accepted is no greater than one.
    This is true by the construction of the mechanism.
\end{proof}

\begin{proof}[Proof of \cref{thm: perfect tests arbitrary qualified region}]
     Let $\qualregion$ be the convex qualified region. We use $\partial \qualregion$ to denote the boundary of this region $\qualregion$. 
    Let $\qualified = \{\features \in \qualregion : \metric(\features, \partial \qualregion) \geq 1/\mc\}$ be the qualified attributes (points inside $\qualregion$) whose distance to the boundary $\partial \qualregion$ is at least $1/\mc$.
    For any $r > 0$, we use $\tilde \qualified(r) = \{\features: \metric(\features, \qualified) = r\}$ to be the set of all attributes with minimum distance $r$ to the region $\qualified$.

 We consider the following sequential mechanism with two tests. 
    
    \emph{First test.}
    Let $\firstfeatures$ be the attribute observed from the first test. 
    An agent with attribute $\firstfeatures$ passes the first test if and only if this attribute is contained in the qualified region $\firstfeatures 
    \in \qualregion$.
    The mechanism immediately rejects every agent with $\firstfeatures \not \in \qualregion$.
    Among the agents who passed the first test, the mechanism immediately accepts those agents with attribute $\firstfeatures \in \qualified$. Every agent with attribute $\firstfeatures \in \qualregion \setminus \qualified$ passes the first test and has to participate in the second test. 

    \emph{Second test.} 
    Let  $r_1 = \inf_{\tilde \qualregion} \metric(\firstfeatures,\genericfeatures)$ be the distance between the first revealed attributes $\firstfeatures$ and the immediately accepted region $\qualified$.
    Then, any agent with revealed attributes $\firstfeatures \in \qualregion \setminus \qualified$ in the first test is accepted if and only if:
    either (1) the second revealed attributes $\secondfeatures \in \tilde \qualified(r_1 - 1/\eta)$ for $r_1 > 1/\eta$, i.e., the second test exhausts the benefit of getting accepted when the agent's first revealed test is far away from the immediately accepted region $\qualified$, or (2) the second revealed attributes
    $\secondfeatures \in \qualified$ for $r_1 \leq 1/\eta$, i.e., the second test requires the agent to move to the immediately accepted region $\qualified$ when the agent's first revealed test is close to the immediately accepted region $\qualified$.

    \textbf{Step 1: We first show that this mechanism does not accept any unqualified agent.}  Consider any agent with attribute $\features \not \in \qualregion$.
    To pass the first test, this agent has to move to the attributes $\firstfeatures \in \qualregion$. 
    Since $\features \not \in \qualregion$, the minimum distance from $\features$ to the immediately qualified region $\qualified$ is at least $1/\eta$, i.e., $\inf_{\genericfeatures \in \qualified} \onecost(\features,\genericfeatures)  > 1$. 
    Thus, this agent has no incentive to produce $\firstfeatures \in \qualified$, which means this agent can not get accepted through only the first test.
    Suppose the agent moves to a attribute $\firstfeatures \in \qualregion \setminus \qualified$. Then, we know the first observed attributes $\firstfeatures \in \tilde \qualified(r_1)$ for $r_1= \inf_{\genericfeatures \in \qualified} \metric(\firstfeatures,\genericfeatures)  > 0$. If $r_1 > 1/\eta$, then this agent has to move to $\secondfeatures \in \tilde \qualified (r_1 - 1/\eta)$ to get accepted. Thus, the cost $\onecost(\firstfeatures,\secondfeatures) = \eta \cdot \metric(\firstfeatures,\secondfeatures) \geq 1$. This requires a total cost $\onecost(\features,\firstfeatures) + \onecost(\firstfeatures,\secondfeatures) > 1$, which is larger than the reward. If $r_1 \leq 1/\eta$, then to get accepted, this agent must produce $\secondfeatures \in \qualified$. By the triangle inequality, the total cost for this agent to get accepted is at least $\onecost(\features,\firstfeatures) + \onecost(\firstfeatures,\secondfeatures) \geq \onecost(\features, \secondfeatures) > 1$. Therefore, this agent has no incentive to manipulate the attributes to get accepted. 

    \textbf{Step 2: We now show that this mechanism accepts all qualified agents.} Every agent with attributes in $\qualified$ gets accepted by providing their original attributes $\features$ in the first test with no cost. Consider any agent with attributes $\features \in \qualregion \setminus \qualified$. Suppose $\features \in \tilde \qualified(r_1)$. Let $\features'$ be the closest point in $\qualified$ to this point $\features$. If $r_1 \leq 1/\eta$, then this agent can get accepted by providing $\firstfeatures = \features$ and $\secondfeatures = \features'$ in two tests sequentially. The total cost is $\onecost(\features, \features') \leq 1$. If $r_1 > 1/\eta$, then there exists a point $\features'' \in \tilde \qualregion(r_1 - 1/\eta)$ which lies on the line connecting $\features$ and $\features'$ due to the convexity of the region $\qualregion$. Thus, this agent will be accepted by producing $\firstfeatures = \features$ and $\secondfeatures = \features''$ with cost $1$. Therefore, under their best response, all qualified agents will be accepted. 

\end{proof}

\begin{proof}[Proof of \cref{thm: single-test mechanism perfect test}]
    
    
    \textbf{Part 1: sufficiency.}
    Suppose the qualified region $\qualregion$ is the convex hull of the union of balls with radius $1/\eta$. Consider a single-test mechanism that accepts an agent if and only if the first revealed attributes is in the immediately accepted region $\firstfeatures \in \qualified$.
    We show such a single-round test mechanism accepts all qualified agents and no unqualified agents.  
    
    \textbf{Step 1: this mechanism accepts every qualified agent.} 
    To do this, we only need to show that the distance between every attribute $\features$ in $\qualregion$ and the accepted region $\qualified$ is at most $1/\eta$, $\inf_{\genericfeatures \in \qualified} \metric(\features,\genericfeatures) \leq 1/\eta$. 
    Since every attribute in $\qualregion$ has a distance at most $1/\eta$ to $\qualified$, any qualified agent prefers to move to $\qualified$ and get accepted with cost at most $1$.  
    
    Let $\{\ballcenter_i\}$ be the centers of radius $1/\eta$ balls in the convex region $\qualregion$. 
    By assumption, we know that the qualified region $\qualregion = \conv(\bigcup_{i} B(\ballcenter_i, 1/\eta))$.
    We know that the distance between the center $\ballcenter_i$ and the boundary of $\qualregion$ is at least $d(\ballcenter_i,\partial \qualregion) \geq 1/\eta$. Thus, all ball centers $\ballcenter_i$ are contained in the accepted region $\qualified$. 
    
    Let $\conv(\{\ballcenter_i\})$ be the convex hull of these centers $\{\ballcenter_i\}$.

    We want to show that the convex hull $\conv(\{\ballcenter_i\})$ is contained in the accepted region $\qualified$.
    Consider any attribute $\features \in \conv(\{\ballcenter_i\})$.
    This attribute $\features$ can be written as a linear combination of two centers $\features = \lambda \ballcenter_i + (1-\lambda)\ballcenter_j$. Then, any attribute $\features' = \features + \boldsymbol{y}$ in the ball centered at $\features$ with radius $1/\eta$ can be written as the linear combination $\features' = \lambda(\ballcenter_i + \boldsymbol{y}) + (1-\lambda) (\ballcenter_j + \boldsymbol{y})$, where $\ballcenter_i + \boldsymbol{y}$ is in the ball with radius $1/\eta$ centered at $\ballcenter_i$. 
    We conclude that this attribute $\features'$ is contained in $\qualregion$, because the region $\qualregion$ is the convex hull of these balls centered at $\ballcenter_i$. Thus, the ball centered at $\features$ with radius $1/\eta$ is contained in the region $\qualregion$. 
    
    Then, this attribute $\features$ is also contained in the (immediately) accepted region $\qualified$. Hence, the convex hull $\conv(\{\ballcenter_i\})$ is contained in the accepted region $\qualified$. 
    
    Consider the union of balls centered at any attribute $\features \in \conv(\{\ballcenter_i\})$ with radius $1/\eta$, $\Sigma = \bigcup_{\features \in \conv(\{\ballcenter_i\})} B(\features, 1/\eta)$. This set $\Sigma$ is convex. The region $\qualregion$ is contained in the convex hull of this set $\Sigma$. Thus, we know that the region $\qualregion$ is contained in $\Sigma$, which means any attribute $\features \in \qualregion$ has a distance at most $1/\eta$ to the set $\conv(\{\ballcenter_i\})$. Since $\conv(\{\ballcenter_i\})$ is contained in $\qualified$, any attribute in $\qualregion$ has a distance at most $1/\eta$ to $\qualified$.
    
    \textbf{Step 2: we now show that this mechanism rejects every unqualified agent.} Note that $\qualified = \{\features \in \qualregion : \metric(\features, \partial \qualregion) \geq 1/\eta\}$. Every unqualified agent with attributes $\features \not \in \qualregion$ has a  minimum distance to $\qualified$ that is greater than $1/\eta$. This implies that the cost any unqualified agent to get accepted exceed the benefit, i.e., any unqualified agent has no incentive to move. 

    \paragraph{Part 2: necessity.} Suppose the qualified region $\qualregion$ is not the convex hull of a union of balls with radius $1/\eta$. We show that there does not exist a single-round test mechanism that accepts all qualified agents but no unqualified agents. 

    We prove this by contradiction. Suppose there is a single-round test mechanism that accepts all qualified agents but no unqualified agents. 
    First, this mechanism must make deterministic decisions based on observed attributes, i.e., it either accepts an agent with probability one or does not accept an agent with probability one. Suppose not, and this mechanism accepts an agent with a probability $p \in (0,1)$, then it will reject qualified agents or accept unqualified agents. 

    Since the qualified region $\qualregion$ is not the convex hull of the union of balls with radius $1/\eta$, we know there exist at least one attributes $\features \in \qualregion\setminus \qualified$ with distance to $\qualified$ greater than $1/\eta$.
    This is true because otherwise $\qualregion$ can be written as the union of balls with radius $1/\eta$ centered at attributes in $\qualified$, which is a contradiction. 
    
    Consider one such attribute $\features \in \qualregion \setminus \qualified$ with distance to $\qualified$ greater than $1/\eta$. 
    By assumption, the mechanism accepts qualified agent with this attribute $\features$.
    For this to happen, this mechanism must accept some attributes $\genericfeatures$ whose  distance to $\features$ is at most $1/\eta$, so that it is profitable for any agent with attribute $\features$ to move to $\genericfeatures$. 
    However, the distance from $x$ to $\qualified$ is greater than $1/\eta$ implies that $\genericfeatures$ is not in $\qualified$.
    This further implies that the ball with radius $1/\eta$ centered at $\genericfeatures$ also contains unqualified attributes. Otherwise, $\genericfeatures$ should be contained in $\qualified$.
    Thus, this mechanism will accept unqualified agents, which is a contradiction.

\end{proof}

\section{Omitted proof in \cref{subsubsec: general cost manipulation}}\label{appendix:general costs}
\begin{proof}[Proof of \cref{prop:optimal simultaneous manipulation general cost}]
Consider any simultaneous mechanism that uses two tests $\tilde \classifier_A$ and $\tilde \classifier_B$.
Denote the intersecting point of the boundary lines of $\tilde \classifier_A$ and $\tilde \classifier_B$ by $\tilde O$.

Since program \ref{max qualified} requires feasible simultaneous mechanism not to accept unqualified agent, we must have the set of attributes that satisfy both $\tilde \classifier_A$ and $\tilde \classifier_B$ is contained in the set of attributes that satisfy both $ \classifier_A$ and $\classifier_B$, i.e., $\tilde \classifier_A \cap \tilde \classifier_B \subset \classifier_A\cap \classifier_B$.

Denote the angle between $ \tilde\classifier_A$ and $\classifier_B$ by $\theta( \tilde\classifier_A,\classifier_B)$.
Denote the angle between $ \classifier_A$ and $\tilde\classifier_B$ by $\theta( \classifier_A,\tilde\classifier_B)$.

\paragraph{Step 1:} Any feasible simultaneous mechanism $(\tilde \classifier_A,\tilde \classifier_B)$ must satisfy $\theta( \tilde\classifier_A,\classifier_B)\leq\theta$ and $\theta( \classifier_A,\tilde\classifier_B)\leq \theta$.

Suppose $\theta( \tilde\classifier_A,\classifier_B)>\theta$. 
This implies that the boundary lines of $\tilde \classifier_A$ and $ \classifier_A$ intersect at some point that is in $ \classifier_A\cap \classifier_B$.
This further implies that $\tilde \classifier_A\cap \tilde \classifier_B$ includes some attributes that are outside $ \classifier_A\cap \classifier_B$.
The simultaneous mechanism $(\tilde \classifier_A,\tilde \classifier_B)$ cannot be feasible. A contradiction.
Analogously, we can show that $\theta( \classifier_A,\tilde\classifier_B)>\theta$ is not possible.

\paragraph{Step 2:} Any simultaneous mechanism $(\tilde \classifier_A,\tilde \classifier_B)$ such that $\theta( \tilde\classifier_A,\classifier_B)<\theta$ or $\theta( \tilde\classifier_A,\classifier_B)<\theta$ cannot be optimal.

Suppose $\theta( \tilde\classifier_A,\classifier_B)<\theta$.
Let $\classifier_A^+$ be the half plane obtained by shifting $\classifier_A$ along $\weights_A$ such that $\tilde O$ falls on the boundary line of  $\classifier_A^+$.

\begin{claim}\label{claim 1}
    The simultaneous mechanism $(\classifier_A^+,\tilde \classifier_B)$ is feasible.
\end{claim}

\begin{proof}[Proof of \cref{claim 1}]
    Suppose not. Then there exists some unqualified attributes $\features$ that are accepted by this simultaneous mechanism $(\classifier_A^+,\tilde \classifier_B)$.
    Then there must exist some attributes $\features'$ in the set of attributes that satisfy both tests, i.e., $\features'\in\classifier_A^+\cap\tilde \classifier_B$, such that $\onecost(\features,\features')\leq 1$.
    Moreover, such attributes $\features'$ are either on the boundary line of $\tilde\classifier_B$ or on the boundary line of $\classifier_A^+$.
    Suppose the attributes $\features'$ are not on the boundary line of  $\classifier_A^+$ ($\tilde\classifier_B$). 
    Then, the line segment between $\features$ and $\features'$ intersects the boundary line of $\classifier_A^+$ ($\tilde\classifier_B$) at some attributes $\hat \features$. 
    Since the cost function is absolute homogeneous, we have  
    $
        \onecost(\features,\features') 
        > \onecost(\features,\hat\features) 
    $.
    This implies that the agent will be better off choosing $\hat\features$, which gives a contradiction.
    
    Next, we analyze each case separately.

    \textbf{Case 1}:  $\features'$ are on the boundary line of $\tilde\classifier_B$. 
    This immediately implies that the unqualified attributes $\features$ are accepted in the simultaneous mechanism $(\tilde \classifier_A,\tilde \classifier_B)$.
    A contradiction.

\textbf{Case 2}:  $\features'$ are on the boundary line of $\classifier_A^+$.
Define the vector$ v$ as $v=\features'-\tilde O$.
Then $v$ is a vector parallel to the boundary line of $\classifier_A$ and $\classifier_A^+$ and it points towards $\tilde O$.
This implies that $\features+v$ are also unqualified.
However, since the cost is translation invariant, we can infer that $\onecost(\features+v,\tilde O)=\onecost(\features,\features')\leq 1$.
In other words, the unqualified attributes $\features+v$ are accepted in the simultaneous mechanism $(\tilde \classifier_A,\tilde \classifier_B)$.
This again contradicts that the simultaneous mechanism $(\tilde \classifier_A,\tilde \classifier_B)$ is feasible.
    
\end{proof}


By construction, $\tilde \classifier_A\cap\tilde \classifier_B\subset \classifier_A^+\cap\tilde \classifier_B$.
Hence the any attributes that can move to some attributes in $\tilde\classifier_A\cap\tilde \classifier_B$ with a cost no greater than one  can also move to some attributes in $\classifier_A^+\cap\tilde \classifier_B$ with a cost no greater than one.
Moreover, 
since the cost function is absolute homogeneous, the set of attributes that can move to $\classifier_A^+\cap\tilde \classifier_B$ with a cost no greater than one is strictly larger when $\tilde \classifier_A\cap\tilde \classifier_B$ is a strict subset of $\classifier_A^+\cap\tilde \classifier_B$.
This implies that the simultaneous mechanism $(\classifier_A^+,\tilde \classifier_B)$ is better than $(\tilde \classifier_A,\tilde \classifier_B)$.

Suppose $\theta( \classifier_A,\tilde\classifier_B)<\theta$.
Let $\classifier_B^+$ be the half plane obtained by shifting $\classifier_B$ along $\weights_B$ such that $\tilde O$ falls on the boundary line of  $\classifier_B^+$.
We can use analogous argument to show that the simultaneous mechanism $(\classifier_A^+, \classifier_B^+)$ is feasible and better than $( \classifier_A^+,\tilde \classifier_B)$.

\end{proof}


\begin{proof}[Proof of \cref{lem:fix-simul}]
We consider the fixed order mechanism $(\classifier_A^{+},\classifier_B^{+},1)$ that first uses test $\classifier_A^{+}$ and then uses test $\classifier_B^{+}$. 
Another fixed order mechanism can be analyzed similarly.

\paragraph{Step 1:}
   We first show that these two fixed-order mechanisms are feasible for~\ref{max qualified}, which means they will not accept unqualified agents. 
    
    Suppose not. Then there exists an unqualified agent with attributes $\features$ that are accepted by this fixed-order mechanism. 
    Let $\firstfeatures, \secondfeatures$ be the attributes provided by this agent in the first test and the second test respectively.
    Then, the manipulation cost of this agent is $c(\features, \firstfeatures, \secondfeatures) \leq 1$.
    
    \textbf{Case 1:} Suppose the attributes $\features$ does not satisfy $\classifier_A$.  
    
    If the attributes $\firstfeatures$ also satisfy $ \classifier_B^+$, then this unqualified agent can be accepted by the simultaneous mechanism with tests $\classifier_A^{+},\classifier_B^{+}$. 
    A contradiction.
    
    If this attributes $\firstfeatures$ do not satisfy $\classifier_B^+$, then we can find a vector $\genericfeatures$ parallel to the boundary lines of $\classifier_A$ and $\classifier_A^+$ such that $\firstfeatures + \genericfeatures$ satisfy both $ \classifier_A^+$ and $\classifier_B^+$. 
    This implies that $\firstfeatures + \genericfeatures$ are accepted by the simultaneous mechanism $(\classifier_A^{+},\classifier_B^{+})$.
    Since the simultaneous mechanism $(\classifier_A^{+},\classifier_B^{+})$ is feasible, $\firstfeatures + \genericfeatures$ are qualified.
    Since $\features$ do not satisfy $\classifier_A$ and the vector $\genericfeatures$ is parallel to the boundary line of $\classifier_A$, we can infer that the shifted attributes $\features + \genericfeatures$ also do not satisfy $\classifier_A$, i.e.,  $\features + \genericfeatures$ are not qualified. 
    Since the cost function is translation invariant and is monotone, we have 
    $$
    \onecost(\features+\genericfeatures, \firstfeatures+\genericfeatures) = \onecost(\features, \firstfeatures) \leq c(\features, \firstfeatures, \secondfeatures) \leq 1.
    $$
    Thus, the unqualified agent with attributes $\features+\genericfeatures$ will be accepted by the simultaneous mechanism, which gives a contradiction. 

    \textbf{Case 2:} Suppose the attributes $\features$ does not satisfy $\classifier_B$. 

    If the attributes $\secondfeatures$ also satisfy $\classifier_A^+$, by triangle inequality, $\onecost(\features,\secondfeatures) \leq \cost(\features,\firstfeatures,\secondfeatures) \leq 1$.
    This implies that the attributes $\features$ can get accepted by the simultaneous mechanism with tests $\classifier_A^{+},\classifier_B^{+}$. 
    A contradiction.
    
    If this attributes $\secondfeatures$ do not satisfy $\classifier_A^+$, then we can find a vector $\genericfeatures$ parallel to the boundary lines of $\classifier_B$ and $\classifier_B^+$ such that $\secondfeatures + \genericfeatures$ satisfy both $\classifier_A^+$ and $ \classifier_B^+$. 
This implies that $\secondfeatures + \genericfeatures$ are accepted by the simultaneous mechanism $(\classifier_A^{+},\classifier_B^{+})$.
    Since the simultaneous mechanism $(\classifier_A^{+},\classifier_B^{+})$ is feasible, $\secondfeatures + \genericfeatures$ are qualified.
    Since $\features$ do not satisfy $\classifier_B$ and the vector $\genericfeatures$ is parallel to the boundary line of $\classifier_B$, we can infer that the shifted attributes $\features + \genericfeatures$ also do not satisfy $\classifier_B$, i.e.,  $\features + \genericfeatures$ are not qualified. 
    Since the cost function is translation invariant and satisfies triangle inequality, we have 
    $$
    \onecost(\features+\genericfeatures, \secondfeatures+\genericfeatures) = \onecost(\features, \secondfeatures) \leq c(\features, \firstfeatures, \secondfeatures) \leq 1.
    $$
    Thus, the unqualified agent with attributes $\features+\genericfeatures$ will be accepted by the simultaneous mechanism, which gives a contradiction.

\paragraph{Step 2:}
    These two fixed-order mechanisms are not worse than the optimal simultaneous mechanism. 
    This is because for any agent accepted by the simultaneous mechanism, the agent can provide the attributes in $\classifier_A^{+}\cap\classifier_B^{+}$ to get accepted in the fixed-order mechanisms.
    That is, any qualified agent accepted by the simultaneous mechanism is also accepted by these two fixed-order mechanisms.
\end{proof}


\begin{proof}[Proof of \cref{lem:stringent tests}]
We prove the contra-positive of the statement.
    Consider any mechanism that uses tests $\tilde\classifier_A$ and $\tilde\classifier_B$ such that $\tilde\classifier_A\cup\tilde\classifier_B\subsetneq \classifier_A\cup\classifier_B$ is not satisfied, which has three possibilities:
    \begin{enumerate}
        \item the set $\tilde \classifier_A\cap \tilde\classifier_B$ covers the qualified region $\classifier_A\cap\classifier_B$, i.e., $\classifier_A\cap\classifier_B \subset\tilde \classifier_A\cap \tilde\classifier_B  $ (\textbf{case 1}); 
        \item the set $\tilde \classifier_A\cap \tilde\classifier_B$ and the qualified region $\classifier_A\cap\classifier_B$ have no inclusion relation, i.e.,$\classifier_A\cap\classifier_B \not\subset\tilde \classifier_A\cap \tilde\classifier_B  $ and  $\tilde \classifier_A\cap \tilde\classifier_B  \not\subset \classifier_A\cap\classifier_B $(\textbf{case 2}).
        \item the set $\tilde \classifier_A\cap \tilde\classifier_B$ and the qualified region $\classifier_A\cap\classifier_B$ coincide, i.e., $\tilde \classifier_A\cap \tilde\classifier_B = \classifier_A\cap\classifier_B$  (\textbf{case 3});   
    \end{enumerate}

    We want to show that in any of the three cases, some unqualified attributes are selected and hence none of these mechanisms is feasible.

    Case 1 implies that the intersection of the set
    $\tilde \classifier_A\cap \tilde\classifier_B$ and the set $ \classifier_A^\compl $ is not empty.
     Notice that any attributes in $\tilde \classifier_A\cap \tilde\classifier_B$ are selected under any mechanism that uses these two tests.
    Hence there exists some unqualified attributes in the intersection, i.e.,$\features \in \tilde \classifier_A\cap \tilde\classifier_B  \cap\classifier_A^\compl$, such that $\features$ are selected automatically under the mechanism that uses tests $\tilde \classifier_A, \tilde\classifier_B$.
    This shows that any mechanism that uses tests in case 1 is not feasible.

    Case 2 implies that the intersection of the set
    $\tilde \classifier_A\cap \tilde\classifier_B$ and the set $ \classifier_i^\compl $ is not empty for at least one $i\in \{A,B\}$.
    Hence there exists some unqualified attributes in the intersection, i.e.,$\features \in \tilde \classifier_A\cap \tilde\classifier_B  \cap \classifier_i^\compl$, such that $\features$ are selected automatically under the mechanism that uses tests $\tilde \classifier_A, \tilde\classifier_B$.
    This shows that any mechanism that uses tests in case 1 is not feasible.

    Lastly, we want to show case 3 is not feasible. 
    
    By \cref{def: existence of minimum cost}, 
    there exists some attributes that do not satisfy $\classifier_A$, i.e., $\firstfeatures\in \classifier_A^\compl$ such that the minimum cost from changing such attributes to another that passes $\classifier_A$ is less than one, i.e.,  $\min_{\genericfeatures\in \classifier_A}\cost(\firstfeatures,\genericfeatures)\leq 1$.
     This means that there exists some attributes $\genericfeatures^1$ in $\classifier_A$  such that the minimum cost from changing $\firstfeatures$ to  pass $\classifier_A$ is achieved, i.e., $\cost(\firstfeatures,\genericfeatures^1)\leq 1$. 
     
     Suppose $\genericfeatures^1$ is also in $\classifier_B$, then we have found unqualified attributes $\firstfeatures$ that could move to  $\tilde\classifier_A\cap \tilde\classifier_B$ with cost no greater than one.
     That is, such attributes have at least one profitable strategy to get selected.
     Such a mechanism is not feasible. 

Suppose $\genericfeatures^1$ is not in $\classifier_B$.
Then we can find some non-zero vector $\genericfeatures^\perp$ such that (1) $\genericfeatures^\perp\cdot \weights_A=0$, i.e., $\genericfeatures^\perp$ is parallel to line $\Line_A: \weights_A\cdot \features =0$, and (2) the shifted attributes $\genericfeatures^1+\genericfeatures^\perp$ also satisfy $\classifier_B$.
    Since the cost function is translation invariant, we must have $\cost(\firstfeatures+\genericfeatures^\perp,\genericfeatures^1+\genericfeatures^\perp)\leq 1$.
    Hence we have found unqualified attributes $\firstfeatures+\genericfeatures^\perp$ that could move to  $\tilde\classifier_A\cap \tilde\classifier_B$ with cost no greater than one.
     That is, such attributes have at least one profitable strategy to get selected.
     Such a mechanism is not feasible.
     
     We conclude that any mechanism that uses tests in case 3 is not feasible. 
    
\end{proof}

Now we prove the main theorem.
\begin{proof}[Proof of \cref{thm:opt_manipulation}]

    \cref{prop:optimal simultaneous manipulation general cost} and  \cref{lem:fix-simul} together imply that there exists a feasible fixed-order mechanism that is no worse than the optimal simultaneous mechanism. Hence the optimal mechanism must be a sequential one.
    \cref{lem:stringent tests} shows that the optimal mechanism must use stringent tests.
\end{proof}

\section{Omitted proof in \cref{subsubsec:investment}}\label{appendix: seq general cost}

\begin{proof}[Proof of \cref{lem: optimal max qualified improving effort}]
    Since this mechanism $(\tilde \classifier_A, \tilde \classifier_B,q,\nullset)$ is feasible, i.e., no unqualified agent is selected, we can infer that the region that satisfies both tests $\tilde \classifier_A$ and $\tilde \classifier_B$ is contained in the qualified region $\classifier_A \cap \classifier_B$, i.e., $\tilde \classifier_A\cup\tilde \classifier_B\subset\classifier_A \cap \classifier_B$. If not, then there is an unqualified agent with attributes $\features \in \tilde\classifier_A \cap \tilde\classifier_B \setminus (\classifier_A \cap \classifier_B)$ and this agent can get accepted by this mechanism without changing his attributes. A contradiction.
    
    Let $O$ be the intersection point of the boundary lines of $\classifier_A$ and $\classifier_B$.
    Let $\tilde O$ be the intersection point of  the boundary lines of $\tilde \classifier_A$ and $ \tilde \classifier_B$. 
    We use $\classifier_A'$ and $\classifier_B'$ to denote the tests shifted from $\tilde\classifier_A$ and $\tilde\classifier_B$ by $O - \tilde O$, respectively. 
    Thus, the boundary lines of the two tests $\classifier_A'$ and $\classifier_B'$ also intersect at point $O$. 
    
    We show that the uninformed random order mechanism with two tests $\classifier_A'$ and $\classifier_B'$ and probability $q$ is also feasible, i.e., no qualified agent is selected.
    It is equivalent to show that any agent being accepted by this new mechanism has qualified attributes, i.e., attributes in the qualified region $\classifier_A \cap \classifier_B$.
    Since the cost is translation invariant, it must be that no agent prefers any two-step strategy in this new mechanism $(\classifier_A',\classifier_B',\probprincipal,\nullset)$. Otherwise, we can find an agent who prefers a two-step strategy in the initial uninformed random order mechanism with tests $\tilde\classifier_A$ and $\tilde \classifier_B$ and probability $q$.
    Note that the region that satisfies both tests $\classifier_A'$ and $\classifier_B'$ is also contained in $\classifier_A \cap \classifier_B$, i.e., $\classifier_A' \cap \classifier_B'\subset \classifier_A \cap \classifier_B$. 
    Combining with the fact that agent only uses one-step strategy in the new mechanism, we can conclude that whoever gets accepted by the new mechanism must eventually have attributes in the qualified region.
    Therefore, this new uninformed random order mechanism only accepts qualified agents after investment. 
    
    We now show that the uninformed random order mechanism with two tests $\classifier_A$ and $\classifier_B$ and probability $q$ is also feasible. To prove this, we show that there is no agent prefers any two-step strategy in this mechanism $(\classifier_A,\classifier_B,q,\nullset)$. Suppose there is an agent with attributes $\features$ who prefers a two-step strategy $(\features, \firstfeatures, \secondfeatures)$. Let $\features'$ be the attributes in $\classifier_A \cap \classifier_B$ that minimizes the cost $c(\features,\features',\features')$. Then, we have the expected utility of this two-step strategy is at least the utility of improving to $\features'$. 
    Next, we show that this agent can get accepted with the same probability by taking this two-step strategy in the uninformed random order mechanism with tests $\classifier_A',\classifier_B'$ and randomization probability $q$, i.e., $(\classifier_A',\classifier_B',\probprincipal,\nullset)$. 
    Without loss of generality, we consider that $\firstfeatures$ satisfies $\classifier_A$ but not $\classifier_B$ and $\secondfeatures$ satisfies $\classifier_B$ but not $\classifier_A$. Since $\classifier_A'$ goes through $O$ and $\classifier_A' \cap \classifier_B'$ is contained in $\classifier_A \cap \classifier_B$, we have $\classifier_A \setminus \classifier_B \subseteq \classifier_A' \setminus \classifier_B$. Thus, we can infer that $\firstfeatures$ also satisfies $\classifier_A'$. Similarly, we have that $\secondfeatures$ satisfies $\classifier_B'$. Thus, this agent can also get accepted with this two-step strategy in the uninformed random order mechanism with tests $\classifier_A'$ and $\classifier_B'$. 
    Let $\features''$ be the attributes in $\classifier_A' \cap \classifier_B'$ that minimizes the cost $c(\features,\features'',\features'')$.
    Since $\classifier_A' \cap \classifier_B'$ is contained in $\classifier_A \cap \classifier_B$, we have $c(\features,\feature'',\features'') \geq c(\features, \features',\features')$. Therefore, this agent would prefer the two-step strategy in the uninformed random order mechanism with $\classifier_A'$ and $\classifier_B'$, which provides a contradiction.

    Since there is no agent prefers a two-step strategy in the uninformed random order mechanism with two tests $\classifier_A,\classifier_B$ and probability $q$, every agent has the same best response under $(\classifier_A,  \classifier_B,q,\nullset)$  as in the optimal simultaneous mechanism $(\classifier_A,  \classifier_B)$.
    We conclude that 
    this uninformed random order mechanism accepts the same set of agent as in the optimal simultaneous mechanism. 
    Hence it implements the optimal simultaneous mechanism. 
\end{proof}


\begin{proof}[Proof of \cref{thm:optimal sequential metric investement}]
    

    Consider any fixed half plane $\classifier_A$. Suppose there exists a half plane $\classifier_B$ such that in the  mechanism $(\classifier_A,\classifier_B,q,\nullset)$, \cref{condition:one-step} holds. 
    Let $\theta \in (0, 180^{\circ})$ be the angle from $\classifier_A$ to $\classifier_B$.

    \paragraph{Step 1}
    We first show that for any half plane $\tilde \classifier_B$ that has an angle $\tilde \theta > \theta$ from $\classifier_A$,
    \cref{condition:one-step} holds in the  mechanism $(\classifier_A,\tilde\classifier_B,q,\nullset)$.
    
    Since the angle $\tilde \theta > \theta$, we have (1) the set of attributes that satisfies $\tilde \classifier_B$ but not $\classifier_A$, i.e., $\tilde \classifier_B \setminus \classifier_A$  is contained in the region $\classifier_B \setminus \classifier_A$; (2) the set of attributes that satisfies $\classifier_A$ but not $\tilde \classifier_B$, i.e., $ \classifier_A \setminus \tilde\classifier_B$  is contained in the region $\classifier_A \setminus \classifier_B$; and (3) the qualified region $\classifier_A \cap \classifier_B$ is contained in $\classifier_A \cap \tilde \classifier_B$. 
    
    \textbf{Case 1.} Consider any agent with attributes $\features \in \tilde \classifier_B \setminus \classifier_A$.
    First, the least costly one-step strategy is the same in $(\classifier_A,\tilde\classifier_B,q,\nullset)$ and in $(\classifier_A,\classifier_B,q,\nullset)$.
    Second, the least costly two-step strategy to first wait and then change to other attributes to pass $\classifier_A$ is also the same in $(\classifier_A,\tilde\classifier_B,q,\nullset)$ and in $(\classifier_A,\classifier_B,q,\nullset)$.
    These together imply that if such an agent 
    prefers a two-step strategy in the uninformed random order mechanism $(\classifier_A,\tilde\classifier_B,q,\nullset)$, then this agent also prefers the two-step strategy in the uninformed random order mechanism $(\classifier_A,\classifier_B,q,\nullset)$. 
    Hence, this agent with attributes $\features$ will not prefer a two-step strategy in the uninformed random order mechanism $(\classifier_A,\tilde\classifier_B,q,\nullset)$.
    
    \textbf{Case 2.} Consider any agent with attributes in the region $\classifier_A \setminus \tilde \classifier_B$. Since such attributes are also in $\classifier_A \setminus \classifier_B$, we know from the assumption that such attributes do not prefer a two-step strategy.
    
    \textbf{Case 3.} Consider any agent with attributes $\features$ that  satisfy neither $\classifier_A$ nor $\tilde \classifier_B$. 
    Let $(\features,\firstfeatures,\secondfeatures)$ be the least costly two-step strategy that first satisfies the test $\tilde \classifier_B$.
    Then, we have $\firstfeatures$ satisfy only test $\tilde \classifier_B$ and $\secondfeatures$ satisfy only test $\classifier_A$.
    
    Consider an agent with attributes $\firstfeatures$.
    He has one available two-step strategy, which is to wait in the first test and then change to $\secondfeatures$ in the second test.
    Such a two-step strategy yields him an expected payoff $1-q - \onecost(\firstfeatures,\secondfeatures)$.
    Let $(\firstfeatures, \features')$ be the least costly one-step strategy for the agent with attributes $\firstfeatures$. 
    Since $\firstfeatures$ satisfy $\tilde \classifier_B$ but not $\classifier_A$, we have shown that this agent with attributes $\firstfeatures$ prefers a one-step strategy. 
    This implies that $1- q - \onecost(\firstfeatures,\secondfeatures) \leq 1 - \onecost(\firstfeatures, \features')$. 
    Therefore, we have the expected utility of the two-step strategy $(\features, \firstfeatures, \secondfeatures)$ is at most
    $$
    1-q - c(\features, \firstfeatures, \secondfeatures) = 1-q - \onecost(\features, \firstfeatures) - \onecost(\firstfeatures, \secondfeatures) \leq 1- \onecost(\features, \firstfeatures) - \onecost(\firstfeatures,\features') \leq 1 - \onecost(\features,\features'),
    $$
    where the first equality is because the cost is additive and the last inequality is from the triangle inequality of the cost function. 
    Similarly, the agent with attributes $\features$ also does not prefer any two-step strategy that first satisfies $\classifier_A$.
    Thus, the agent with attributes $\features$ prefers some one-step strategy, i.e., \cref{condition:one-step} holds in the mechanism $(\classifier_A,\tilde\classifier_B,q,\nullset)$.

\paragraph{Step 2} 
   Next we show that for any fixed half plane $\classifier_A$, there is an angle $\theta^*(\classifier_A) \in (0,180^{\circ})$ such that for any half plane $\classifier_B$, if the angle between $\classifier_A$ and  $\classifier_B$ is greater than $\theta^*(\classifier_A)$, then for any $q\in [0,1]$, the mechanism $(\classifier_A,\classifier_B,q,\nullset)$, \cref{condition:one-step} holds.
   The proof is constructive.
   We first construct a half plane $\tilde \classifier_B$ (Step 2a).
   Second, we show that
   for the uninformed random order mechanism with tests $\classifier_A$ and $\tilde \classifier_B$, the best one-step strategy is less costly than the best two-step strategy that first satisfies $\tilde \classifier_B$ and then $\classifier_A$ (Step 2b). 
Let $\tilde O$ be the intersection point of the boundaries of $\classifier_A$ and $\tilde \classifier_B$.
Third, we then show that this property  is preserved when we rotate the half plane $\tilde\classifier_B$ around $\tilde O$ to increase the angle between $\classifier_A$ and $\tilde\classifier_B$ (Step 2c). 
Fourth, we show that after the rotation, there exists a half plane $\classifier_B$ that satisfies that the best one-step strategy is less costly than the best two-step strategy that first satisfies $\classifier_A$ (Step 2d). 
   
    \paragraph{Step 2a}
    Consider any attributes $\features$ that do not satisfy $\classifier_A$. Let $\features'$ be the attributes that minimize the cost $\onecost(\features,\features')$ among all attributes satisfying $\classifier_A$.
    Let $v$ be the vector orthogonal to the vector $\features' - \features$, i.e. $v \cdot (\features' - \features) = 0$.
    We pick the direction (of $v$) so that the clockwise angle from the normal vector $\weights_A$ of $\classifier_A$ to $v$ is in $(0, 180^{\circ})$ (See \cref{fig: general cost investment optimal sequential case 1}).\footnote{In this way, the angle between $\classifier_A$ and 
 any half plane $\tilde \classifier_B$ whose normal vector is $v$ is between $(0,180^{\circ})$.} 
 Consider any half plane $\tilde \classifier_B$ whose normal vector is $v$. 


\begin{figure}[t]
\centering
\begin{tikzpicture}[xscale=6,yscale=6]

\draw [domain=-0.4:0.4, thick] plot (\x, {3/4*\x});
\node [above] at (0.4,0.3 ) {$\tilde\classifier_B$};
\draw [thick] (-0.4,0) -- (0.4,0);
\node [right] at (0.4,0 ) {$\classifier_A$};
\node [above] at (0.3,0 ) {$+$};
\node [right] at (0.3,0.22) {$+$};

 \draw [<-]  (0.34,0.3) to [bend left=45]  (0.4,0.3);

\node [above] at (0,0 ) {\footnotesize$\tilde O$};

\node at (-0.2,-0.08 ) {\textbullet};
\node [left] at (-0.2,-0.08 ) {\footnotesize$\features$};

\node at (-0.28/3,0 ) {\textbullet};
\node [above] at (-0.28/3,0 ) {\footnotesize$\features'$};

\node at (0.04,-0.120 ) {\textbullet};
\node [right] at (0.04,-0.12 ) {\footnotesize$\firstfeatures$};

\node at (0.2,0 ) {\textbullet};
\node [below] at (0.22,0 ) {\footnotesize$\secondfeatures$};

\draw [ ->] (-0.2,-0.08 )  -- (-0.28/3,0 );

\draw [ ->, red] (-0.28/3,0 )-- (0.18/4-0.28/3,-0.06);
\node [right, red] at (0.18/4-0.28/3,-0.06) {\footnotesize$v$};

\draw [ ->, blue] (-0.2,-0.08 )  -- (0.04,-0.120 );
\draw [ ->, blue] (0.04,-0.120 )  -- (0.2,0);

\end{tikzpicture}
\caption{Step 2a in \cref{thm:optimal sequential metric investement}}
\label{fig: general cost investment optimal sequential case 1}
\end{figure}

\paragraph{Step 2b} We show that the best two-step strategy that first satisfies $\tilde\classifier_B$ and then $\classifier_A$ is more costly than some one-step strategy.
Suppose the best two-step strategy that first satisfies $\tilde\classifier_B$ and then $\classifier_A$ is $\strategies=(\firstfeatures,\secondfeatures)$.
Then the vector $\secondfeatures-\firstfeatures$ is parallel to the vector $\features'-\features$ by absolute homogeneity.
Since $\firstfeatures$ satisfy $\tilde\classifier_B$ and the normal vector of $\tilde\classifier_B$ is orthogonal to $\features'-\features$, the attributes $\secondfeatures$ also satisfy $\tilde\classifier_B$. 
Thus, attributes $\secondfeatures$ satisfy both $\classifier_A$ and $\tilde \classifier_B$.
By triangle inequality, $c(\features,\firstfeatures,\secondfeatures)\leq \onecost(\features,\secondfeatures)$.
Hence for any $q$, the agent receives lower expected utility from the two-step strategy $(\features,\firstfeatures,\secondfeatures)$ than the utility from one-step strategy $(\features,\features_2)$.

\paragraph{Step 2c}
     We fix the intersection point $\tilde O$ of the boundaries of $\classifier_A$ and $\tilde \classifier_B$ and rotate $\tilde \classifier_B$ to increase the angle between $\classifier_A$ and $\tilde \classifier_B$.
     We call the half plane $\hat \classifier_B$ after rotation.
    For the same attributes $\features$, let $\hat\strategies=(\features_1,\features_2)$ be the best two-step strategy that first satisfies $\hat \classifier_B$ after the rotation.
    Notice that after the rotation, $\hat\classifier_B \setminus \classifier_A$ is contained in $\tilde \classifier_B \setminus \classifier_A$. 
    Hence the two-step strategy $\hat\strategies=(\features_1,\features_2)$ is also available to the agent before the rotation.
    By revealed preference, we can infer that the total cost of the best two-step strategy $\hat\strategies=(\features_1,\features_2)$ under test $\classifier_A$ and $\hat\classifier_B$ is weakly higher than the two-step strategy $\strategies=(\firstfeatures,\secondfeatures)$ that the agent uses under test $\classifier_A$ and $\tilde\classifier_B$.
    Hence, the best one-step strategy still dominates the best two-step strategy that first satisfies $\hat \classifier_B$ as the angle between $\classifier_A$ and $\hat \classifier_B$ increases.
    
\paragraph{Step 2d}
Again, consider the half plane $\tilde \classifier_B$  and $\hat \classifier_B$, where the latter is obtained by  rotating $\tilde\classifier_B$ to increase its angle to $\classifier_A$.
Denote the angle between $\tilde\classifier_B$ and $\classifier_A$ by $\tilde \theta$ and the angle between $\hat\classifier_B$ and $\classifier_A$ by $\hat \theta$.
We have $\hat\theta >\tilde \theta$.
    Consider any attributes $ \features_1 \in \classifier_A \setminus \tilde \classifier_B$. 
    Consider any attributes $\tilde \features_2$ that satisfies only $\tilde \classifier_B$ but not $\classifier_A$. Let $\tilde u = \tilde\features_2 - \features_1$ be the direction moving to $\tilde\features_2$.

    As we have argued in step 1, Since the angle $\hat \theta >\tilde \theta$, we have (1) the set of attributes that satisfy $\hat \classifier_B$ but not $\classifier_A$, i.e., $\hat \classifier_B \setminus \classifier_A$  is contained in the region $\tilde \classifier_B \setminus \classifier_A$; (2) the set of attributes that satisfy $\classifier_A$ but not $\hat \classifier_B$, i.e., $ \classifier_A \setminus \hat\classifier_B$  is contained in the region $\classifier_A \setminus \tilde\classifier_B$; and (3) the qualified region $\classifier_A \cap \tilde \classifier_B$ is contained in $\classifier_A \cap \hat \classifier_B$.
    
    Hence we can find  attributes $\hat \features_2$ that satisfies only $\hat \classifier_B$ but not $\classifier_A$ such that $\hat\features_2 =\features_1 +\lambda \tilde u$ for some $\lambda > 1$, i.e., the direction $\hat \features_2 -\features_1$ is in parallel to $\tilde u = \tilde\features_2 - \features_1$.

    Since the cost function is absolute homogeneous, 
    the cost to reach $\hat \classifier_B \setminus \classifier_A$ from $\features_1$ only increases.
    Similarly, the cost to reach $\hat \classifier_B \cap \classifier_A$ from $\features_1$ only decreases. 
    Therefore, as the rotation angle increases, there exists a half plane $\classifier_B$ with an angle $\theta^*(\classifier_A)$ from $\classifier_A$ such that for any attribute $\features_1 \in \classifier_A \setminus \classifier_B$, the attributes $\features_2$ in $\classifier_B$ that minimizes the cost $\onecost(\features_1,\features_2)$ also satisfies $\classifier_A$.
    Consider any attributes $\features \in \classifier_A^\compl \cap \classifier_B^\compl$. Let $(\features,\features_1,\features_2)$ be the best two-step strategy that first satisfies $\classifier_A$. Since $\features_1$ satisfies only $\classifier_A$, we have $\features_2$ satisfies both $\classifier_A$ and $\classifier_B$. By the triangle inequality of the cost function, we have $\cost(\features,\features_1,\features_2) \leq \onecost(\features,\features_2)$.
    
    Thus, there exists an angle $\theta^*(\classifier_A) \in (0, 180^{\circ})$ such that when the angle between $\classifier_A$ and $\tilde \classifier_B$ is $\theta^*(\classifier_A)$, no agent prefers the two-step strategy.  

    Finally, by \cref{lem: optimal max qualified improving effort}, when no agent prefers the two-step strategy in the uninformed random order mechanism with tests $\classifier_A$ and $\classifier_B$, this uninformed random order mechanism is the optimal sequential mechanism. This random order mechanism achieves the first best.
\end{proof}

\begin{proof}[Proof of \cref{prop:investment fixed order better than informed}]
    Suppose there exists a feasible informed random order mechanism. Let $\tilde \classifier_A$ and $\tilde \classifier_B$ be the two tests used in this mechanism. 
    Denote the intersecting points of the two boundary lines of  $\tilde \classifier_A$ and $\tilde \classifier_B$ by $\tilde O$.
    
    We first show that the fixed-order mechanism $(\tilde \classifier_A,\tilde \classifier_B,1)$ is feasible. The proof for the feasibility of another fixed-order mechanism $(\tilde \classifier_A,\tilde \classifier_B,0)$ is analogous.

    Since the informed random order mechanism with two tests $\tilde \classifier_A$ and $\tilde \classifier_B$ is feasible, the region $\tilde \classifier_A \cap \tilde \classifier_B$ is contained in $\classifier_A \cap \classifier_B$. 
    Otherwise, the agent with attributes in $\tilde \classifier_A \cap \tilde \classifier_B \setminus \classifier_A \cap \classifier_B$ is accepted without improvement. In the fixed order mechanism, the agent becomes qualified if this agent eventually moves to $ \classifier_A \cap  \classifier_B$. 
    
    Suppose an agent takes a one-step strategy in the fixed order mechanism.
    For this strategy to be a best response, the agent must be accepted.
    Hence his final attributes must fall in $\tilde \classifier_A \cap \tilde \classifier_B$.
    This implies that such an agent eventually becomes qualified.
    It remains to show that any agent who takes a two-step strategy also becomes qualified.

    Suppose an agent with attributes $\features$ takes a two-step strategy $\strategies(\features)=( \firstfeatures, \secondfeatures)$ in the fixed order mechanism.
    We will show that the final attributes $\secondfeatures$ satisfies $\classifier_A$ and $\classifier_B$. 
 
    \begin{claim}\label{claim: final attributes}
        The final attributes $\secondfeatures$ satisfy $\tilde \classifier_B$ but not $\tilde \classifier_A$. 
    \end{claim}

    \begin{proof}[Proof of \cref{claim: final attributes}]
        Since this agent takes a two-step strategy,
    for this strategy to be a best response, it must be that the agent is accepted by the mechanism.
    Hence we know that $\firstfeatures\in\tilde \classifier_A$ and $\secondfeatures\in \tilde\classifier_B$.
    Suppose $\secondfeatures$ satisfy $\tilde\classifier_A$.
    Then  by the triangle inequality of the cost function, this agent has a lower cost for directly moving to $\secondfeatures$.
    Note that the final attributes $\secondfeatures$ are the attributes that minimize the cost $\onecost(\firstfeatures,\secondfeatures)$ among all attributes satisfying the classifier $\tilde\classifier_B$, and the cost $\onecost(\firstfeatures,\secondfeatures) \leq 1$. 
    Note that the attributes $\firstfeatures$ satisfy the classifier $\tilde \classifier_A$ but do not satisfy $\tilde \classifier_B$. 
    We show that the attributes $\firstfeatures$ are on the boundary of the classifier $\tilde\classifier_A$. 
    Suppose the attributes $\firstfeatures$ are not on the boundary of the classifier $\tilde\classifier_A$. 
    Then, the line segment between $\firstfeatures$ and $\secondfeatures$ intersects the boundary of $\tilde \classifier_A$ at some attributes $\features'$. 
    Since the cost is additive and induced by the metric $d$, we have
    \begin{align*}
        c(\features,\firstfeatures,\secondfeatures) &= \onecost(\features,\firstfeatures) + \onecost(\firstfeatures,\secondfeatures) = \onecost(\features,\firstfeatures) + \onecost(\firstfeatures, \features') + \onecost(\features', \secondfeatures) \\
        &\geq \onecost(\features,\features') + \onecost(\features',\secondfeatures) = c(\features,\features',\secondfeatures).
    \end{align*}
    Thus, the agent will prefer a two-step strategy $(\features,\features',\secondfeatures)$, which gives a contradiction.
    \end{proof}

    \begin{claim}\label{claim: pair of attributes}
        There exists a pair of attributes $\tilde \features$ and $\tilde \features'$ that satisfies the following three properties: 
        \begin{enumerate}
            \item the attributes $\tilde \features$ lie on the boundary of the classifier $\tilde \classifier_A$;
    \item the attributes $\tilde \features'$ are the attributes that minimize the cost $\onecost(\tilde \features, \tilde \features')$ among all attributes satisfying the classifier $\tilde \classifier_B$; 
    \item the cost from $\tilde \feature$ to $\tilde \features'$ is $\onecost(\tilde\features, \tilde \features') = 1$.
        \end{enumerate}
    \end{claim}

\begin{proof}[Proof of \cref{claim: pair of attributes}]
     Consider any attributes $\hat \features$ on the boundary of the classifier $\tilde \classifier_A$ that do not satisfy the classifier $\tilde \classifier_B$. Let the attributes $\hat \features'$ be the attributes that minimize the cost $\onecost(\hat \features, \hat \features')$ among all attributes satisfying the classifier $\tilde \classifier_B$.
    Since the cost is translation invariant and homogeneous, we have $\hat \features' - \hat \features$ is parallel to $\secondfeatures - \firstfeatures$. 
    Let $\weights_B$ be the weight vector of the classifier $\tilde \classifier_B$.
    As the attributes $\hat \features$ move away from the classifier $\tilde \classifier_B$, i.e. $\weights_B \cdot \hat\features$ decreases, the norm of the vector $\hat \features' - \hat\features$ increases. 
    Thus, there exists a pair of such attributes $\tilde \features$ and $\tilde \features'$ that satisfies $\tilde \features' - \tilde \features = (\secondfeatures - \firstfeatures)/ \onecost(\firstfeatures,\secondfeatures)$, which implies $\onecost(\tilde\features, \tilde \features') = 1$.
\end{proof}
   
    \begin{claim}\label{claim: qualified attributes}
         $\tilde \features'$ is qualified, i.e., $\tilde \features'\in\classifier_A \cap \classifier_B$.
    \end{claim}
    \begin{proof}[Proof of \cref{claim: qualified attributes}]
         Consider the agent with attributes $\tilde \features$. Under the informed random order mechanism, this agent prefers the two-step strategy because the cost for directly improving to $\tilde \classifier_A \cap \tilde \classifier_B$ is greater than one. 
    The best strategy for this agent is to first provide the original attributes $\tilde \features$, and if he passes the first test, improve the attributes to $\tilde\features'$ in the second test.
    If the informed random order mechanism uses the classifier $\tilde \classifier_A$ first and then uses the classifier $\tilde \classifier_B$ with probability $q$, then this agent is accepted with probability $q$, and the expected utility is zero. 
    Since the informed random order mechanism is feasible, the final attributes $\tilde \features'$ is in $\classifier_A \cap \classifier_B$. 
    \end{proof}

    According to the construction, the attributes $\secondfeatures$ lie on the line segment connecting $\tilde \features'$ and $\tilde O$.
    Notice that both $\tilde \features'$ and $\tilde O$ are qualified, and the qualified region is convex.
    Hence $\secondfeatures$ are also qualified.  
    
     To summarize, every agent accepted by the fixed order mechanism is qualified.
     Hence the fixed order mechanism is  feasible. 

    Since the two fixed-order mechanisms $(\tilde \classifier_A,\tilde \classifier_B,1)$ and $(\tilde \classifier_A,\tilde \classifier_B,0)$ are feasible, we can apply the same argument as in \cref{lem:feasible-informed-rand-distance cost} by constructing a mixed mechanism to show that one of these two fixed-order mechanism is no worse than the informed random order mechanism. 
\end{proof}
\section{Omitted proof in \cref{sec:Pii}}\label{appendix: alternative objective}

We first state the two lemmas required for the proof and then we prove the two lemmas.
\begin{lemma}\label{lem:opt_sim Pii}
Consider manipulation setting and program \ref{min unqualified}.
     Suppose $\nabla \density \cdot \weights_A \leq 0$, and $\nabla \density \cdot \weights_B \geq 0$.
    Then the optimal simultaneous mechanism $(\classifier_A^S,\classifier_B^S)$ satisfies  $\classifier_A^S\subset\classifier_A$ and $\classifier_B\subset\classifier_B^S$.
\end{lemma}

\begin{lemma}\label{lem:fix>sim Pii}

    Suppose $(\classifier_A^S,\classifier_B^S)$ is a feasible simultaneous mechanism that satisfies $\classifier_i^S\subset\classifier_i$ and $\classifier_j\subset\classifier_j^S$ for $i,j\in \{A,B\}$ and $i\neq j$.
    Then there exists a feasible fixed-order mechanism such that it is weakly better than the simultaneous mechanism $(\classifier_A^S,\classifier_B^S)$ .
\end{lemma}

\begin{proof}[Proof of \cref{prop:Pii manipulation}]
    By \cref{{lem:opt_sim Pii}}, we know that the optimal simultaneous mechanism $(\classifier_A^S,\classifier_B^S)$ satisfies  $\classifier_A^S\subset\classifier_A$ and $\classifier_B\subset\classifier_B^S$.
    By \cref{lem:fix>sim Pii}, there exists a feasible fixed order mechanism that is no worse than the optimal simultaneous mechanism. 
\end{proof}

\begin{proof}[Proof of \cref{lem:opt_sim Pii}]

    It suffices to show that given any feasible simultaneous mechanism $(\tilde\classifier_A,\tilde\classifier_B)$ that violates `$\classifier_A^S\subset\classifier_A$ and $\classifier_B\subset\classifier_B^S$', there exits another feasible simultaneous mechanism  $(\classifier_A^S,\classifier_B^S)$ that (1) satisfies  $\classifier_A^S\subset\classifier_A$ and $\classifier_B\subset\classifier_B^S$; (2) is better than $(\tilde\classifier_A,\tilde\classifier_B)$.

    Suppose the simultaneous mechanism $(\tilde\classifier_A,\tilde\classifier_B)$ is feasible and  it violates `$\tilde\classifier_A\subset\classifier_A$ and $\classifier_B\subset\tilde\classifier_B$'.
    There are three possibilities: (1) $\tilde\classifier_A$ is not a subset of $\classifier_A$ and $\classifier_B$ is a subset of $\tilde\classifier_B$ (Case 1); (2) $\tilde\classifier_A$ is not a subset of $\classifier_A$ and $\classifier_B$ is not a subset of $\tilde\classifier_B$ (Case 2);
    (3) $\classifier_B$ is not a subset of $\tilde\classifier_B$ and and $\tilde\classifier_A$ is a subset of $\classifier_A$ (Case 3).

    \paragraph{Case 1:} Suppose $\tilde\classifier_A$ is not a subset of $\classifier_A$.

    We first show that the boundary of $\tilde\classifier_A$, denoted by $\tilde\Line_A$, must be in parallel with the boundary of $\classifier_A$, denoted by $\Line_A$. 
    This is true because of \cref{claim:Pii optimal sim not parallel case 1 a} and \cref{claim:Pii optimal sim not parallel case 1 b}.
    Hence $\classifier_A$ is a subset of $\tilde\classifier_A$.
    Similarly, we can show that the boundary of $\tilde\classifier_B$ must be in parallel with the boundary of $\classifier_B$.

    \begin{claim}\label{claim:Pii optimal sim not parallel case 1 a}
        Suppose $\tilde\Line_A$ intersects with $\classifier_A$ at $\classifier_A\cap\classifier_B$. Then the simultaneous mechanism $(\tilde\classifier_A,\tilde\classifier_B)$ is not feasible.
    \end{claim}

    \begin{claim}\label{claim:Pii optimal sim not parallel case 1 b}
        Suppose $\tilde\Line_A$ intersects with $\classifier_A$ at $(\classifier_A\cap\classifier_B)^\compl$. Then the simultaneous mechanism $(\tilde\classifier_A,\tilde\classifier_B)$ is not optimal.
    \end{claim}



Denote the intersection of the boundary lines of $\tilde\classifier_i$, $i\in \{A,B\}$ by point $\tilde O$.
Let $r=\min \metric(\tilde O,\classifier_A)$.
Define point $O^S$ as the point that lies on the boundary line of $\classifier_A$ and is obtained by shifting $\tilde O$ along $\weights_A$ as $O^S= \tilde O+r\cdot \frac{\weights_A}{\lVert \weights_A \rVert }$.
Then the vector $v^S = O^S- \tilde O$ is parallel to $\weights_A$.
    
    Next, consider the following simultaneous mechanism: 
    let $\classifier_A^S$ be the stringent test $\classifier_A^S=\classifier_A$, and let another test $\classifier_B^S =\tilde\classifier_B+v^S$.
    See \cref{fig: Pii case 1} for illustrations.
    Then the intersecting point of the boundary lines of $\classifier_i^S$ is $O^S$.
    Moreover, by construction, $\classifier_A^S\subset \classifier_A$ and $\classifier_B\subset\classifier_B^S$.
    Thus, $\classifier_A\cap\classifier_B\subset \classifier_A^S\cap\classifier_B^S$, i.e., the simultaneous mechanism $(\classifier_A^S,\classifier_B^S)$ is feasible.

    It remains to show that the simultaneous mechanism $(\classifier_A^S,\classifier_B^S)$ is better than the simultaneous mechanism $(\tilde\classifier_A,\tilde\classifier_B)$.
    Since feasibility requires that $\classifier_A\cap\classifier_B\subset \classifier_A^S\cap\classifier_B^S$ and $\classifier_A\cap\classifier_B\subset \tilde\classifier_A\cap\tilde\classifier_B$, to compare the probability of selecting an unqualified agent under the two mechanisms, it suffices to compare the probability of selecting an agent under the two mechanisms.

   We first argue that there is a one-to-one mapping between the set of agent selected under simultaneous mechanism $(\classifier_A^S,\classifier_B^S)$ and the set of agent selected under simultaneous mechanism $(\tilde\classifier_A,\tilde\classifier_B)$.
    This is because by construction and translation invariance of the cost function, for any attributes $\features$ selected under the original simultaneous mechanism $(\tilde\classifier_A,\tilde\classifier_B)$, the attributes $\features+v^S$ are also selected under the newly constructed simultaneous mechanism $(\classifier_A^S,\classifier_B^S)$.

    Since the density function $\density$ is weakly decreasing along $v^S$ (or $\weights_A$), the density of each point in the set of agent selected under simultaneous mechanism $(\classifier_A^S,\classifier_B^S)$ is weakly smaller than the corresponding point in the set of agent selected under simultaneous mechanism $(\tilde\classifier_A,\tilde\classifier_B)$.
    Hence the probability of selecting an agent under  $(\classifier_A^S,\classifier_B^S)$ is smaller.



\begin{figure}[t]
\centering
\begin{tikzpicture}[xscale=6,yscale=6]

\draw [domain=0.76:1.36, thick] plot (\x, {3/4*\x+1/4});
\node [left] at (1.32, 1.25 ) {$\classifier_B$};
\draw [thick] (1,0.76) -- (1,1.25);
\node [right] at (1, 1.25 ) {$\classifier_A $};
\node [left] at (1.27,1.2) {$+$};
 \node [right] at (1, 1.18) {$+$};

\draw [thick, blue] (0.8,0.74) -- (0.8,1.38);
\node [right, blue] at (0.8, 1.36 ) {$\tilde\classifier_A$};
\node [right, blue] at (1.2, 1.08 ) {$\tilde\classifier_B$};
\draw [domain=0.76:1.2, thick, blue] plot (\x, {3/4*(\x-0.8)+0.78});
\node [left, blue] at (0.8, 0.78 ) {\footnotesize$\tilde O$};

\node [right, red] at (1, 0.78 ) {\footnotesize$O^S$};
\draw [thick, red] (1,0.7) -- (1,1.38);
\node [above, red] at (1, 1.38 ) {$\classifier_A^S$};

\node [ red, right] at (1.36, 1.05 ) {$\classifier_B^S$};
\draw [domain=0.86:1.36, thick, red] plot (\x, {3/4*(\x-1)+0.78});

\draw [->, thick]  (0.8, 0.78 ) --(1, 0.78 ) ;
\node [below] at (0.9, 0.79 ) {\footnotesize$v^S$};

\end{tikzpicture}
\caption{Case 1 in \cref{lem:opt_sim Pii}}
\label{fig: Pii case 1}
\rule{0in}{1.2em}$^\dag$\scriptsize In this graph, the gray line represents the boundary line of $\tilde\classifier_B+v^S$, which does not cover
$\classifier_B$. Hence we pick $\classifier_B^S=\classifier_B$ so that  the simultaneous mechanism $(\classifier_A^S,\classifier_B^S)$ is feasible and improves upon the simultaneous mechanism $(\tilde\classifier_A,\tilde\classifier_B)$.
\end{figure}

\paragraph{Case 2:}  Suppose $\tilde\classifier_A$ is not a subset of $\classifier_A$.
    and $\classifier_B$ is not a subset of $\tilde\classifier_B$.
    
    The analysis is similar to case 1.
    Denote the intersection of the boundary lines of $\tilde\classifier_i$, $i\in \{A,B\}$ by point $\tilde O$.
Let $r=\min \metric(\tilde O,\classifier_A)$.
Define point $O^S$ as the point that lies on the boundary line of $\classifier_A$ and is obtained by shifting $\tilde O$ along $\weights_A$ as $O^S= \tilde O+r\cdot \frac{\weights_A}{\lVert \weights_A \rVert }$.
Then the vector $v^S = O^S- \tilde O$ is parallel to $\weights_A$.

If $\tilde\classifier_B+v^S$ covers $\classifier_B$, 
consider the following simultaneous mechanism: 
    let $\classifier_A^S$ be the stringent test $\classifier_A^S=\classifier_A$, and let another test $\classifier_B^S =\tilde\classifier_B+v^S$.
    We can use the same argument as in case 1 to show that  the simultaneous mechanism $(\classifier_A^S,\classifier_B^S)$ is feasible and improves upon the simultaneous mechanism $(\tilde\classifier_A,\tilde\classifier_B)$.
    See \cref{fig: Pii case 2} for illustrations.
    
If $\tilde\classifier_B+v^S$ does not cover $\classifier_B$, 
    we consider the following simultaneous mechanism instead: 
    let $\classifier_A^S$ be the stringent test $\classifier_A^S=\classifier_A$, and let another test $\classifier_B^S =\tilde\classifier_B$.

    The simultaneous mechanism $(\classifier_A^S,\classifier_B^S)$ is feasible. This is because $\classifier_A^S\cap\classifier_B^S$  covers $\classifier_A\cap\classifier_B$.

 Consider the vector $v^S =O- \tilde O$.
    Then there is a one-to-one mapping between the set of agent selected under simultaneous mechanism $(\classifier_A^S,\classifier_B^S)$ and the set of agent selected under the simultaneous mechanism $(\tilde\classifier_A,\tilde\classifier_B)$.
    For any attributes $\features$ selected under the original simultaneous mechanism $(\tilde\classifier_A,\tilde\classifier_B)$, the attributes $\features+v^S$ are also selected under the newly constructed simultaneous mechanism $(\classifier_A^S,\classifier_B^S)$.
    
Since the density function $\density$ is weakly decreasing along  $\weights_A$ and $-\weights_B$, 
the density function $\density$ is also weakly decreasing along the vector $v^S$. 
Hence, 
the density of each point in the set of agent selected under simultaneous mechanism $(\classifier_A^S,\classifier_B^S)$ is weakly smaller than the corresponding point in the set of agent selected under simultaneous mechanism $(\tilde\classifier_A,\tilde\classifier_B)$.
    Hence the probability of selecting an agent under  $(\classifier_A^S,\classifier_B^S)$ is smaller, implying that  $(\classifier_A^S,\classifier_B^S)$ improves upon the simultaneous mechanism $(\tilde\classifier_A,\tilde\classifier_B)$.

\begin{figure}[t]
\centering
\begin{tikzpicture}[xscale=6,yscale=6]

\draw [domain=0.76:1.36, thick] plot (\x, {3/4*\x+1/4});
\node [above] at (1.36, 1.24 ) {$\classifier_B$};
\draw [thick] (1,0.78) -- (1,1.25);
\node [right] at (1, 0.8 ) {$\classifier_A \text{ }+$};
\node [left] at (1.27,1.2) {$+$};

\draw [thick, blue] (0.8,0.78) -- (0.8,1.38);
\node [right, blue] at (0.8, 1.36 ) {$\tilde\classifier_A$};
\node [left, blue] at (1.16, 1.36 ) {$\tilde\classifier_B$};
\draw [domain=0.76:1.2, thick, blue] plot (\x, {3/4*(\x-0.8)+1.08});
\node [left, blue] at (0.8, 1.08 ) {\footnotesize$\tilde O$};

\node [right, gray] at (1, 1.08 ) {\footnotesize$O^S$};
\draw [thick, red] (1,1) -- (1,1.38);
\node [above, red] at (1, 1.38 ) {$\classifier_A^S$};

\node [ gray] at (1.35, 1.35 ) {$\tilde\classifier_B+v^S$};
\draw [domain=0.86:1.36, thick, gray] plot (\x, {3/4*(\x-1)+1.08});


\end{tikzpicture}
\caption{Case 2 in \cref{lem:opt_sim Pii}}
\label{fig: Pii case 2}
\rule{0in}{1.2em}$^\dag$\scriptsize In this graph, the gray line represents the boundary line of $\tilde\classifier_B+v^S$, which does not cover
$\classifier_B$. Hence we pick $\classifier_B^S=\classifier_B$ so that  the simultaneous mechanism $(\classifier_A^S,\classifier_B^S)$ is feasible and improves upon the simultaneous mechanism $(\tilde\classifier_A,\tilde\classifier_B)$.
\end{figure}

\paragraph{Case 3:} Suppose $\classifier_B$ is not a subset of $\tilde\classifier_B$
 and  $\tilde\classifier_A$ is a subset of $\classifier_A$.

    Using similar arguments in  \cref{claim:Pii optimal sim not parallel case 1 a} and \cref{claim:Pii optimal sim not parallel case 1 b}, we can show that the boundary of $\tilde\classifier_B$ must be in parallel with the boundary of $\classifier_B$. 

   Let $\weights_A^{\perp}$ be the vector that is orthogonal to $\weights_A$ and satisfies $\weights_A^{\perp} \cdot \weights_B<0$.
   Since $\nabla \density \cdot \weights_A \leq 0$, and $\nabla \density \cdot \weights_B \geq 0$, we must have $\nabla \density \cdot \weights_A^{\perp} \leq 0$, i.e., the density function is weakly decreasing along the direction of $\weights_A^{\perp}$.
  Denote the intersection of the boundary lines of $\tilde\classifier_i$, $i\in \{A,B\}$ by point $\tilde O$.
Let $r=\min \metric(\tilde O,\classifier_B)$.
Define point $O^S$ as the point that lies on the boundary line of $\classifier_B$ and is obtained by shifting $\tilde O$ along $-\weights_B$ as $O^S= \tilde O-r\cdot \frac{\weights_B}{\lVert \weights_B \rVert }$.
Then the vector $v^S = O^S- \tilde O$ is parallel to $-\weights_B$.
    
    Next, consider the following simultaneous mechanism: 
    let $\classifier_A^S$ be the stringent test $\classifier_A^S=\tilde\classifier_A+v^S$, and let another test $\classifier_B^S =\classifier_B$.
    Then the intersecting point of the boundary lines of $\classifier_i^S$ is $O^S$.
    Moreover, by construction, $\classifier_B\subset \classifier_B^S$ and $\classifier_A^S\subset\classifier_A$.
    
    We show that  the simultaneous mechanism $(\classifier_A^S,\classifier_B^S)$ is feasible.
    Using the characterization of the agent's best response in a simultaneous mechanism, it suffices to show that attributes $O$ is accepted by  the simultaneous mechanism $(\classifier_A^S,\classifier_B^S)$.
    Since $(\tilde\classifier_A,\tilde\classifier_B)$ is feasible, then the cost for attributes $O$ to move to some $\features\in \tilde\classifier_A\cap\tilde\classifier_B$ must be less than one, i.e., $\onecost(O,\features)\leq 1$.
    Under the cost function, we know that such attributes $\features=\tilde O$.
    By triangle inequality, $\onecost(O,O^S)\leq \onecost(O,\tilde O) - \onecost(O^S,\tilde O)< 1$.
    Hence attributes $O$ is accepted by  the simultaneous mechanism $(\classifier_A^S,\classifier_B^S)$.
    This implies that the distance between the boundary line of $\classifier_A$ and $\classifier_A^S$ is less than $1/\eta$.
    Hence the simultaneous mechanism $(\classifier_A^S,\classifier_B^S)$ is feasible.

Similarly, there is a one-to-one mapping between the set of agent selected under simultaneous mechanism $(\classifier_A^S,\classifier_B^S)$ and the set of agent selected under the simultaneous mechanism $(\tilde\classifier_A,\tilde\classifier_B)$.
Since the density function $\density$ is weakly decreasing along $v^S$ (or $-\weights_B$), the density of each point in the set of agent selected under simultaneous mechanism $(\classifier_A^S,\classifier_B^S)$ is weakly smaller than the corresponding point in the set of agent selected under simultaneous mechanism $(\tilde\classifier_A,\tilde\classifier_B)$.
    Hence the probability of selecting an agent under  $(\classifier_A^S,\classifier_B^S)$ is smaller, implying that  $(\classifier_A^S,\classifier_B^S)$ improves upon the simultaneous mechanism $(\tilde\classifier_A,\tilde\classifier_B)$.
    
\end{proof}

\begin{proof}[Proof of \cref{lem:fix>sim Pii}]

    Without loss of generality, consider the case $\classifier_B^S\subset\classifier_B$ and $\classifier_A\subset\classifier_A^S$.

    \paragraph{Step 1: constructing stringent test $\classifier_i^{+}, i\in \{A,B\}$.}
    Let $\classifier_i^{+}$ to be the half plane obtained from shifting $\classifier_i$ along the direction of $\weights_i$ by a distance of $1/\eta$.

\paragraph{Step 2: constructing the feasible fixed order mechanism.} Consider the fixed order mechanism that uses  $\classifier_A$  as the first test and $\classifier_B^+$ as the second test.
For any qualified attributes $\features\in \classifier_A\cap\classifier_B$, if they are also in $\features\in \classifier_A\cap\classifier_B^+$, then they are selected in the fixed order mechanism.
If they are not in $\features\in \classifier_A\cap\classifier_B^+$, then it must fall in $\classifier_B^+\cap\classifier_B$.
Since such qualified attributes satisfy $\classifier_A$, they can pass the first test under the fixed order mechanism by not making any changes.
Moreover, by the construction of $\classifier_B^+$, such attributes can change to some attributes that pass $\classifier_B^+$ with cost less than one.
Therefore, this fixed order mechanism selects all qualified attributes.

\paragraph{Step 3: fixed order mechanism $(\classifier_A,\classifier_B^+,1)$ is no worse than the simultaneous mechanism.} 
To show this, we want to show that any unqualified attributes that are not selected by the simultaneous mechanism are not selected by the fixed order mechanism.

\textbf{Case 1.} Consider some  attributes $\features$  that do not satisfy $\classifier_B$. By the construction of $\classifier_B^+$, such attributes do not have a profitable one-step strategy to pass $\classifier_B^+$. We want to show that there does not exist any profitable two-step strategy for such attributes to get selected either. Suppose not and now such attributes find it profitable to first change to some attributes $\firstfeatures$ that satisfy $\classifier_A$ and then other attributes $\secondfeatures$ that satisfy $\classifier_B^+$.
Since such a strategy is profitable we must have $\onecost(\features,\firstfeatures)+\onecost(\firstfeatures,\secondfeatures)\leq1$.
By triangle inequality, we have $\onecost(\features,\secondfeatures)\leq1$, which implies the existence of a profitable one-step strategy for such attributes. A contradiction. Hence attributes that do not satisfy $\classifier_B$ can never be selected by the fixed mechanism.

\textbf{Case 2.} Consider some  attributes $\features$  that do not satisfy $\classifier_A$. Suppose such  attributes $\features$  are not selected by the simultaneous mechanism. This implies that such attributes can only get selected by the fixed order mechanism through some two-step strategy. Suppose such attributes find it profitable to first change to some attributes $\firstfeatures$ that satisfy $\classifier_A$ and then other attributes $\secondfeatures$ that satisfy $\classifier_B^+$. Since it is a profitable strategy, we must have $\onecost(\firstfeatures,\secondfeatures)\leq1$ and hence $\firstfeatures\in \classifier_B$.
Therefore, $\firstfeatures\in \classifier_A\cap \classifier_B$, i.e., $\firstfeatures$ are qualified and selected by the simultaneous mechanism.
Since $\features$  are not selected by the simultaneous mechanism, $\firstfeatures$ must not satisfy the two tests $\classifier_i^S,i\in \{1,2\}$ of the simultaneous mechanism at the same time, i.e., $\firstfeatures\not\in\classifier_A^S\cap \classifier_B^S$. We know that $\firstfeatures$ satisfies $\classifier_A$ and therefore must satisfy $\classifier_A^S$. This implies that $\firstfeatures$ do not satisfy $\classifier_B^S$.
However, since the simultaneous mechanism is feasible, $\firstfeatures$ are selected under the simultaneous mechanism.
This implies that there exists at least one profitable one-step strategy for $\firstfeatures$ to satisfy $\classifier_B^S$, i.e., $\min_{\genericfeatures\in \classifier_B^S}\onecost(\firstfeatures,\genericfeatures)<1$.
Since $\classifier_B^+\subset\classifier_B^S$, we can find the attributes $\features'$ that are the intersection of the line connecting $\features$ and $\secondfeatures$ and the boundary of $\classifier_B^S$.
We then have  $\onecost(\firstfeatures,\features')\leq \min_{\genericfeatures\in \classifier_B^S}\onecost(\firstfeatures,\genericfeatures)<1$.
Since $\classifier_B^+\subset\classifier_B^S$, we have $\onecost(\firstfeatures,\features')< \onecost(\firstfeatures,\secondfeatures)=\min_{\genericfeatures\in \classifier_B^+}\onecost(\firstfeatures,\genericfeatures)$.
By triangle inequality, we have $\onecost(\features,\features')\leq \onecost(\features,\firstfeatures)+\onecost(\firstfeatures,\features')\leq \onecost(\features,\firstfeatures)+\onecost(\firstfeatures,\secondfeatures)\leq1$.
Since $\features'$ satisfy the two tests $\classifier_i^S,i\in\{1,2\}$ under the simultaneous mechanism, this implies that $\features$ are also selected under the simultaneous mechanism. A contradiction.

We conclude that any unqualified attributes that are not selected by the simultaneous mechanism are not selected by the fixed order mechanism. Hence the fixed order mechanism is no worse than the simultaneous mechanism.

\end{proof}

\end{document}